\documentclass[11pt,a4paper]{article}
\usepackage[top=1 in,bottom=1 in,left=1.0 in,right=1.0 in]{geometry}

\usepackage{setspace}
\usepackage{threeparttable}
\usepackage{stmaryrd}
\usepackage{booktabs}
\usepackage{amssymb}
\usepackage{mathrsfs}
\usepackage{latexsym}
\usepackage{epsfig}
\usepackage{subfigure}
\usepackage{graphicx}
\usepackage{epstopdf}
\usepackage{amsthm}
\usepackage{amsfonts}
\usepackage{amsmath}
\usepackage{amssymb,amsmath,amsthm,amscd}
\usepackage{algorithm}
\usepackage{algorithmic}
\usepackage{multirow}
\usepackage{xcolor}
\usepackage{array}
\usepackage{booktabs}
\usepackage{bm}
\usepackage[numbers,sort&compress]{natbib}
\usepackage{lscape}
\usepackage[indentafter]{titlesec}
\usepackage[colorlinks,linkcolor=red,linkcolor=blue,linktocpage]{hyperref}
\usepackage{threeparttable}
\usepackage{caption}
\usepackage{enumitem,calc}
\newcommand\preitem{\mdseries\textbullet\space}
\newlist{desclist}{description}{3}
\setlist[desclist,1]{format=\preitem\bfseries,leftmargin=\widthof{\preitem},style=sameline}
 \makeatletter

\newcommand{\Rmnum}[1]{\expandafter\@slowromancap\romannumeral #1@}

\makeatother \providecommand{\U}[1]{\protect \rule{.1in}{.1in}}
\newtheorem{theorem}{Theorem}
\newtheorem{lemma}{Lemma}
\newtheorem{proposition}{Proposition}
\makeatother \providecommand{\U}[1]{\protect \rule{.1in}{.1in}}
\hypersetup{colorlinks=true, citecolor=blue}

\def \b#1{\overline{#1}}

\def \W#1{\widehat{#1}}
\def \t#1{\widetilde{#1}}

\title{The Fokas-Lenells equations: Bilinear approach}

\author{Shu-zhi Liu, ~ Jing Wang,~  Da-jun Zhang\footnote{Corresponding author. Email: djzhang@staff.shu.edu.cn}\\
{\small  Department of Mathematics,
 Shanghai University, Shanghai 200444,  P.R. China}}

\date{\today}

\begin{document}

\maketitle

\begin{abstract}
In this paper, the Fokas-Lenells equations are investigated via bilinear approach.
We bilinearize the unreduced Fokas-Lenells system, derive double Wronskian solutions,
and then, by means of a reduction technique we obtain variety of solutions of the reduced equations.
This enables us to have a full profile of solutions of the classical and nonlocal Fokas-Lenells equations.
Some obtained solutions are illustrated based on asymptotic analysis.
As a notable new result, we obtain solutions to the Fokas-Lenells equation,
which are related to real discrete eigenvalues and not reported before in the analytic approaches.
These solutions behave like  (multi-)periodic waves or solitary waves with algebraic decay.
In addition, we also obtain solutions to the two-dimensional massive Thirring model
from those of the Fokas-Lenells equation.

\end{abstract}

\section{Introduction}\label{sec-1}

The Fokas-Lenells (FL) equation,
\begin{equation}\label{FL1}
iu_{t}-\nu u_{tx}+\gamma u_{xx}+\sigma |u|^{2}(u+i\nu u_x)=0,\quad \sigma=\pm 1,
\end{equation}
as a novel generalization of the nonlinear Schr\"odinger (NLS) equation,
was first derived using bi-Hamiltonian structures of the NLS equation\cite{Fokas-PD-1995},
where $\nu$ and $\gamma$ are real parameters.
This equation is integrable and belongs to
the derivative nonlinear Schr\"odinger (DNLS) hierarchy \cite{Lenells-F-Non-2009}
that is related to the Kaup-Newell (KN) spectral problem.
It is notable that the FL equation \eqref{FL1} is equivalent to the following one \cite{Lenells-SAPM-2009}
\begin{equation}\label{FL2}
u_{xt}+u-2i\delta|u|^{2}u_{x}=0,~~\delta=\pm 1.
\end{equation}
The latter is derived for modeling   propagation of nonlinear  pulses in monomode optical fibers
where $u$ is assumed to describe the slowly varying envelope of the pulse \cite{Lenells-SAPM-2009}.
It is interesting that Eq.\eqref{FL2} is a reduced potential form of the first negative member in the KN hierarchy,
\begin{subequations} \label{CFL}
\begin{align}
& u_{xt}+u-2iuvu_{x}=0,\label{CFLu}\\
& v_{xt}+v+2iuvv_{x}=0,\label{CFLv}
\end{align}
\end{subequations}
by imposing reduction $v=\delta u^*$, where $i$ is the imaginary unit and $*$ denotes complex conjugate.
Both \eqref{FL1} and \eqref{FL2} can be called the Fokas-Lenells equation,
and in the following we call \eqref{CFL} the pKN$(-1)$ for convenience.
Note that the pKN$(-1)$ system \eqref{CFL} is also known as  the Mikhailov model
(cf.\cite{GIK-1980,GI-1982}).
It is A.V. Mikhailov in 1976 who first gave a Lax pair for the two-dimensional massive Thirring model
in laboratory coordinates.
Later it was shown that the Lax pair in light-cone coordinates is gauge equivalent to
the  KN spectral problem \cite{KN-1977} and solutions of the massive Thirring model
could be obtained by solving Eq.\eqref{FL2}, i.e. the reduced pKN$(-1)$ \cite{KN-1977,GIK-1980}.
We will explain how Eq.\eqref{FL2} and the massive Thirring model are related in Appendix \ref{app-0}.

Before Fokas and Lenells' work, Eq.\eqref{FL2} has been solved using
direct linearization approach \cite{NCQV-1983}.
More recently,
by virtue of a clear integrable background associated with the well studied KN spectral problem,
solutions of the FL equation \eqref{FL1} or \eqref{FL2} have been derived by means of
the Riemann-Hilbert method or inverse scattering transform
\cite{Lenells-F-Non-2009,Ai-Xu-AML-2019, Zhao-F-JNMP-2021},
dressing chain \cite{Lenells-JNS-2010},
algebro-geometric method \cite{Zhao-F-JNMP-2013},
Darboux transformation \cite{He-XP-JPSJ-2012,Xu-HCP-MMAS-2014,Wang-XL-AML-2020},
and a variable separation technique \cite{Wang-QGM-ND-2019}.
Note that in \cite{Wang-XL-AML-2020} Eq.\eqref{FL2} is shown to be related to
the Zakharov-Shabat and  Ablowitz-Kaup-Newell-Segur (ZS-AKNS)  spectral problem.

Without using Lax pairs, the FL equation \eqref{FL2} was bilinearized and
determinantal solutions of the bilinear FL equations were constructed in a  delicate direct way
by introducing determinants of the Cauchy matrix type \cite{Matsuno-JPA-2012a,Matsuno-JPA-2012b}.
Another direct approach was presented in \cite{Vek-Non-2011},
where a chain of B\"acklund transformations of the pKN$(-1)$ \eqref{CFL} was constructed and viewed
as  semi-discrete equations in the Merola-Ragnisco-Tu  hierarchy (cf.\cite{MRT-IP-1994})
and solutions of the FL equation \eqref{FL2} were given in term of Cauchy matrix
by using connection between  the Merola-Ragnisco-Tu  hierarchy and the Ablowitz-Ladik hierarchy.
From the bilinear form given in  \cite{Matsuno-JPA-2012a} it is easy to get 3-soliton solution
in terms of Hirota's polynomials of exponential functions \cite{LiuZLX-Opt-2020}.
However, it is difficult to give double Wronskian solutions to those  FL equations.
As a matter of fact, Eqs.\eqref{CFL} is the potential form KN$(-1)$ (i.e. the first negative member in the KN hierarchy,
see Eq.\eqref{KN-1}).
To our understanding, it is difficult to express the integral $u=\partial^{-1}_x q$ in terms of double Wronskians.

In this paper, we aim to construct double Wronskian solutions for the FL equation \eqref{FL2}.
This will allow us to have more freedom to understand possible distributions of eigenvalues,
 present different kinds of solutions
(e.g. solitons, breathers and multiple pole solutions), explore their interactions
(cf.\cite{Zha-Wro-2019,ZhaZSZ-RMP-2014}),
and as a result, give a full profile of the FL equations
from bilinear approach and double Wronskian forms.
We will start from the pKN$(-1)$ system \eqref{CFL}.
First we will bilinearize \eqref{CFL} and prove their double Wronskian solutions.
Note that the double Wronskians that we employ in the paper have different structures from those of the
AKNS hierarchy (cf.\cite{ChenDLZ-SAPM-2018}), the KN equation (cf.\cite{KakeiSS-JPSJ-1995})
and the Chen-Lee-Liu equation (cf.\cite{ZhaiC-PLA-2008}).
After that we will impose reductions on the double Wronskians using the technique
developed in \cite{ChenDLZ-SAPM-2018}.
This allows us to have solutions not only for the FL equation \eqref{FL2}
but also for its nonlocal partner
\begin{equation}\label{non-equv}
u_{xt}+u-2i\delta uu(-x,-t)u_{x}=0,~~ \delta=\pm 1,
\end{equation}
which is reduced from \eqref{CFL} via a nonlocal reduction $v(x,t)=\delta u(-x,-t)$.
Note that nonlocal integrable systems were  first systematically proposed in 2013
in \cite{AblM-PRL-2013} and have received intensive attention
(e.g.\cite{Zhou-SAPM-2018,AblFLM-SAPM,YY-SAPM-2018,Caud-SAPM-2018,AblM-JPA-2019,YY-LMP-2019,
Lou-SAPM-2019,BioW-SAPM-2019,Lou-CTP-2020,AblLM-Nonl-2020,RaoCPMH-PD-2020,GurPKZ-PLA-2020,
LiFW-SAPM-2020,RybS-JDE-2021,RybS-CMP-2021}).
The reduction also enables us to see how the distribution  of eigenvalues varies with the constraints
imposed in the local and nonlocal reductions.
It is worthy to mention that, apart from those solutions related to complex discrete eigenvalues
(cf.\cite{Lenells-F-Non-2009,Ai-Xu-AML-2019,Zhao-F-JNMP-2021}),
the FL equation \eqref{FL2} allows solutions related to real discrete eigenvalues.
These solutions exhibit (multi-)periodic behaviors,
and also provide solitary waves with algebraic decay as $|x|\to \infty$,
which are not found in the analytic approaches based on spectral analysis
(cf.\cite{Lenells-F-Non-2009,Ai-Xu-AML-2019,Zhao-F-JNMP-2021}).

The paper is organized as follows. As preliminary, in Sec.\ref{sec-2} we recall integrable backgrounds
of the pKN$(-1)$ system \eqref{CFL} and give notations of double Wronskians and some determinantal identities.
Then in Sec.\ref{sec-3} we bilinearize \eqref{CFL}, present double Wronskian solutions,
and implement the reduction technique to get solutions for \eqref{FL2} and \eqref{non-equv}.
Next, dynamics of some obtained solutions are investigated for the FL equation \eqref{FL2} in Sec.\ref{sec-4}
and for the nonlocal FL equation \eqref{non-equv} in Sec.\ref{sec-5}.
Finally, concluding remarks are given in Sec.\ref{sec-6}.
There are three appendices.
The first one introduces known results that how the massive Thirring model,
the KN spectral problem and the FL equation are connected.
The second one presents $N$-soliton solution formula of the pKN$(-1)$ system \eqref{CFL}
via Hirota's expression, and the third one consists of a detailed proof of double Wronskian solutions.

\section{Preliminary}\label{sec-2}

\subsection{Integrability of the FL equations}\label{sec-2-1}

As an integrable background let us recall the relation between the FL equation and the KN hierarchy.
The KN spectral problem reads \cite{KN-1977,KaupN-JMP-1978}
\begin{equation} \label{KN-sp}
\begin{pmatrix} \varphi_1\\ \varphi_2\end{pmatrix}_x
=\begin{pmatrix}\frac{i}{2}\lambda^2 & \lambda q \\ \lambda r &-\frac{i}{2}\lambda^2 \end{pmatrix}
\begin{pmatrix} \varphi_1\\ \varphi_2\end{pmatrix},
\end{equation}
from which one can derivative the well known KN hierarchy
\begin{equation} \label{KN-hie}
\begin{pmatrix} q\\ r\end{pmatrix}_{t_n}
=L^n \begin{pmatrix} -q \\ r \end{pmatrix},
\end{equation}
where the recursion operator $L$ is
\begin{equation*}
L=\partial_x\begin{pmatrix}-i+2q\partial^{-1}_x r  &   2q\partial^{-1}_xq  \\
 2r\partial^{-1}_x r  & i+2r\partial^{-1}_x q  \end{pmatrix},
\end{equation*}
and $\partial_x=\frac{\partial}{\partial x}$, $\partial_x^{-1}\partial_x=\partial_x\partial_x^{-1}=1$.
Here $\lambda$ is the spectral parameter, and both $q$ and $r$ are functions of $(x,t)$.
When $n=2$, the hierarchy \eqref{KN-hie} yields the second order KN system (KN(2) for short)
\begin{subequations} \label{KN2}
\begin{align}
&q_{t}+iq_{xx}-2(q^{2}r)_{x}=0, \\
&r_{t}-ir_{xx}-2(qr^{2})_{x}=0,
\end{align}
\end{subequations}
from which the DNLS equation,
\begin{equation}\label{DNLS}
q_t+iq_{xx}-2\delta(|q|^{2}q)_{x}=0, ~~ \delta=\pm 1,
\end{equation}
is obtained by imposing reduction $r=\delta q^*$.
When $n=-1$, we have KN$(-1)$ equations, i.e.
\begin{equation*}
L\begin{pmatrix} q\\ r\end{pmatrix}_{t }
=  \begin{pmatrix} -q \\ r \end{pmatrix},
\end{equation*}
which reads( with $t\rightarrow -it$)
\begin{subequations} \label{KN-1}
\begin{eqnarray}
  &&q_t+\partial^{-1}q-2iq\partial^{-1}(q\partial^{-1}r+r\partial^{-1}q)=0, \\
  &&r_t+\partial^{-1}r+2ir\partial^{-1}(q\partial^{-1}r+r\partial^{-1}q)=0.
\end{eqnarray}
\end{subequations}
Introduce potentials $(u,v)$ by
\begin{equation}
q=u_x,~~ r=v_x,
\end{equation}
then \eqref{KN-1} can easily be written into the local form \eqref{CFL}.
Thus, the latter is the potential form of the KN$(-1)$ equations \eqref{KN-1}.
As we mentioned before, the pKN$(-1)$ system \eqref{CFL} allows two reductions,
$v=\pm u^*$ and $v(x,t)=\pm u(-x,-t)$.

\subsection{Wronskians and some determinantal identities}\label{sec-2-2}

Let $\phi$ and $\psi$ be $(N+M)$-th order column vectors
\begin{equation}\label{phipsi}
\phi=(\phi_1,\phi_2,\ldots,\phi_{N+M})^T,~~\psi=(\psi_1,\psi_2,\ldots,\psi_{N+M})^T,
\end{equation}
where elements $\phi_j$ and $\psi_j$ are smooth functions of $(x,t)$.
A standard double Wronskian is a determinant of the following form
\[
W^{[N,M]}_{}(\phi; \psi)=|\phi,\partial_{x}\phi,\cdots,\partial_{x}^{N-1}\phi;
\psi,\partial_{x}\psi,\cdots,\partial_{x}^{M-1}\psi|.
\]
We introduce short-hand (cf.\cite{ChenDLZ-SAPM-2018,ChenZ-AML-2018,Deng-AMC-2018})
\[W^{[N,M]}_{}(\phi; \psi)=|\W{\phi^{[N-1]}_{}}; \W{\psi^{[M-1]}_{}}|,
\]
where by $\W{\phi^{[N-1]}_{}}$ we mean
consecutive columns $(\phi,\partial_{x}\phi,\cdots,\partial_{x}^{N-1}\phi)$.
Without making any confusion, we also employ the conventional compact notation
\[W^{[N,M]}_{}(\phi; \psi)=|0,1,\cdots,N-1; 0,1,\cdots,M-1|=|\W{N-1}; \W{M-1}|
\]
that was introduced in \cite{Nimmo-NLS-1983}.

Following the above notations, we introduce four more double Wronskians
\begin{equation}
\begin{array}{l}
|\t N; \W {M-1}|=|1,2,\cdots,N; 0,1,\cdots,M-1|,\\
|\W N; \t{M-1}|=|0,1,\cdots,N; 1,2,\cdots,M-1|,\\
|\b N; \W M|=|2,3,\cdots,N; 0,1,\cdots,M|,\\
|\t N; \t M|=|1,2,\cdots,N; 1,2,\cdots,M|,
\end{array}
\end{equation}
which will be employed in presenting solutions of the bilinear FL-equations in next section.

We also need the following identities which will be used in verifying solutions of bilinear equations.

\begin{lemma}\label{lemma 1}\cite{FreN-PLA-1983}
$$|\mathbf{M}, \mathbf{a},\mathbf{b}||\mathbf{M},\mathbf{c},\mathbf{d}|
-|\mathbf{M}, \mathbf{a},\mathbf{c}||\mathbf{M},\mathbf{b},\mathbf{d}|
+|\mathbf{M}, \mathbf{a},\mathbf{d}||\mathbf{M},\mathbf{b}, \mathbf{c}|=0,$$
where $\mathbf{M}$ is an arbitrary $N\times (N-2)$ matrix, and $\mathbf{a}$, $\mathbf{b}$,
$\mathbf{c}$ and $\mathbf{d}$ are $N$th-order column vectors.
\end{lemma}

\begin{lemma}\label{lemma 2}\cite{ZDJ-arxiv,ZhaZSZ-RMP-2014}
Let $\Xi=(a_{js})_{N\times N}$ be a $N\times N$ matrix with
column vectors $\{\alpha_j\}$. $\Gamma=(\gamma_{js})_{N\times N}$ is an $N\times N$
operator matrix and each $\gamma_{js}$ is a operator. Then we have
\begin{equation*}
\sum^N_{j=1}|\Gamma_{j}\ast \Xi|=\sum^N_{j=1}|(\Gamma^{T})_{j}\ast \Xi^{T}|,
\end{equation*}
where
\begin{equation*}
|A_{j}\ast \Xi|=|\Xi_{1},\ldots,\Xi_{j-1},A\circ \Xi_{j},\Xi_{j+1},\ldots,\Xi_{N}|,
\end{equation*}
and
\begin{equation*}
A_{j}\circ B_{j}=(A_{1,j}B_{1,j},A_{2,j}B_{2,j},\ldots,A_{N,j}B_{N,j}),
\end{equation*}
in which $A_{j}=(A_{1,j},A_{2,j},\ldots,A_{N,j})^{T}$ and $B_{j}=(B_{1,j},B_{2,j},\ldots,B_{N,j})^{T}$
are $N$th-order column vectors.
\end{lemma}

\section{Solutions of the FL equations}\label{sec-3}

In this section,  we  derive  bilinear form of  the pKN$(-1)$ equation \eqref{CFL},
present its double Wronskian solutions,
and apply  reduction technique to obtain solutions of   classical and nonlocal FL equations.

\subsection{Bilinear form and double Wronskian solution}\label{sec-3-1}

With dependent variable transformation
\begin{equation}\label{tran}
u=\frac{g}{f},\quad v=\frac{h}{s},
\end{equation}
the pKN$(-1)$ equation \eqref{CFL} can be bilinearized as the following
\begin{subequations}\label{bilinear-form}
\begin{align}
& D_xD_t\ g\cdot f+gf=0, \label{bilinear-a}\\
& D_xD_t\ h\cdot s+hs=0, \label{bilinear-b}\\
& D_xD_t\ f\cdot s+iD_x\ g\cdot h=0,\label{bilinear-c} \\
& D_t\ f\cdot s+igh=0,\label{bilinear-d}
\end{align}
\end{subequations}
where $D$ is the Hirota bilinear operator defined as \cite{Hirota-1974}
\begin{equation}
 D_x^mD_t^n\ f\cdot g\equiv(\partial_x -\partial_{x^{\prime}})^m (\partial_t-\partial_{t^{\prime}})^n
 f(x,t) g(x^{\prime},t^{\prime})|_{x^{\prime}=x,t^{\prime}=t}.
\end{equation}

$N$-soliton solution in Hirota's form is presented in Appendix \ref{app-A}.
With regard to double Wronskion solutions, we have the following.

\begin{theorem}\label{Theorem 1}
The bilinear equations \eqref{bilinear-form} admit double Wronskian solutions
\begin{equation}
\label{wronskian-1}
f=|\widetilde{N}; \widehat{M-1}|,\quad g=|\widehat{N}; \widetilde{M-1}|,
\quad h=-\frac{i}{2}|\overline{N}; \widehat{M}|,\quad s=|\widetilde{N}; \widetilde{M}|,
\end{equation}
where the elementary column vectors $\phi$ and $\psi$ satisfy
\begin{subequations}\label{wron-cond-x}
\begin{align}
& \phi_x=\frac{i}{2}A^{2}\phi,\quad \phi_t=-\frac{1}{4}\partial^{-1}_x \phi,\label{wron-cond-x-a}\\
& \psi_x=-\frac{i}{2}A^{2}\psi,\quad \psi_t=-\frac{1}{4}\partial^{-1}_x \psi,\label{wron-cond-x-b}
\end{align}
\end{subequations}
Here $A$ is an arbitrary invertible constant matrix in $\mathbb{C}_{(N+M)\times(N+M)}$.
A general form of $\phi$ and $\psi$ obeying \eqref{wron-cond-x} is
\begin{subequations}\label{wronskian}
\begin{align}
& \phi=\exp\Bigl(\frac{i}{2}A^{2}x+\frac{i}{2}A^{-2}t\Bigr) C ,\label{phi}\\
& \psi=\exp\Bigl(-\frac{i}{2}A^{2}x-\frac{i}{2}A^{-2}t\Bigr) B,\label{psi}
\end{align}
\end{subequations}
where $B$ and $C$ are $(N+M)$-th order constant column vectors.
\end{theorem}

The proof is given in Appendix \ref{app-B}.

\begin{proposition}\label{prop-3-1}
$A$ and its any similar form lead to same $u$ and $v$ through
\eqref{tran} and \eqref{wronskian-1}.
\end{proposition}
\begin{proof}
Let $\Lambda=\Gamma A \Gamma^{-1}$ be a matrix similar to $A$ and
$\phi'=\Gamma\phi,~\psi'=\Gamma\psi$.
Then, $(\phi',\psi')$ also satisfy \eqref{wron-cond-x} but with $A$ replaced by $\Lambda$,
and the double Wronskians composed by $(\phi,\psi)$ and by $(\phi',\psi')$ are simply connected by
\[w(\phi',\psi')=|\Gamma|w(\phi,\psi),\]
where $w$ can be $f,g,h$ and $s$.
Thus, the double Wronskians composed by $(\phi',\psi')$ also solve the bilinear equations \eqref{bilinear-form}
and lead to same $u$ and $v$ as before.
\end{proof}

This proposition enables us to have a full profile of solutions for the pKN$(-1)$ equations \eqref{CFL}
by considering canonical forms of $A$.

\subsection{Reductions and solutions }\label{sec-3-2}

In this subsection, we impose suitable constraints on $\phi$ and $\psi$ given in \eqref{wronskian},
so that $u$ and $v$ defined through \eqref{tran} and \eqref{wronskian-1}
can satisfy the relations $v(x,t)=\pm u^*( x,t)$ and $v(x,t)=\pm u(-x,-t)$.
We will also look for explicit forms of $\phi$ and $\psi$ that obey those constraints.
As a result, explicit double Wronskian solutions for the classical FL equation \eqref{FL2}
and nonlocal FL equation \eqref{non-equv} will be obtained.

\subsubsection{Case of the classical FL equation}\label{sec-3-2-1}

We note that, compared with the double Wronskian solutions of the AKNS hierarchy
(cf. Eq.(20) in \cite{ChenDLZ-SAPM-2018}),
the solutions presented in \eqref{wronskian-1} for the pKN$(-1)$ system are more complicated.
To implement a reasonable reduction, let us take $M=N$ and assume
\begin{equation}\label{AS}
A^2=\delta SS^*,~~ \delta=\pm 1,
\end{equation}
where $S$ is an undetermined invertible matrix in $\mathbb{C}_{2N\times 2N}$.
Note that relation \eqref{AS} indicates
\begin{equation}\label{AS1}
A^2S=  S {A^*}^2.
\end{equation}
Then we immediately have from \eqref{wronskian} that
\begin{align*}
S\phi^* & =S \exp\Big[-\frac{i}{2}(A^{*2}x+(A^{*})^{-2}t)\Big]C^* ,\nonumber\\
& =\exp\Big[-\frac{i}{2}(A^{2}x+A^{-2}t)\Big]S C^*=\psi, \label{phipsi*}
\end{align*}
where we have taken $B=SC^*$.
With this relation it then follows that we can write the double Wronskians \eqref{wronskian-1} in terms of
only $\phi$:
\begin{subequations}
\begin{align}
 &f= |\widetilde{N}; \widehat{N-1}|=\Big(\frac{i}{2}\Big)^{N}
 |A^{2}\W{\phi^{[N-1]}_{}}; S \W{{\phi^{[N-1]}_{}}^*}|,\\
 &g=|\widehat{N}; \widetilde{N-1}|=\Big(-\frac{i}{2}\Big)^{N-1}
 |\W{\phi^{[N]}}; A^{2}S \W{{\phi^{[N-2]}_{}}^*}|,\\
 &h=-\frac{i}{2}|\overline{N}; \widehat{N}|
 =-\Big(\frac{i}{2}\Big)^{2N-1}|A^{4}\W{\phi^{[N-2]}_{}}; S\W{{\phi^{[N]}_{}}^*}|,\\
 &s=|\widetilde{N}; \widetilde{N}|
 =\Big(\frac{i}{2}\Big)^{2N}(-1)^N|A|^{2}|\W{\phi^{[N-1]}_{}}; S \W{{\phi^{[N-1]}_{}}^*}|.
\end{align}
\end{subequations}

Then, using \eqref{AS} (i.e. ${A^*}^2=\delta S^* S$), we have
\begin{align*}
f^{*} &= \Big(-\frac{i}{2}\Big)^{N}|{A^*}^{2}\W{{\phi^{[N-1]}_{}}^*}; S^* \W{\phi^{[N-1]}_{}}|,\\
&= \Big(\frac{i}{2}\Big)^{N}|S^* \W{\phi^{[N-1]}_{}}; {A^*}^{2}\W{{\phi^{[N-1]}_{}}^*}|\\
&=\Big (\frac{i}{2}\Big)^{N}|S^*|| \W{\phi^{[N-1]}_{}}; (S^*)^{-1}{A^*}^{2}\W{{\phi^{[N-1]}_{}}^*}|\\
&= \Big(\frac{i}{2}\Big)^{N}|S^*|| \W{\phi^{[N-1]}_{}}; \delta S \W{{\phi^{[N-1]}_{}}^*}|\\
&= (2i)^{N}\delta^{N}|S|^{-1}s,
\end{align*}
and
\begin{align*}
g^{*} & =\Big(\frac{i}{2}\Big)^{N-1}|\W{{\phi^{[N]}_{}}^*}; {A^*}^{2}S^* \W{\phi^{[N-2]}_{}}|\\
&=\Big(-\frac{i}{2}\Big)^{N-1}|{A^*}^{2}S^* \W{\phi^{[N-2]}_{}}; \W{{\phi^{[N]}_{}}^*}|\\
&=\Big(-\frac{i}{2}\Big)^{N-1}|S^*{A}^{2} \W{\phi^{[N-2]}_{}}; \W{{\phi^{[N]}_{}}^*}|\\
&=\Big(-\frac{i}{2}\Big)^{N-1}|S|^{-1}|SS^*{A}^{2} \W{\phi^{[N-2]}_{}}; S\W{{\phi^{[N]}_{}}^*}|\\
&=\Big(-\frac{i}{2}\Big)^{N-1}|S|^{-1}|\delta{A}^{4} \W{\phi^{[N-2]}_{}}; S\W{{\phi^{[N]}_{}}^*}|\\
&=(2i)^{N}\delta^{N-1}|S|^{-1} h.
\end{align*}
This leads to, by noting that $\delta=\pm 1$ and $|A|^2=|S||S|^*$,
$\frac{g^*}{f^*}=\delta \frac{h}{s}$,
i.e. $v(x,t)=\delta u^*(x,t)$.

Let us summarize the above results by the following lemma.

\begin{lemma}\label{Lem-3-2a}
For the Wronskians \eqref{wronskian-1} with \eqref{wronskian},
taking $M=N$ and assuming \eqref{AS} and $B=SC^*$,
where  $S$ is some invertible matrix in $\mathbb{C}_{2N\times 2N}$,
we have the relation
\begin{equation}\label{psi-S-phi}
\psi=S\phi^*
\end{equation}
and
\begin{subequations}
\begin{align}
f^{*} &=(2i)^{N}\delta^{N}|S|^{-1}s,\\
g^{*} & =(2i)^{N}\delta^{N-1}|S|^{-1} h.
\end{align}
\end{subequations}
These give rise to $v(x,t)=\delta u^*(x,t)$
when $u$ and $v$ are defined by \eqref{tran}.
\end{lemma}

In practice we replace $S$ by $S=AT$  where $T\in \mathbb{C}_{2N\times 2N}$, and assume
\begin{equation}\label{TA}
AT=TA^*,~~ TT^*=\delta \mathbf{I}_{2N}^{},
\end{equation}
where $\mathbf{I}_{2N}^{}$ is the identity matrix of $2N$ order.
\eqref{TA} is sufficient condition for \eqref{AS}.
In fact,
\[A^2=\delta A^2 TT^*=\delta A (A T)T^*=\delta AT A^*T^*=\delta S S^*.\]
Thus, we can write the above lemma in terms of $T$.

\begin{theorem}\label{Theorem 21}
The classical FL equation \eqref{FL2} admits the following solution
\begin{equation}\label{DDW}
u(x,t)=\frac{|\widehat{N}; \widetilde{N-1}|}{|\widetilde{N}; \widehat{N-1}|},
\end{equation}
where the elementary vector $\phi$ is given by \eqref{phi} and
\begin{equation}\label{psi-cla-AT}
\psi=AT\phi^*,
\end{equation}
and $A,T\in \mathbb{C}_{2N\times 2N}$ are invertible and satisfy the equation \eqref{TA}.
In addition, the double Wronskians \eqref{wronskian-1} composed by the above $\phi$ and $\psi$
satisfy the following bilinear FL equations,
\begin{subequations}\label{bil-FL}
\begin{align}
& D_xD_t\ g\cdot f+gf=0,  \\
& D_xD_t\ f\cdot f^{*}+i\delta D_x\ g\cdot g^{*}=0,  \\
& D_t\ f\cdot f^{*}+i\delta gg^{*}=0,
\end{align}
\end{subequations}
and the envelope $|u|^2$ can be given by
\begin{equation}\label{u-enve}
|u|^2=i\delta \left( \ln \frac{f}{f^*}\right)_t
=2 \delta \left( \arctan \frac{\mathrm{Re}[f]}{\mathrm{Im}[f]}\right)_t.
\end{equation}
\end{theorem}

Next, we give explicit expression of $\phi$. To achieve that, we assume both $T$ and $A$ are 2 by 2 block matrices
\begin{equation}\label{real-TA}
T=\left( \begin{array}{cc} T_1 & T_2 \\T_3 & T_4 \\ \end{array}\right),~~~
A=\left(\begin{array}{cc} K_1 & \mathbf{0}_N \\ \mathbf{0}_N & K_4 \\\end{array} \right),
\end{equation}
where $T_i$ and $K_i$ are $N\times N$ matrices.
Equation \eqref{TA} allows the following solutions (cf.\cite{ChenDLZ-SAPM-2018}) as we list in Table \ref{tab-1},
where $|\mathbf{K}_N||\mathbf{H}_{N}|\neq 0$.
\begin{table}[h]
\begin{center}
\begin{tabular}{cllll}
  \hline
 case &  $\delta$ & T & A \\
  \hline
(1)&  $\pm 1$ & $T_{1}=T_{4}=\mathbf{0}_{N}, T_{2}=\delta T_{3}=\mathbf{I}_{N}$
   & $K_{1}=K_{4}^{*}=\mathbf{K}_{N}\in \mathbb{C}_{N\times N}$\\
(2)&   $1$  &$T_{1}=\pm T_{4}= \mathbf{I}_{N}, T_{2}=T_{3}=\mathbf{0}_{N}$
   & $K_{1}=\mathbf{K}_N\in \mathbb{R}_{N\times N}, K_{4}=\mathbf{H}_{N}\in \mathbb{R}_{N\times N}$\\
  \hline
\end{tabular}
\caption{$T$ and $A$ for \eqref{TA}}
\label{tab-1}
\end{center}
\end{table}

Let us introduce
\begin{equation}\label{phi-pm}
\phi=\left( \begin{array}{c} \phi^+   \\ \phi^{-}\end{array}\right),
\end{equation}
where $\phi^{\pm}=(\phi^{\pm}_1, \phi^{\pm}_2, \cdots, \phi^{\pm}_N)^T$.
When $A$ takes the form in \eqref{real-TA}, the elementary vector $\phi$ given in \eqref{phi}
can be written as in \eqref{phi-pm}, where
\begin{equation}\label{phi-pmm}
\phi^+=\exp\Big[\frac{i}{2}(K^{2}_1 x+K^{-2}_1 t)\Big]C^+,~~
\phi^-=\exp\Big[\frac{i}{2}(K^{2}_4 x+K^{-2}_4 t)\Big]C^-,
\end{equation}
and $C^{\pm}=(c^{\pm}_1, c^{\pm}_2, \cdots, c^{\pm}_N)^T$.
Vector $\psi$ is defined by \eqref{psi-cla-AT}.

\vskip 6pt
\noindent
\textbf{Solutions corresponding to Case (1) in Table \ref{tab-1}:}
More explicitly, when $\mathbf{K}_N$ is a diagonal matrix
\begin{equation}\label{A-diag}
\mathbf{K}_{N}=\mathrm{Diag}(k_{1},k_{2},\cdots,k_{N}),~~ k_j\in \mathbb{C},
\end{equation}
we have
\begin{subequations}
\begin{align}
\phi^+ & =  (c^+_{1}e^{\eta(k_{1})}, c^+_{2}e^{\eta(k_{2})},\cdots,  c^+_{N}e^{\eta(k_{N})} )^{T},\\
\phi^- & =  (c^{-}_{1}e^{\eta(k_{1}^*)}, c^{-}_{2}e^{\eta(k_{2}^*)},\cdots,
c^{-}_{N}e^{\eta(k_{N}^*)} )^{T},
\end{align}
\end{subequations}
where
\begin{equation}\label{eta}
\eta(k )=\frac{i}{2}(k^{2}x+ k^{-2}t).
\end{equation}
When $\mathbf{K}_{N}$ is a  Jordan block matrix $J_{N}(k)$,
\begin{equation}\label{A-jordan}
J_{N}(k)=\left(
  \begin{array}{cccc}
    k & 0 & \cdots & 0 \\
    1& k & \cdots & 0 \\
    \vdots & \ddots & \ddots & \vdots \\
    0 & \cdots& 1 & k \\
  \end{array}
\right)_{N\times N},
\end{equation}
we have
\begin{subequations}
\begin{align}
\phi^+ &=\mathcal{A}_N\Bigl(c^+ e^{\eta(k)}, \partial_k(c^+ e^{\eta(k)}),
\frac{1}{2!}\partial_k^{2}(c^+ e^{\eta(k)}),
\cdots,  \frac{1}{(N-1)!}\partial_k^{N-1}(c^+ e^{\eta(k)}) \Bigr)^{T},\\
\phi^- & = \mathcal{B}_N \Bigl(c^{-} e^{\eta(k^*)}, \partial_{k^*}(c^{-} e^{\eta(k^*)}),
\frac{1}{2!}\partial_{k^*}^{2}(c^- e^{\eta(k^*)}),\cdots,
\frac{1}{(N-1)!}\partial_{k^*}^{N-1}(c^- e^{\eta(k^*)}) \Bigr)^{T},
\end{align}
\end{subequations}
where
$\partial_k=\frac{\partial}{\partial_k}$, $k, c^{\pm}\in \mathbb{C}$,
$\mathcal{A}_N$ and $\mathcal{B}_N$ belong to
an Abelian group $G_N$, which is composed by
all invertible lower triangular Toeplitz matrices (LTTMs) of the following form
\begin{equation}\label{A-jordan}
\mathcal{G}_{N} =\left(
  \begin{array}{ccccc}
   g_0 & 0 & 0& \cdots & 0 \\
    g_1& g_0 & 0 & \cdots & 0 \\
    g_2& g_1 & g_0 & \cdots & 0 \\
    \vdots & \vdots & \vdots &\ddots & \vdots \\
   g_{N-1}^{}& g_{N-2}^{} & g_{N-3}^{} & \cdots & g_0 \\
  \end{array}
\right)_{N\times N},~~ g_i\in \mathbb{C},~ g_0\neq 0.
\end{equation}
Note that the LTTMs have been widely used in presenting multiple pole solutions
(cf.\cite{ZhaZSZ-RMP-2014,Zha-Wro-2019,ZDJ-arxiv}).
From \eqref{psi-cla-AT} one can find that $\psi$ always takes a form
\begin{equation}\label{psi-phi-pm}
\psi=\left( \begin{array}{r} \mathbf{K}_N^{}\,\phi^{-*}   \\
\delta \mathbf{K}_N^*\,\phi^{+*}\end{array}\right).
\end{equation}

\vskip 6pt
\noindent
\textbf{Solutions corresponding to Case (2) in Table \ref{tab-1}:}
In this case, both $\mathbf{K}_N^{}$ and $\mathbf{H}_N^{}$ are real.
$\phi^+$ in \eqref{phi-pm} is governed by $\mathbf{K}_N^{}$.
When $\mathbf{K}_N$ is  diagonal, i.e.
\begin{equation}\label{A-diag-real}
\mathbf{K}_{N}=D[k_j]_{j=1}^{N}=\mathrm{Diag}(k_{1},k_{2},\cdots,k_{N}),~~ k_j\in \mathbb{R},
\end{equation}
one has
\begin{align}\label{phi-m-KD}
\phi^+  =  (c^+_{1}e^{\eta(k_{1})}, c^+_{2}e^{\eta(k_{2})},\cdots,  c^+_{N}e^{\eta(k_{N})} )^{T},
\end{align}
where $\eta$ is defined by \eqref{eta}, and we note that $c^+_{j}\in \mathbb{C}$.
When $\mathbf{K}_{N}$ is a  Jordan block matrix $\mathbf{K}_{N}=J_{N}(k)$ as given in \eqref{A-jordan}
where $k\in \mathbb{R}$, one has
\begin{align}\label{phi-m-KJ}
\phi^+ =\mathcal{A}_N\Bigl(c^+ e^{\eta(k)}, \partial_k(c^+ e^{\eta(k)}),
\frac{1}{2!}\partial_k^{2}(c^+ e^{\eta(k)}),
\cdots,  \frac{1}{(N-1)!}\partial_k^{N-1}(c^+ e^{\eta(k)}) \Bigr)^{T},
\end{align}
where $\mathcal{A}_N$ is a real element in $G_N$ but $c^+$ is complex.
$\phi^-$ in \eqref{phi-pm} is determined by $\mathbf{H}_N^{}$.
When $\mathbf{H}_N$ is a diagonal  matrix
\begin{equation}\label{A-diag-real-H}
\mathbf{H}_{N}=D[h_j]_{j=1}^{N}=\mathrm{Diag}(h_{1}, h_{2},\cdots, h_{N}),~~ h_j\in \mathbb{R},
\end{equation}
one has
\begin{align}\label{phi-m-HD}
\phi^-  =  (c^-_{1}e^{\eta(h_{1})}, c^-_{2}e^{\eta(h_{2})},\cdots,  c^-_{N}e^{\eta(h_{N})} )^{T},
\end{align}
where $\eta$ is defined by \eqref{eta} and $c^-_{j}\in \mathbb{C}$.
In Jordan block case when $\mathbf{H}_{N}=J_{N}(h)$ as given in \eqref{A-jordan}
where $h\in \mathbb{R}$, one has
\begin{align}\label{phi-m-HJ}
\phi^- =\mathcal{B}_N\Bigl(c^- e^{\eta(h)}, \partial_h(c^- e^{\eta(h)}),
\frac{1}{2!}\partial_h^{2}(c^- e^{\eta(h)}),
\cdots,  \frac{1}{(N-1)!}\partial_h^{N-1}(c^- e^{\eta(h)}) \Bigr)^{T},
\end{align}
where
$\mathcal{B}_N$ is a real element in $G_N$ but  $c^-$ is complex.
In this case, $\phi$ takes the form \eqref{phi-pm} where
$\phi^{+}$ can be either \eqref{phi-m-KD} or \eqref{phi-m-KJ}
and $\phi^{-}$ can be either \eqref{phi-m-HD} or \eqref{phi-m-HJ},
and consequently $\psi$  takes a form
\begin{equation}\label{psi-phi-pmm}
\psi=\left( \begin{array}{r} \mathbf{K}_N^{}\,\phi^{+*}   \\ - \mathbf{H}_N\,\phi^{-*}\end{array}\right).
\end{equation}

Note that, for the above both cases,  since \eqref{wron-cond-x} is a linear system w.r.t. $\phi$ and $\psi$,
both $\mathbf{K}_N^{}$ and $\mathbf{H}_N^{}$ can take  block diagonal forms,
e.g.,  for Case (2),
\begin{align*}
\mathbf{K}_N^{}&=\mathrm{Diag}\Bigl(D[k_j]_{j=1}^{N_0}, J_{N_1}(k^{}_{N_0^{}+1}),
\cdots, J_{N_s}(k^{}_{N_0^{}+s}) \Bigr),\\
\mathbf{H}_N^{}&=\mathrm{Diag}\Bigl(D[h_j]_{j=1}^{M_0}, J_{M_1}(h^{}_{M_0^{}+1}),
\cdots, J_{M_z}(h^{}_{M_0^{}+z}) \Bigr),
\end{align*}
where $\sum N_j=\sum M_i=N$,
and explicit forms of $\phi^{\pm}$ and $\phi$ can be easily written out accordingly.

We also note that, since Eq.\eqref{TA} is bilinear  w.r.t both $A$ and $T$, from Table \ref{tab-1},
when $\delta=1$, one can combine the above cases and get mixed solutions.
In details, when $\delta=1$, Eq.\eqref{TA} allows a more general solution
\begin{equation}\label{TA-mix}
 T=\left(
     \begin{array}{cccc}
      \mathbf{I}_{N_1} & \mathbf{0}_{N_1} & & \\
      \mathbf{0}_{N_1}   & -\mathbf{I}_{N_1} & & \\
      & & \mathbf{0}_{N_2}   & \mathbf{I}_{N_2} \\
      & & \mathbf{I}_{N_2} & \mathbf{0}_{N_2}
     \end{array}
   \right),~~
 A=\left(
     \begin{array}{cccc}
      \mathbf{K}'_{N_1} & \mathbf{0}_{N_1} & & \\
      \mathbf{0}_{N_1}   & \mathbf{H}'_{N_1} & & \\
      & & \mathbf{K}_{N_2}   & \mathbf{0}_{N_2} \\
      & & \mathbf{0}_{N_2} & \mathbf{K}^*_{N_2}
     \end{array}
   \right),
\end{equation}
where $\mathbf{K}'_{N_1}, \mathbf{H}'_{N_1} \in \mathbb{R}_{N_1\times N_1}$,
 $\mathbf{K}_{N_2}  \in \mathbb{C}_{N_2\times N_2},~N_1+N_2=N$.
Obviously,  explicit expression for $\phi$ of this case can be easily composed accordingly
by referring to the above two cases.

Dynamics of some obtained solutions will be investigated in Sec.\ref{sec-4}.

\subsubsection{Case of the nonlocal FL equation}\label{sec-3-2-2}

The nonlocal relation
\begin{equation}
v(-x,-t)=\delta u(x,t), ~~ \delta=\pm 1
\label{v-u}
\end{equation}
reduces the pKN$(-1)$ system \eqref{CFL} to a one-field equation, the nonlocal FL equation \eqref{non-equv}.
In the following, from \eqref{tran} and \eqref{wronskian-1},
we recover the above nonlocal relation and get solutions to the nonlocal FL equation \eqref{non-equv}.

Let us consider  $M=N$ and impose constraint on \eqref{wronskian} by
\begin{equation}\label{phi-psi-non}
\psi(x,t)=S \phi(-x,-t).
\end{equation}
This holds if
\begin{equation}\label{AS2}
A^2=\delta S^2,
\end{equation}
and $B=SC$.
Note that \eqref{AS2} indicates
$A^2S=  S {A}^2$.
Next, for convenience we introduce a notation (cf.\cite{ChenDLZ-SAPM-2018,ChenZ-AML-2018})
\begin{equation}\label{phi-abc}
\W{\phi^{[N]}_{}}(ax,bt)_{[cx]}=\left(\phi(ax,bt), \partial_{cx}\phi(ax,bt),
 \partial_{cx}^{2}\phi(ax,bt),\cdots, \partial_{cx}^{N}\phi(ax,bt)\right).
\end{equation}
Thus, the double Wronskians \eqref{wronskian-1} with the constraint \eqref{phi-psi-non}
are written as
\begin{subequations}
\begin{align}
 &f(x,t)= |\widetilde{N}; \widehat{N-1}|=\Big(\frac{i}{2}\Big)^{N}
 |A^{2}\W{\phi^{[N-1]}}(x,t)_{[x]}; S \W{\phi^{[N-1]}_{}}(-x,-t)_{[x]}|,\\
 &g(x,t)=|\widehat{N}; \widetilde{N-1}|=\Big(-\frac{i}{2}\Big)^{N-1}
 |\W{\phi^{[N]}}(x,t)_{[x]}; A^{2}S \W{\phi^{[N-2]}_{}}(-x,-t)_{[x]}|,\\
 &h(x,t)=-\frac{i}{2}|\overline{N}; \widehat{N}|
 =-\Big(\frac{i}{2}\Big)^{2N-1}|A^{4}\W{\phi^{[N-2]}_{}}(x,t)_{[x]}; S\W{\phi^{[N]}_{}}(-x,-t)_{[x]}|,\\
 &s(x,t)=|\widetilde{N}; \widetilde{N}|
 =\Big(\frac{i}{2}\Big)^{2N}(-1)^N|A|^{2}|\W{\phi^{[N-1]}_{}}(x,t)_{[x]}; S \W{\phi^{[N-1]}_{}}(-x,-t)_{[x]}|.
\end{align}
\end{subequations}
Then we find that
\begin{align*}
f(-x,-t) &=\Big(\frac{i}{2}\Big)^{N} |A^{2}\W{\phi^{[N-1]}}(-x,-t)_{[-x]}; S \W{\phi^{[N-1]}_{}}(x,t)_{[-x]}|\\
&=\Big(\frac{i}{2}\Big)^{N} |A^{2}\W{\phi^{[N-1]}}(-x,-t)_{[x]}; S \W{\phi^{[N-1]}_{}}(x,t)_{[x]}|\\
&=\Big(\frac{i}{2}\Big)^{N}|S \W{\phi^{[N-1]}_{}}(x,t)_{[x]}; A^{2}\W{\phi^{[N-1]}}(-x,-t)_{[x]}|\\
&=\Big(\frac{i}{2}\Big)^{N}|S| | \W{\phi^{[N-1]}_{}}(x,t)_{[x]}; S^{-1}A^{2}\W{\phi^{[N-1]}}(-x,-t)_{[x]}|\\
&=\Big(\frac{i}{2}\Big)^{N}|S| | \W{\phi^{[N-1]}_{}}(x,t)_{[x]}; \delta S \W{\phi^{[N-1]}}(-x,-t)_{[x]}|\\
&=(-2i)^{N}\delta^{N}|S|^{-1}s(x,t),
\end{align*}
and in a similar way,
\begin{align*}
g(-x,-t) =(-2i)^{N}\delta^{N-1}|S|^{-1} h(x,t).
\end{align*}
These results immediately give rise to $\frac{g(-x,-t)}{f(-x,-t)}=\delta \frac{h(x,t)}{s(x,t)}$, i.e. the relation \eqref{v-u}.

Introduce $S=AT$  and to keep \eqref{AS2} we assume that
\begin{equation}\label{TA2}
AT=TA,~~ T^2=\delta \mathbf{I}_{2N}^{}.
\end{equation}
Then, the solutions of the nonlocal case are summarized as the following.

\begin{theorem}\label{Theorem 22}
The nonlocal FL equation \eqref{non-equv} admits  solutions
\begin{equation}\label{DDW-non}
u(x,t)=\frac{|\widehat{N}; \widetilde{N-1}|}{|\widetilde{N}; \widehat{N-1}|},
\end{equation}
where the elementary vector $\phi$ is given by \eqref{phi} and
\begin{equation}\label{psi-cla}
\psi(x,t)=AT\phi(-x,-t),
\end{equation}
and $A,T\in \mathbb{C}_{2N\times 2N}$ are invertible and satisfy the equation \eqref{TA2}.
The double Wronskians \eqref{wronskian-1} composed by the above $\phi$ and $\psi$
satisfy the following bilinear nonlocal FL equations:
\begin{align*}
& D_xD_t\ g(x,t)\cdot f(x,t)+g(x,t)f(x,t)=0,  \\
& D_xD_t\ f(x,t)\cdot f(-x,-t)+i\delta D_x\ g(x,t)\cdot g(-x,-t)=0,  \\
& D_t\ f(x,t)\cdot f(-x,-t)+i\delta g(x,t)g(-x,-t)=0.
\end{align*}
A special solution to \eqref{TA2} is given by block  matrices  form \eqref{real-TA} with
\begin{equation}\label{TA3}
 T_{1}=-T_{4}=\sqrt{\delta}\,\mathbf{I}_{N}, ~T_{2}=T_{3}=\mathbf{0}_{N},~
 K_{1}=\mathbf{K}_N\in \mathbb{C}_{N\times N},~ K_{4}=\mathbf{H}_{N}\in \mathbb{C}_{N\times N}.
\end{equation}
\end{theorem}

Explicit expression of $\phi$ is given through the form \eqref{phi-pm} where $\phi^{\pm}$
are  given by those formulaes from \eqref{A-diag-real} to \eqref{phi-m-HJ} but
at this stage $k,k_j, h, k_j \in \mathbb{C}$, and $\psi$  is
\begin{equation}\label{psi-phi-pmm-non}
\psi=\left( \begin{array}{r}\sqrt{\delta} \mathbf{K}_N^{}\,\phi^+(-x,-t)   \\
-\sqrt{\delta} \mathbf{H}_N\,\phi^{-}(-x,-t)\end{array}\right).
\end{equation}

\section{Dynamics of the classical FL equation \eqref{FL2}}\label{sec-4}

In this section we analyze dynamics of solutions of the FL equation \eqref{FL2},
which we obtained in the previous section.
We will investigate solutions related to discrete complex eigenvalues
and also discrete real eigenvalues, i.e. Case (1) and Case (2) in Table \ref{tab-1}.
In the first case,  one-solition feature and two-soliton interactions were already considered in \cite{Matsuno-JPA-2012a},
so we will focus more on breathers and double pole solutions.
The second case contributes solutions related to real discrete eigenvalues,
which, to our knowledge, were not reported in the past literatures.
These solutions allow periodic and double-periodic waves,
and quite interestingly, solitary waves with algebraic decays as $|x|\to \infty$.

\subsection{Solutions related to complex eigenvalues}\label{sec-4-1}

Let us consider Case (1) in Table 1 where we take $\delta=1$.
Note that in this case, when $\mathbf{K}_{N}$ is diagonal one will obtain
the usual $N$-soliton solutions, which coincide with those results that have been obtained from
the inverse scattering transform (or Riemann-Hilbert method)
\cite{Lenells-F-Non-2009,Lenells-SAPM-2009,Ai-Xu-AML-2019,Zhao-F-JNMP-2021},
dressing method \cite{Lenells-JNS-2010},
Darboux transformation \cite{Xu-HCP-MMAS-2014,Wang-XL-AML-2020}
with zero as a seed solution
and bilinear method by Matsuno \cite{Matsuno-JPA-2012a},
including the solutions for the FL equation \eqref{FL1} in light of the transformation
that converts Eq.\eqref{FL1} to Eq.\eqref{FL2} (see proposition 1 in \cite{Lenells-SAPM-2009}).
Apart from the diagonal $\mathbf{K}_{N}$, when $\mathbf{K}_{N}$ is a Jordan matrix
or contains Jordan blocks, the solution, in principle,
can be obtained by a limit procedure from soliton solutions
(e.g. \cite{MatS-Book-1991,ZDJ-arxiv}).

\subsubsection{1SS}\label{sec-4-1-1}

When $\mathbf{K}_{N}$ is given in \eqref{A-diag} with $N=1$,
we get one-soliton solution (1SS)
\begin{equation}\label{1ss-fg}
u=\frac{g}{f},~f=|\phi_x;\psi|,~ g=|\phi,\phi_x|,
\end{equation}
which reads
\begin{equation}\label{1ss}
u=\frac{c_{1}d_{1}(k_{1}^{2}-k_{1}^{*2})}
{|k_{1}|^{2}\left[k_{1}^{*}|d_{1}|^{2}\, \mathrm{e}^{-i(k_{1}^{2}x +\frac{t}{k_{1}^{2}})}
-k_{1}|c_{1}|^{2}\, \mathrm{e}^{-i(k_{1}^{*2}x+\frac{t}{k_{1}^{*2}})}\right]},
\end{equation}
where $c_1=c^+_1, ~d_1=c^-_1$.
The carrier wave is expressed as
\begin{equation}\label{1-sol-FL-1}
|u|^2= \frac{8a_{1}^{2}b_{1}^{2}}{(a_{1}^{2}+b_{1}^{2})^{3}}
\frac{1}{\cosh\left(4a_{1}b_{1}x-\frac{4a_{1}b_{1}t}{(a_{1}^{2}+b_{1}^{2})^{2}}
+2\ln\frac{|d_{1}|}{|c_{1}|}\right)
-\frac{a_{1}^{2}-b_{1}^{2}}{a_{1}^{2}+b_{1}^{2}}},
\end{equation}
where we have taken $k_{j}=a_{j}+ib_{j}$, $a_j,b_j\in \mathbb{R}$.
\eqref{1-sol-FL-1} describes a single direction soliton traveling with amplitude
$\frac{2|a_{1}|}{a_{1}^{2}+b_{1}^{2}}$,
initial phase $2\ln\frac{|d_{1}|}{|c_{1}|}$,
velocity  $\frac{1}{(a_{1}^{2}+b_{1}^{2})^{2}}$,
and trajectory (top trace)
\begin{equation}\label{top trace-FL}
x(t)=\frac{1}{(a_{1}^{2}+b_{1}^{2})^{2}}t-\frac{1}{2a_{1}b_{1}}\ln\frac{|d_{1}|}{|c_{1}|}.
\end{equation}
Obviously, $a_1b_1$ should not be zero, which means $k_1$ cannot be real or pure imaginary.
This coincides with the assumption on the distribution of eigenvalues from scattering analysis
(cf.\cite{Lenells-F-Non-2009,Ai-Xu-AML-2019}).
\eqref{1-sol-FL-1} is depicted in Fig.\ref{F-1}(a).
\captionsetup[figure]{labelfont={bf},name={Fig.},labelsep=period}
\begin{figure}[ht]
\centering
\subfigure[ ]{
\begin{minipage}[t]{0.45\linewidth}
\centering
\includegraphics[width=2.1in]{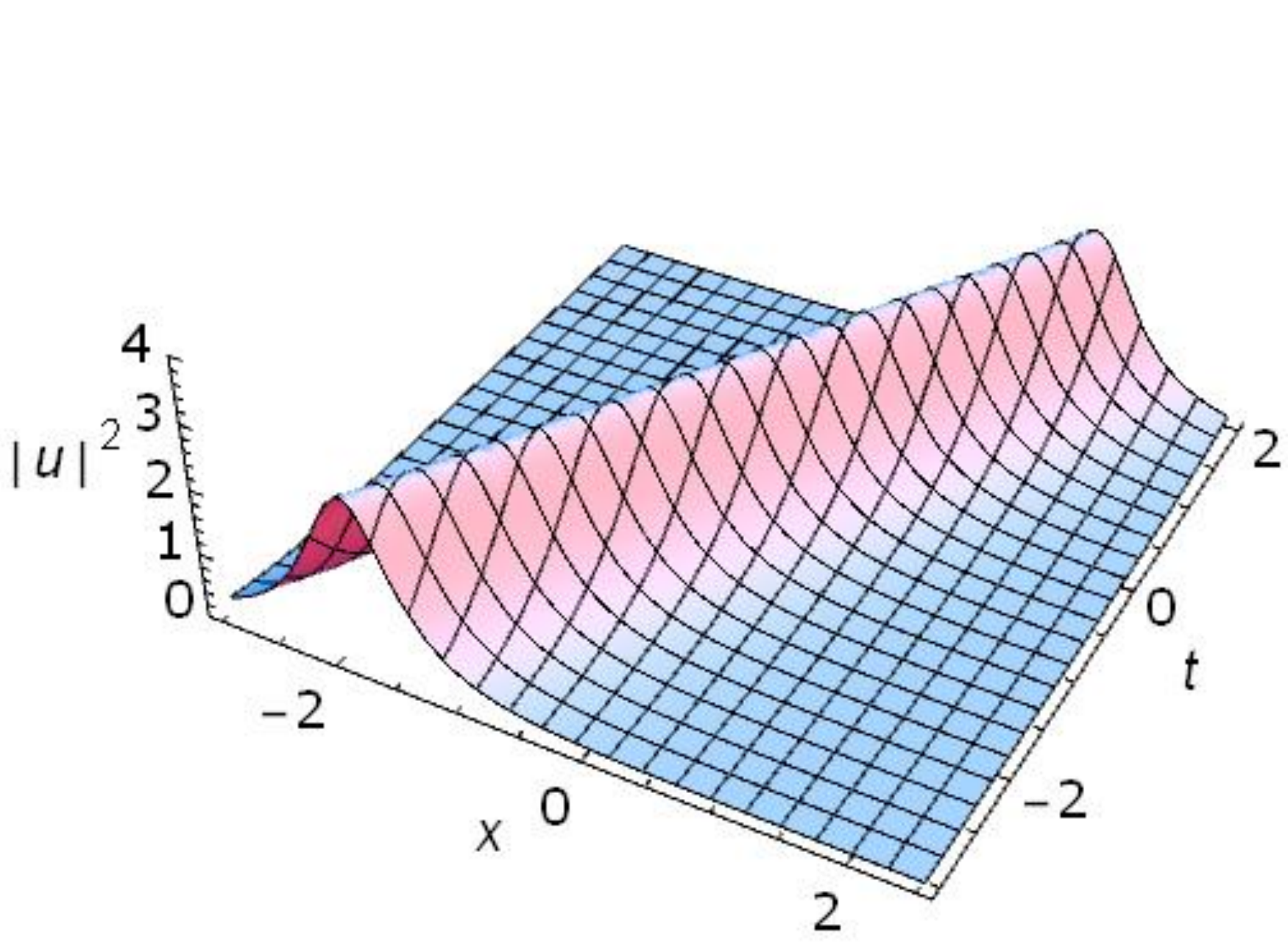}
\end{minipage}%
}%
\subfigure[ ]{
\begin{minipage}[t]{0.45\linewidth}
\centering
\includegraphics[width=2.2in]{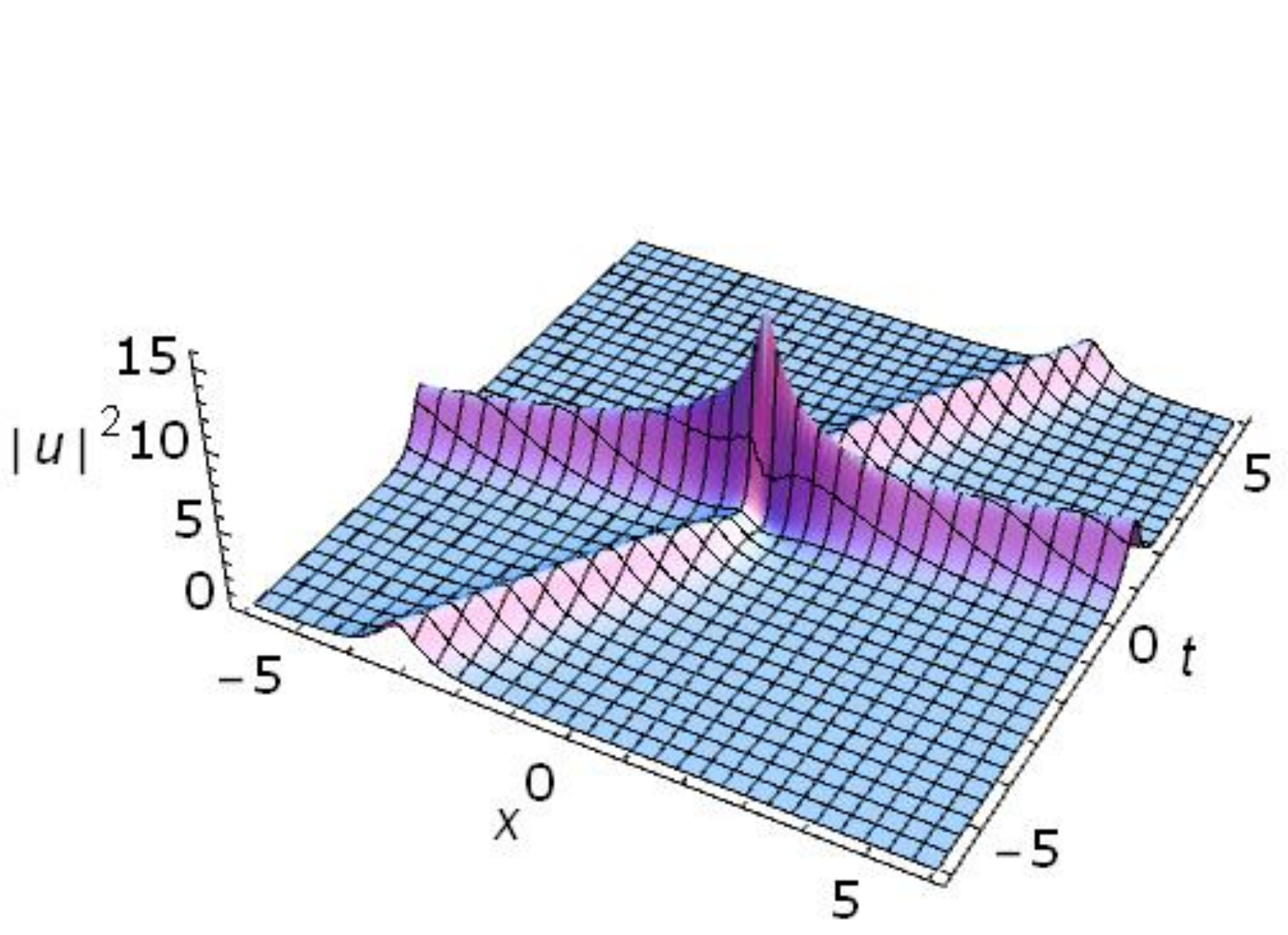}
\end{minipage}%
}
\caption{Shape and motion of 1SS and 2SS of the   FL equation \eqref{FL2}.
(a) 1SS given by \eqref{1-sol-FL-1} for $k_1=1+0.5i$, $c_1=d_1=1$.~
(b) 2SS $|u|^2$ where $u$ is given in \eqref{2ss} with  $k_1=1+0.5i$, $k_2=0.5+0.5i$,
$c_1=d_1=c_2=d_2=1$.~}
\label{F-1}
\end{figure}

\subsubsection{2SS}\label{sec-4-1-2}

Two-soliton solution (2SS) is obtained when $\mathbf{K}_{N}$ is given in \eqref{A-diag} with $N=2$.
It can be expressed as
\begin{subequations}\label{2ss-pp}
\begin{equation}
u^{}_{2\mathrm{SS}}=\frac{g}{f},
\end{equation}
with
\begin{equation}
f=|\partial_x\phi, \partial^2_x\phi; \psi, \partial_x\psi|,~~
g=|\phi, \partial_x\phi, \partial^2_x\phi;\partial_x \psi|,
\end{equation}
\end{subequations}
where
\begin{align*}
& \phi=(c_1e^{\eta(k_1)}, c_2e^{\eta(k_2)}, d_1e^{\eta(k^*_1)}, d_2e^{\eta(k^*_2)})^T,\\
& \psi=(k_1d_{1}^{*}e^{-\eta(k_1)}, k_2d_{2}^{*}e^{-\eta(k_2)},
k_1^*c_{1}^{*}e^{-\eta(k^*_1)}, k_2^* c_{2}^{*}e^{-\eta(k^*_2)})^T,
\end{align*}
$\eta$ is defined by \eqref{eta}, $k_j, c_j,d_j\in \mathbb{C}$.
2SS has been investigated in \cite{Matsuno-JPA-2012a} where the solution is expressed
in terms of determinants of Cauchy matrix type.
We can conduct similar analysis on two-soliton interaction and present same results as in \cite{Matsuno-JPA-2012a}.
For completeness of this paper, in the following we skip details but only sketch main results.

To analyze  two-soliton interaction, we rewrite 1SS \eqref{1ss} in the following form,
\begin{equation}\label{1ss-u}
u^{}_{\mathrm{1SS}}[\xi_1;\mathcal{X}_1]
=\frac{(k_{1}^{2}-k_{1}^{*2})y_{1}}{|k_1|^2(k_{1}^{*}-k_{1}|y_{1}|^2)},
\end{equation}
where
\[ y_{j}=e^{\xi_{j}+i\mathcal{X}_{j}}, ~~ \xi_{j}=-2a_{j}b_{j}(x-m_{j}t),~~
\mathcal{X}_{j}=(a_{j}^{2}-b_{j}^{2})(x+m_{j}t),~~
m_{j}=\frac{1}{|k_j|^{4}},
\]
and we also assume $a_{j}b_{j}>0$, $c_{j}=d_{j}=1$ without loss of generality.

With these notations, the  2SS is written as
\begin{equation}\label{2ss}
u^{}_{\mathrm{2SS}}=\frac{g }{f },
\end{equation}
where
\begin{align*}
g =&\frac{1}{|k_{1}k_2|^2}\biggl[\!-\!k_{2}^{*3}(k_{1}^{2}\!-\!k_{1}^{*2})(k_{1}^{2}\!-
\!k_{2}^{2})(k_{1}^{*2}\!-\!k_{2}^{2})\frac{y_{1}y_{2}}{y_{1}^{*}}\!
+\!k_{2}^{3}(k_{1}^{2}\!-\!k_{1}^{*2})(k_{1}^{2}\!
-\!k_{2}^{*2})(k_{1}^{*2}\!-\!k_{2}^{*2})\frac{y_{1}}{y_{1}^{*}y_{2}^{*}}\\
&-\!k_{1}^{*3}(k_{1}^{2}\!-\!k_{2}^{2})(k_{1}^{2}\!-\!k_{2}^{*2})(k_{2}^{2}\!-\!k_{2}^{*2})
\frac{y_{1}y_{2}}{y_{2}^{*}}
+k_{1}^{3}(k_{1}^{*2}-k_{2}^{2})(k_{1}^{*2}-k_{2}^{*2})(k_{2}^{2}-k_{2}^{*2})
\frac{y_{2}}{y_{1}^{*}y_{2}^{*}}\biggr],\\
f =&|k_{1}^{*2}\!-\!k_{2}^{2}|^2 \biggl( k_{1}^{*}k_{2}\frac{y_2}{y_{1}^{*}}\!+\!
k_{1}k_{2}^{*} \frac{y_{1}}{y_{2}^{*}}\biggr)
-|k_{1}^{2}\!-\!k_{2}^{2}|^2\biggl( k_1 k_2 y_1y_2+ k_{1}^{*}k_{2}^{*}
\frac{1}{y_{1}^{*}y_{2}^{*}}\biggr)\\
& +(k_{1}^{2}\!-\!k_{1}^{*2})(k_{2}^{2}\!-\!k_{2}^{*2})\biggl( |k_1|^2\frac{y_{1}}{y_{1}^{*}}
+\!|k_{2}|^2 \frac{y_{2}}{y_{2}^{*}}\biggr).
\end{align*}
Then we have (cf.\cite{Matsuno-JPA-2012a}), in the coordinate frame $(\xi_j,t)$,
\[u^{}_{\mathrm{2SS}}
\sim u^{}_{\mathrm{1SS}}[\xi_{j}+\Delta \xi_{j}^{(\pm)};\mathcal{X}_{j}+\Delta \mathcal{X}_{j}^{(\pm)}],
~~ t \to \pm \infty,~~ (j=1,2)\]
where $u^{}_{\mathrm{1SS}}[\xi_j;\mathcal{X}_j]$ is given as \eqref{1ss-u},
\begin{align*}
&\Delta \xi_{1}^{(\pm)}=\pm\ln|\frac{k_{1}^{2}-k_{2}^{*2}}{k_{1}^{2}-k_{2}^{2}}|,
~~ \Delta \mathcal{X}_{1}^{(\pm)}=\pm \arg\frac{k_{1}^{2}-k_{2}^{*2}}{k_{1}^{2}-k_{2}^{2}}
\pm \arg\frac{k_{2}^{2}}{k_{2}^{*2}},\\
&\Delta \xi_{2}^{(\pm)}=\mp \ln|\frac{k_{1}^{*2}-k_{2}^{2}}{k_{1}^{2}-k_{2}^{2}}|,~~
\Delta \mathcal{X}_{2}^{(\pm)}=\mp \arg\frac{k_{1}^{*2}-k_{2}^{2}}{k_{1}^{2}-k_{2}^{2}}
\mp \arg\frac{k_{1}^{2}}{k_{1}^{*2}}.
\end{align*}
This indicates that, after interaction, the soliton described  by $k_j$ gets a phase shift
\[2(-1)^{j-1}\left(\frac{\Delta \xi_{j}^{(+)}}{-2a_jb_j} +
\frac{i \Delta \mathcal{X}_{j}^{(+)}}{a_{j}^{2}-b_{j}^{2}}\right).\]
Such an interaction is depicted in Fig.\ref{F-1}(b).

\subsubsection{Breathers}\label{sec-4-1-3}

Note that velocity of a single soliton is governed by $1/|k_j|^4$.
This means in 2SS when $|k_1|= |k_2|$ there will be two parallel solitons,
while in this case periodic interactions, i.e. breathers, occur.

When  $|k_1|= |k_2|$, the envelop of the 2SS \eqref{2ss} is
\begin{equation}\label{2ss-FL}
|u|^2=\frac{G(x,t)}{2(a_{1}^{2}+a_{2}^{2})^{2}F(x,t)},
\end{equation}
with
\begin{align*}
G(x,t)=&16\Bigl\{a_{2}b_{2}[Z_{1}(X^{2}-4A_{1}A_{2})-2Z_{2}X(A_{1}+A_{2})]
\cosh 2Y_{1}\cos(Y_{3}-Y_{4})\\
& -a_{2}b_{2}[Z_{2}(X^{2} -4A_{1}A_{2})+2Z_{1}X(A_{1}+A_{2})] \sinh 2Y_{1}\sin(Y_{3}-Y_{4})\\
&+a_{1}b_{1}[Z_{3}(X^{2}+4A_{1}A_{2})+2Z_{4}X (A_{1}-A_{2})] \cosh 2Y_{2}\Bigr\}^{2}\\
&+16\Bigl\{a_{2}b_{2}[Z_{2}(X^{2}-4A_{1}A_{2})+2Z_{1}X(A_{1}+A_{2})]\sinh 2Y_{1}\cos(Y_{3} -Y_{4})\\
& +a_{2}b_{2}[Z_{1}(X^{2}-4A_{1}A_{2})-2Z_{2}X(A_{1}+A_{2})]\cosh 2Y_{1}\sin(Y_{3}-Y_{4})\\
& +a_{1}b_{1}[Z_{4}(X^{2}+4A_{1}A_{2})-2Z_{3}X(A_{1}-A_{2})]\sinh 2Y_{2}\Bigr\}^{2},\\
F(x,t)=&\Bigl[16 a_{1}a_{2}b_{1}b_{2}|k_1|^{2}\cos(Y_{3}-Y_{4})\\
& -A_{3}(X^{2}+4A_{1}^{2})\cosh 2(Y_{1}-Y_{2})+A_{4}(X_{1}^{2}
+4A_{2}^{2})\cosh 2(Y_{1}+Y_{2})\Bigr]^{2}\\
&+\Bigl[A_{6}(X^{2}+4A_{1}^{2})\sinh 2(Y_{1}-Y_{2})
+A_{5}(X^{2}+4A_{2}^{2})\sinh 2(Y_{1}+Y_{2})\Bigr]^{2},
\end{align*}
where
\begin{equation}\label{Y1}
\begin{array}{l}
X=a_{1}^{2}-b_{1}^{2}-a_{2}^{2}+b_{2}^{2},\quad \theta_{1}=x-\frac{t}{(a_{1}^{2}+b_{1}^{2})^{2}},
\quad \theta_{2}=x+\frac{t}{(a_{1}^{2}+b_{1}^{2})^{2}},\\
Y_{1}=a_{1}b_{1}\theta_{1},\quad Y_{2}=a_{2}b_{2}\theta_{1},
\quad Y_{3}=(a_{1}^{2}-b_{1}^{2})\theta_{2},\quad Y_{4}=(a_{2}^{2}-b_{2}^{2})\theta_{2},\\
Z_{1}=3a_{1}^{2}b_{1}-b_{1}^{3},\quad Z_{2}=a_{1}^{3}-3a_{1}b_{1}^{2},
\quad Z_{3}=3a_{2}^{2}b_{2}-b_{2}^{3},\quad Z_{4}=a_{2}^{3}-3a_{2}b_{2}^{2},\\
A_{1}=a_{1}b_{1}+a_{2}b_{2},\quad A_{2}=a_{1}b_{1}-a_{2}b_{2},\quad A_{3}=a_{1}a_{2}+b_{1}b_{2},
\quad  A_{4}=a_{1}a_{2}-b_{1}b_{2},\\
A_{5}=a_{1}b_{2}+a_{2}b_{1},\quad A_{6}=a_{1}b_{2}-a_{2}b_{1},
\end{array}
\end{equation}
and we have taken $c_j=d_j=1$ without loss of generality.
In particular, on the line   $x=t/{|k_1|^{4}}$,
the value of $|u|^2$ is
\begin{equation}\label{2ss-br-FL}
|u|^2_{x=t/{|k_1|^{4}}}=\frac{G_{1}}{2(a_{1}^{2}+a_{2}^{2})^{2}F_{1}},
\end{equation}
where
\begin{align*}
G_{1}=&16a_{2}^{2}b_{2}^{2}[Z_{1}(X^{2}-4A_{1}A_{2})-2Z_{2}X(A_{1}+A_{2})]^{2}
+16a_{1}^{2}b_{1}^{2}[Z_{3}(X^{2}+4A_{1}A_{2})\\
&+2Z_{4}X(A_{1}-A_{2})]^{2}+32a_{1}b_{1}a_{2}b_{2}[Z_{3}(X^{2}+4A_{1}A_{2})+2Z_{4}X(A_{1}-A_{2})]\\
&\times [Z_{1}(X^{2}-4A_{1}A_{2})-2Z_{2}X(A_{1}+A_{2})]\cos \frac{2Xt}{(a_{1}^{2}+b_{1}^{2})^{2}},\\
F_{1} =&\Bigl[16 a_{1}a_{2}b_{1}b_{2}|k_1|^{2}
\cos \frac{2Xt}{(a_{1}^{2}+b_{1}^{2})}-A_{3}(X^{2}+4A_{1}^{2})+A_{4}(X_{1}^{2}
+4A_{2}^{2}))\Bigr]^{2}.
\end{align*}
Thus, it is obvious to see that the period of interaction is given by
\begin{equation}\label{T-breather}
T=\frac{(a_{1}^{2}+b_{1}^{2})^{2}\pi}{(a_{1}^{2}-b_{1}^{2})-(a_{2}^{2}-b_{2}^{2})}.
\end{equation}

To summarize, we have the following,
\begin{proposition}\label{prop-4-1}
A breather from a 2SS occurs when
$k_1\in \mathbb{C}, a_1b_1\neq 0$, $|k_2|=|k_1|$, but
$(a_{1}^{2}-a_{2}^{2})(b_{1}^{2}-b_{2}^{2})\neq 0$,
i.e., $k_2$ is not any reflection point of $k_1$ w.r.t. $x$-axis, $y$-axis, or the origin;
the breather travels along the line $x=t/{|k_1|^{4}}$
with period $T$ given in \eqref{T-breather}.
\end{proposition}

Fig.\ref{F-2}(a) describe a breather coming from two solitons whit same initial phase,
while Fig.\ref{F-2}(b) describe a breather coming from two solitons with different initial phases.

\captionsetup[figure]{labelfont={bf},name={Fig.},labelsep=period}
\begin{figure}[ht]
\centering
\subfigure[ ]{
\begin{minipage}[t]{0.45\linewidth}
\centering
\includegraphics[width=2.1in]{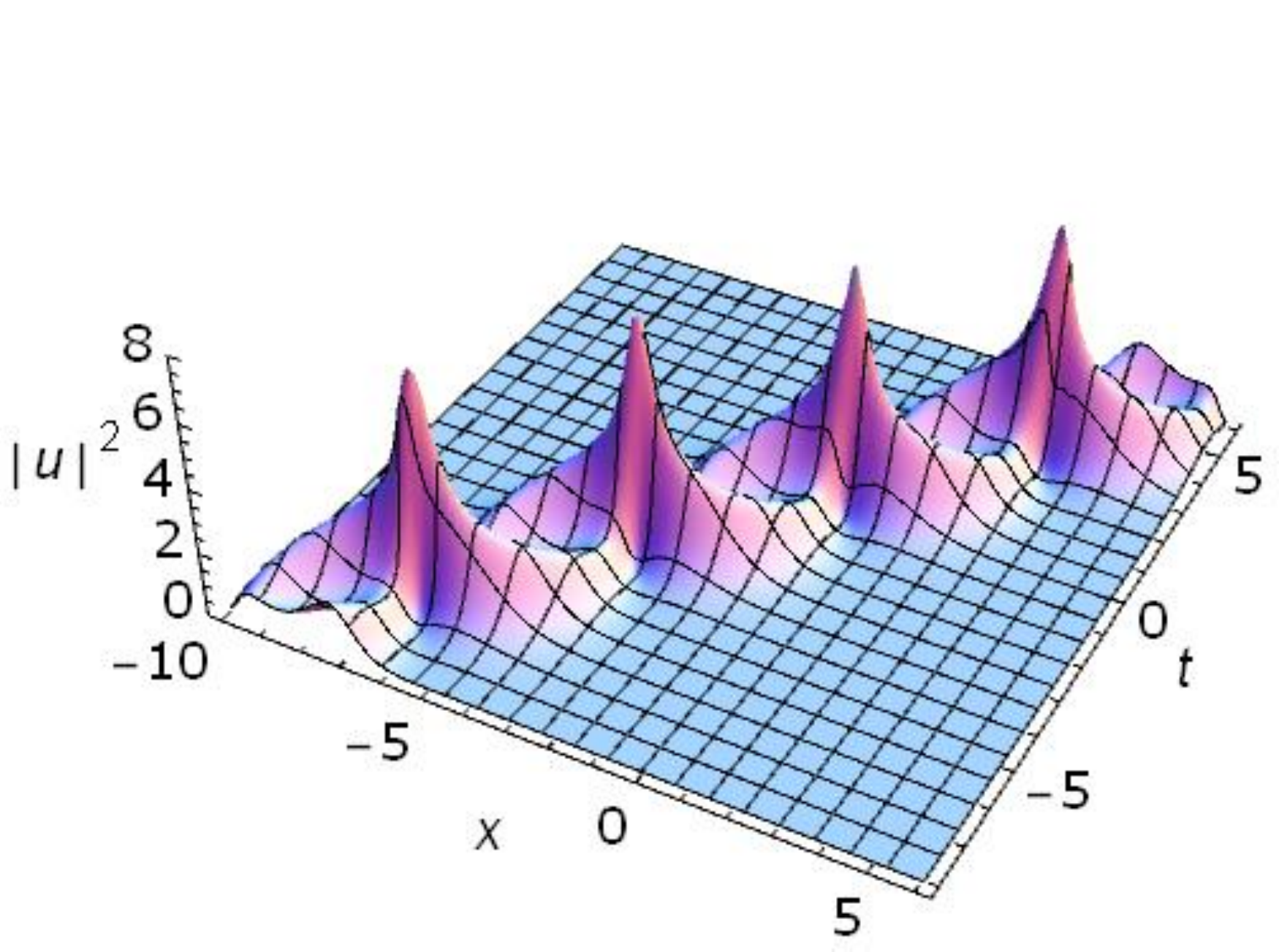}
\end{minipage}%
}%
\subfigure[ ]{
\begin{minipage}[t]{0.45\linewidth}
\centering
\includegraphics[width=2.1in]{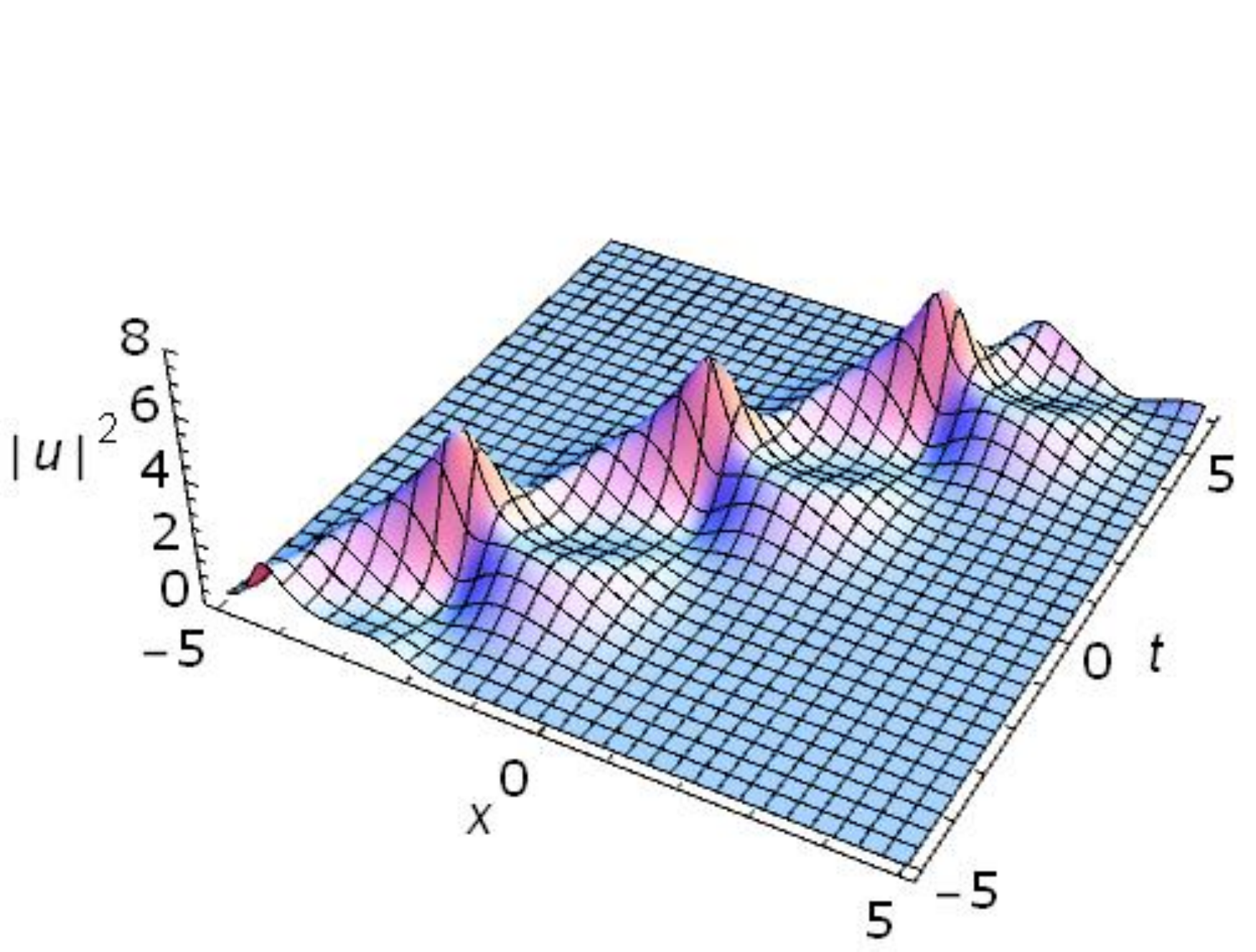}
\end{minipage}%
}
\caption{Breathers of the FL equation \eqref{FL2}.
(a) $|u|^2$ give by \eqref{2ss-FL} with $k_1=\frac{1}{2}+\frac{\sqrt{3}}{2}i,
k_2=\frac{\sqrt{3}}{2}+\frac{1}{2}i$, $c_1=0.5$, $d_1=1$,  $c_2=0.5$ and $d_2=1$.~
(b) $|u|^2$ give by \eqref{2ss-FL} with $k_1=1+\frac{1}{2}i, k_2=\frac{1}{2}-i$, $c_1=0.5$,
$d_1=1$, $c_2=1$ and $d_2=1$.
}
\label{F-2}
\end{figure}

\subsubsection{Double-pole solutions}\label{sec-4-1-4}

The simplest Jordan block solution  is given through \eqref{2ss-pp},
\begin{equation*}
u^{}_{2}=\frac{g}{f},
\end{equation*}
with
\begin{equation*}
f=|\partial_x\phi, \partial^2_x\phi; \psi, \partial_x\psi|,~~
g=|\phi, \partial_x\phi, \partial^2_x\phi;\partial_x \psi|,
\end{equation*}
but here $\phi$ and $\psi$ are taken as
\begin{align*}
& \phi=\left(c_1 e^{\eta(k_1)}, c_1 \partial_{k_1} e^{\eta(k_1)},
d_1 e^{\eta(k^*_1)}, d_1\partial_{k^*_1} e^{\eta(k^*_1)}\right)^T,\\
& \psi=\left(k_1d_{1}^{*}e^{-\eta(k_1)}, d_{1}^{*} e^{-\eta(k_1)}+d_{1}^{*}k_{1}\partial_{k_1} e^{-\eta(k_1)},
c_{1}^{*}k_{1}^{*} e^{-\eta(k^*_1)}, c_{1}^{*}e^{-\eta(k^*_1)}
+c_{1}^{*}k^*_1\partial_{k^*_1} e^{-\eta(k^*_1)}\right)^T.
\end{align*}
The corresponding envelope is
\begin{equation}\label{2ss-D-FL}
|u_{2}|^2=\frac{16a^{2}b^{2}G_{2} }{F_{2} },
\end{equation}
with
\begin{align*}
G_{2} =&\Bigl[(B_{1}+B_{2})\mathrm{e}^{2Y_{1}}-(B_{1}-B_{2})\mathrm{e}^{-2Y_{1}}\Bigr]^{2}+
\Bigl[(B_{3}+B_{4})\mathrm{e}^{2Y_{1}}+(B_{3}-B_{4})\mathrm{e}^{-2Y_{1}}\Bigr]^{2},\\
F_{2} =&16\Bigl[B_{5}(\mathrm{e}^{4Y_{1}}
+\mathrm{e}^{-4Y_{1}})-2(a^{2}+b^{2})^{2}(a^{4}+b^{4})-16a^{2}b^{2}t^{2}\\
& -32a^{2}b^{2}(6a^{2}b^{2}-a^{4}-b^{4})xt-16a^{2}b^{2}(a^{2}+b^{2})^{4}x^{2}\Bigr]^{2}\\
&+16\Bigl\{B_{6}(\mathrm{e}^{4Y_{1}}-\mathrm{e}^{-4Y_{1}})
+16a^{2}b^{2}(a^{2}-b^{2})[t-(a^{2}+b^{2})^{2}x]\Bigr\}^{2},
\end{align*}
where $Y_1$ is defined as in \eqref{Y1},
\begin{align*}
&B_{1}=4a^{3}(a^{2}+b^{2}),\quad B_{2}=8a^{2}b[t-(a^{2}-3b^2)(a^2+b^2)x],\\
&B_{3}=4b^{3}(a^{2}+b^{2}),\quad  B_{4}=8ab^{2}[t+(3a^{2}-b^2)(a^2+b^2)x],\\
&B_{5}=(a^{2}-b^{2})(a^{2}+b^{2})^2,\quad B_{6}=2ab(a^{2}+b^{2})^3,
\end{align*}
and we have taken $k_1=a+ib,~c_1=d_1=1$.

In order to understand asymptotic behavior of $|u|^{2}$,
we consider $|u|^{2}$ in a coordinate frame $(z^{(+)}_{\pm}, t)$, where
\begin{subequations}
\begin{equation}\label{Jordon-trace-1}
z^{(+)}_{\pm}=x-\frac{t}{(a^2+b^2)^2}\pm \frac{2\ln t+\gamma}{4ab},\quad \gamma=\frac{\ln H}{2},
\end{equation}
with
\begin{equation}
H=\frac{2^{16}a^{8}b^{8}}{(a^{2}+b^{2})^{8}[(a^2-b^2)^{2}+4a^{2}b^{2}(a^{2}+b^{2})]}.
\end{equation}
\end{subequations}
In this frame when $t\rightarrow +\infty$, we get
\begin{equation}
|u|^{2}\rightarrow \frac{2a^{2}b^{2}}{(a^{2}+b^{2})^{2}\sqrt{(a^{2}-b^{2})^{2}+4a^{2}b^{2}(a^{2}+b^{2})}
\cosh z^{(+)}_{\pm}-a^{4}+b^{4}}.
\end{equation}
Similarly, in the coordinate $(z^{(-)}_{\pm},t)$, where
\begin{equation}\label{Jordon-trace-2}
z^{(-)}_{\pm}=x-\frac{t}{(a^2+b^2)^2}\pm \frac{2\ln (-t)+\gamma}{4ab},
\end{equation}
when $t\rightarrow -\infty$, we obtain
\begin{equation}
|u|^{2}\rightarrow \frac{2a^{2}b^{2}}{(a^{2}+b^{2})^{2}\sqrt{(a^{2}-b^{2})^{2}+4a^{2}b^{2}(a^{2}+b^{2})}
\cosh z^{(-)}_{\pm}-a^{4}+b^{4}}.
\end{equation}
The above asymptotic analysis indicates, as depicted in Fig.\ref{F-3},
when $|t|$ is large enough the wave will separate into two single solitons asymptotically traveling along  the curves
\begin{equation}\label{Jordon-trace}
x(t)=\frac{t}{(a^2+b^2)^2}\mp\frac{2\ln|t|+\gamma}{4ab}.
\end{equation}

Note that in Fig.\ref{F-3}(b) we give a density plot of (a) as well as the curves given in \eqref{Jordon-trace},
see the red curves. This also illustrates our asymptotic analysis.

\captionsetup[figure]{labelfont={bf},name={Fig.},labelsep=period}
\begin{figure}[ht]
\centering
\subfigure[ ]{
\begin{minipage}[t]{0.45\linewidth}
\centering
\includegraphics[width=2.3in]{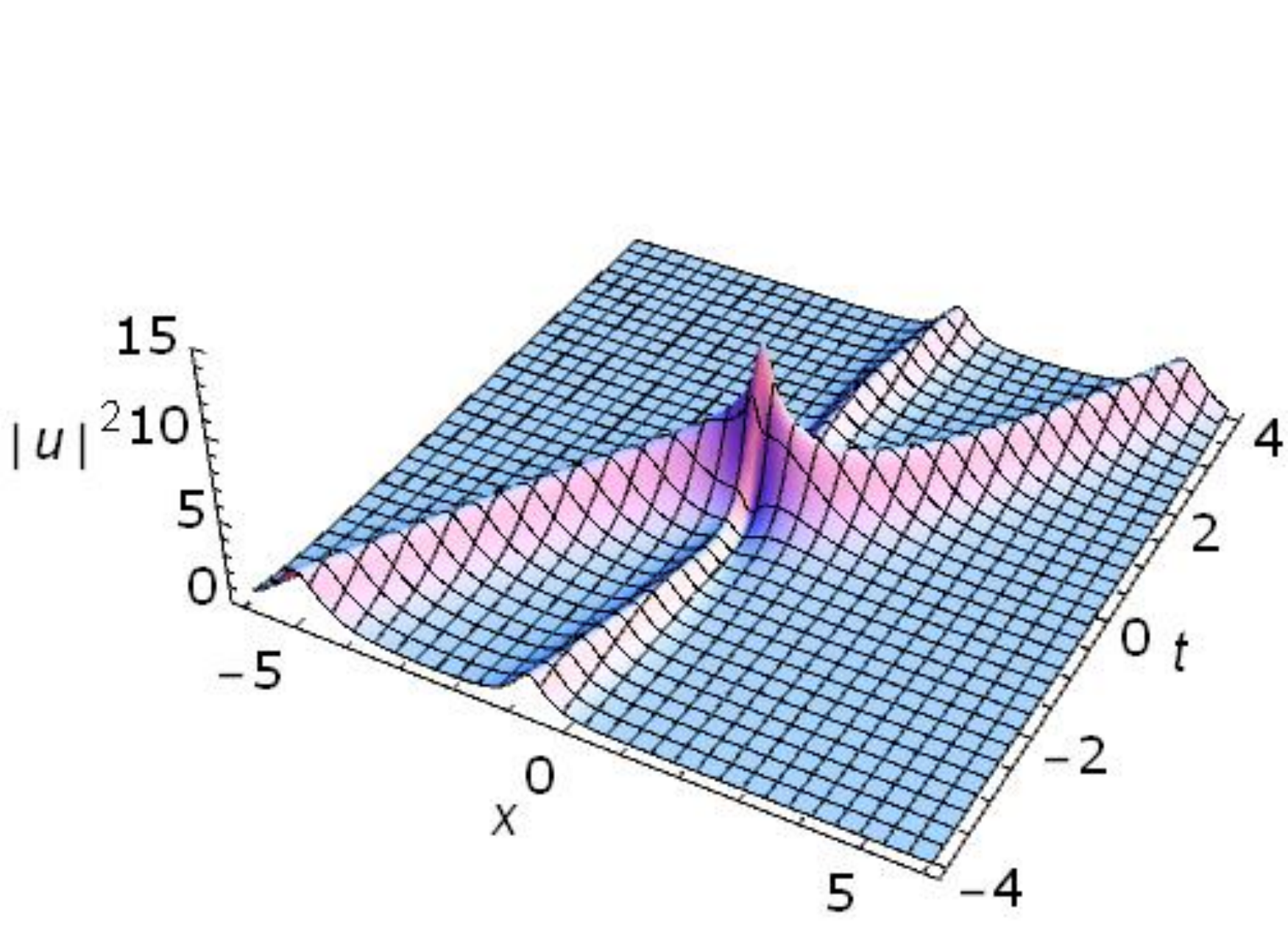}
\end{minipage}%
}%
\subfigure[ ]{
\begin{minipage}[t]{0.45\linewidth}
\centering
\includegraphics[width=1.5in]{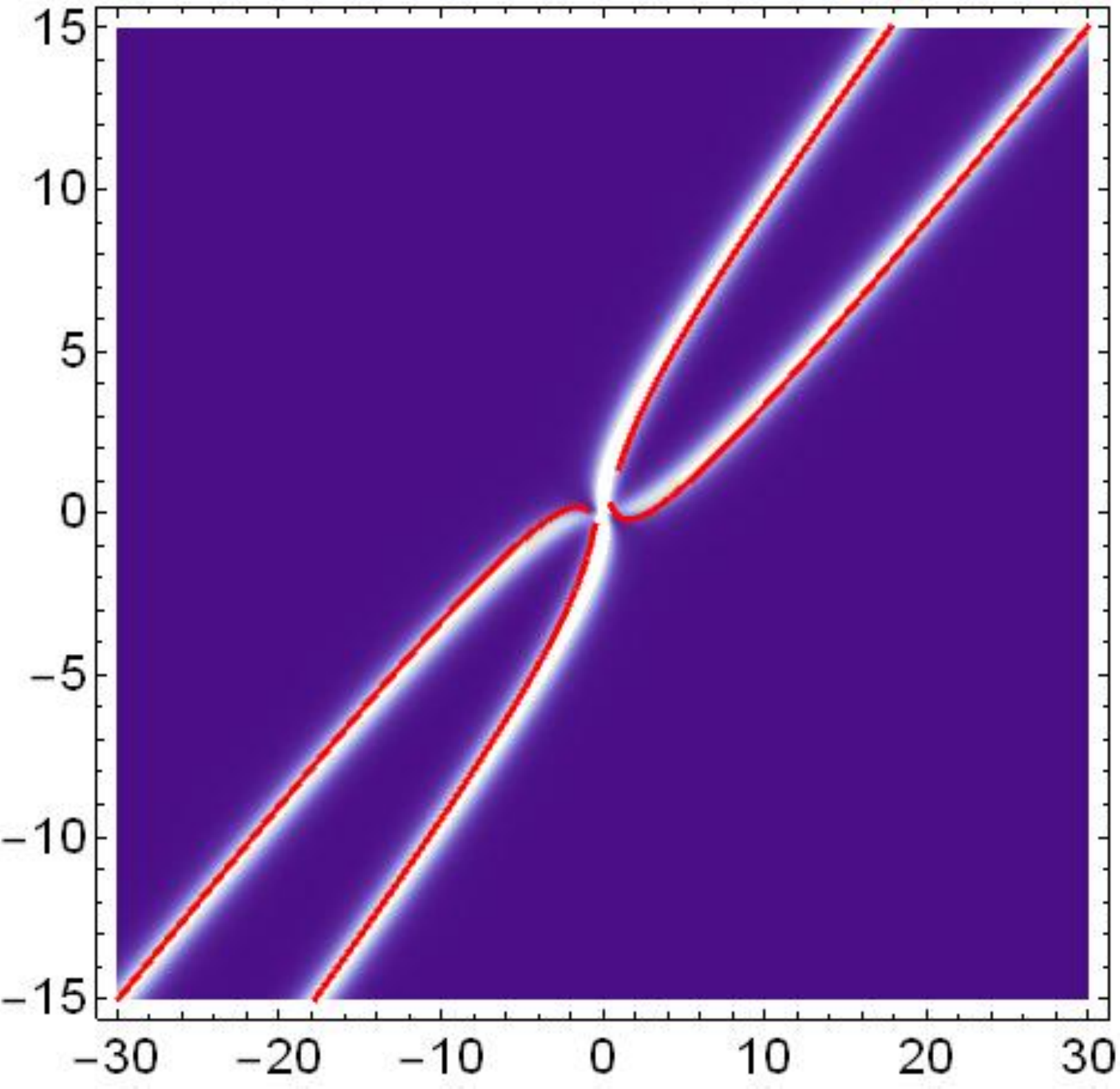}
\end{minipage}%
}
\caption{Shape and motion of Jordan block solution to the  FL equation \eqref{FL2}.
(a) Jordan block solution given by \eqref{2ss-D-FL} with $k=1+0.5i$, $c=d=1$.~
(b) Trajectories of the  solution  in (a).
}
\label{F-3}
\end{figure}

\subsection{Solutions related to real eigenvalues}\label{sec-4-2}

Case (2) in Table \ref{tab-1} contributes solutions that are related to real discrete eigenvalues.
Note that so far these type of solutions are not obtained in inverse scattering transform \cite{Lenells-F-Non-2009}
or Riemann-Hilbert approach \cite{Ai-Xu-AML-2019},
as eigenvalues in those two approaches do not locate on axes.

\subsubsection{Periodic and double periodic solutions}\label{sec-4-2-1}

Consider $\mathbf{K}_{N}$ given in \eqref{A-diag-real} and
$\mathbf{H}_{N}$ given in  \eqref{A-diag-real-H}, where $k_j, h_j \in \mathbb{R}$.
Note that $c^{\pm}\in \mathbb{C}$.
When $N=1$ we have
\begin{equation}\label{1ss-p}
u^{}_{1\mathrm{SS}}=\frac{c_{1}d_{1}(k_{1}^{2}-h_{1}^{2})}
{k_{1}h_{1}\left[c_{1}d_{1}^{*}k_{1}\mathrm{e}^{-i(h_{1}^{2}x+\frac{t}{h_{1}^{2}})}
+c_{1}^{*}d_{1}h_{1}\mathrm{e}^{-i(k_{1}^{2}x+\frac{t}{k_{1}^{2}})}\right]},
\end{equation}
and the corresponding envelop is
\begin{equation}\label{1ss-pp}
|u^{}_{1\mathrm{SS}}|^{2}=
\frac{(k_{1}^{2}-h_{1}^{2})^{2}}
{k_{1}^{2}h_{1}^{2}[k_{1}^{2}+h_{1}^{2}+2k_{1} h_{1}
\sin(\omega-\vartheta_{1})]},
\end{equation}
where
\begin{equation}\label{vartheta}
\vartheta_{1}=(k_{1}^{2}-h_{1}^{2})\left(x-\frac{t}{k_{1}^{2}h_{1}^{2}}\right),
~~ \omega=\arctan\frac{\mathrm{Re}[c_{1}^{2}d_{1}^{*2}]}{\mathrm{Im}[c_{1}^{2}d_{1}^{*2}]},
\end{equation}
and $c_1=c_1^+, d_1=c_1^-$.
We require $k_{1}^{2}\neq h_{1}^{2}$, otherwise $|u|^2=0$.
\eqref{1ss-pp} is a periodic wave characterized as the following,
\begin{align*}
&\textrm{top trajectories}:~~x(t)=\frac{t}{k_{1}^{2}h_{1}^{2}}+\frac{1}{k_{1}^{2}-h_{1}^{2}}
\arctan \frac{\mathrm{Re}[c_{1}^{2}d_{1}^{*2}]}{\mathrm{Im}[c_{1}^{2}d_{1}^{*2}]}+
 \frac{2\kappa \pi}{k_{1}^{2}-h_{1}^{2}}~,~ \kappa\in \mathbb{Z},\\
&\textrm{amplitude}:~~\frac{(k_{1}^{2}-h_{1}^{2})^{2}}{k_{1}^{2}h_{1}^{2}
(k_{1}^{2}+h_{1}^{2}-2|k_{1}h_{1}|)},\\
&\textrm{velocity}:~~\frac{1}{k_{1}^{2}h_{1}^{2}},\\
& \mathrm{period~in~}x, y~\mathrm{direction}:~~T_{x}=\frac{2\pi}{h_{1}^{2}-k_{1}^{2}},
~~T_{y}=\frac{2\pi k_{1}^{2}h_{1}^{2}}{k_{1}^{2}-h_{1}^{2}},\\
& \mathrm{distance~ between~ two~ adjacent~ trajectories:}~~
T_d=\frac{2\pi k_{1}^{2}h_{1}^{2}}{|k_{1}^{2}-h_{1}^{2}|\sqrt{k_{1}^{4}h_{1}^{4}+1}}.
\end{align*}
The wave is depicted in Fig.\ref{F-4}.

\captionsetup[figure]{labelfont={bf},name={Fig.},labelsep=period}
\begin{figure}[ht]
\centering
\begin{minipage}[t]{0.45\linewidth}
\centering
\includegraphics[width=2.3in]{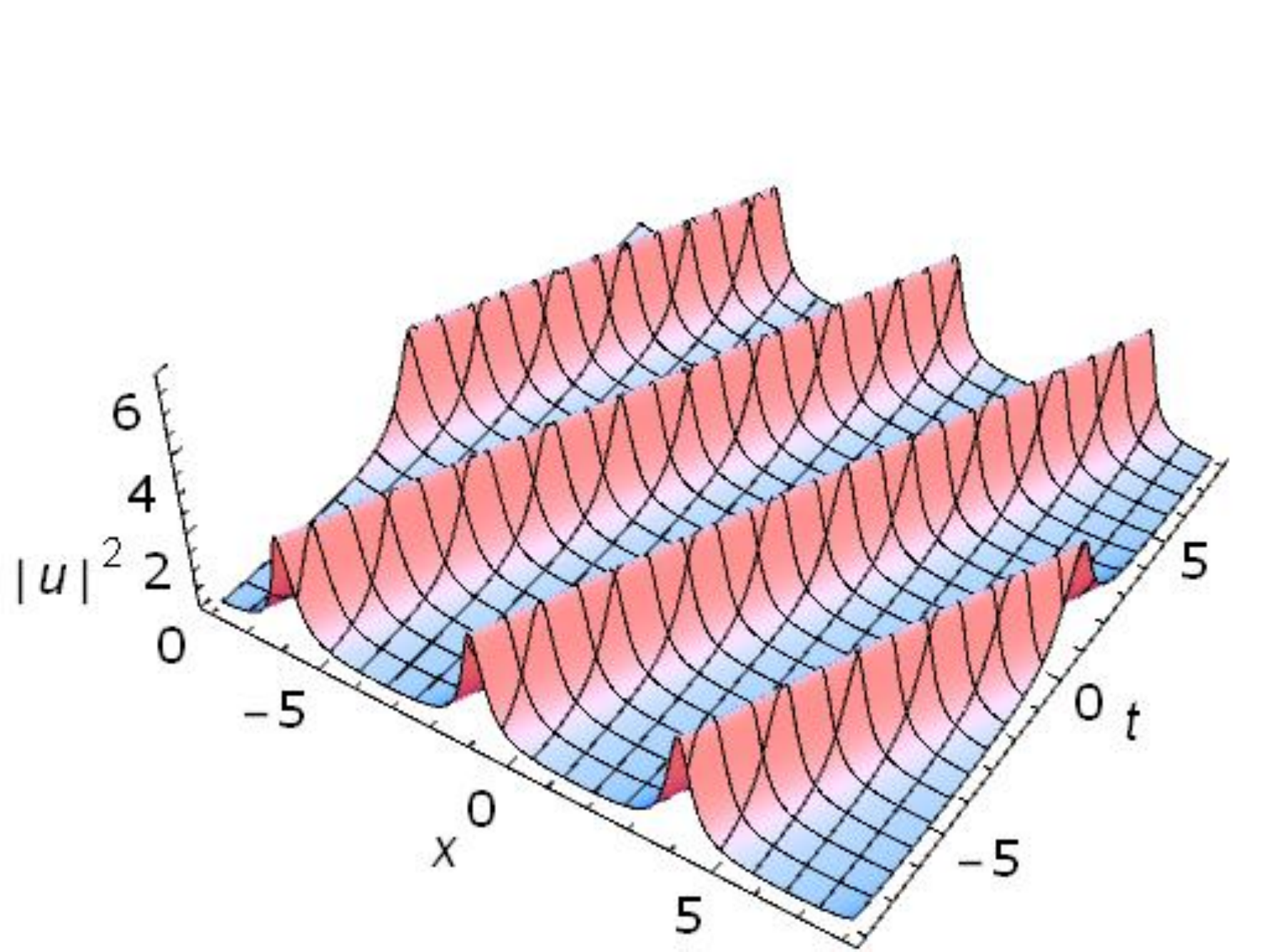}
\end{minipage}%
\caption{Periodic solution of the FL equation \eqref{FL2}, given by \eqref{1ss-pp}
with $k_{1}=1$, $h_{1}=1.5$, $c_{1}=d_{1}=1$.}
\label{F-4}
\end{figure}

In 2SS case, i.e. $N=2$ in  \eqref{A-diag-real} and \eqref{A-diag-real-H}, 2SS is given by \eqref{2ss-pp}
where
\begin{subequations}\label{2ss-phs}
\begin{align}
& \phi=(c_1e^{\eta(k_1)}, c_2e^{\eta(k_2)}, d_1e^{\eta(h_1)}, d_2e^{\eta(h_2)})^T,\\
& \psi=(k_1c_{1}^{*}e^{-\eta(k_1)}, k_2c_{2}^{*}e^{-\eta(k_2)},
-h_1d_{1}^{*}e^{-\eta(h_1)}, -h_2d_{2}^{*}e^{-\eta(h_2)})^T,
\end{align}
\end{subequations}
$\eta$ is defined by \eqref{eta}, $k_j, h_j \in \mathbb{R}$ and $c_j,d_j\in \mathbb{C}$.

In this case, $|u^{}_{\mathrm{2SS}}|^2$ exhibit double periodic interactions, as illustrated in Fig.\ref{F-5}.
This is not surprised from the periodic feature of 1SS.

\captionsetup[figure]{labelfont={bf},name={Fig.},labelsep=period}
\begin{figure}[ht]
\centering
\subfigure[ ]{
\begin{minipage}[t]{0.45\linewidth}
\centering
\includegraphics[width=2.3in]{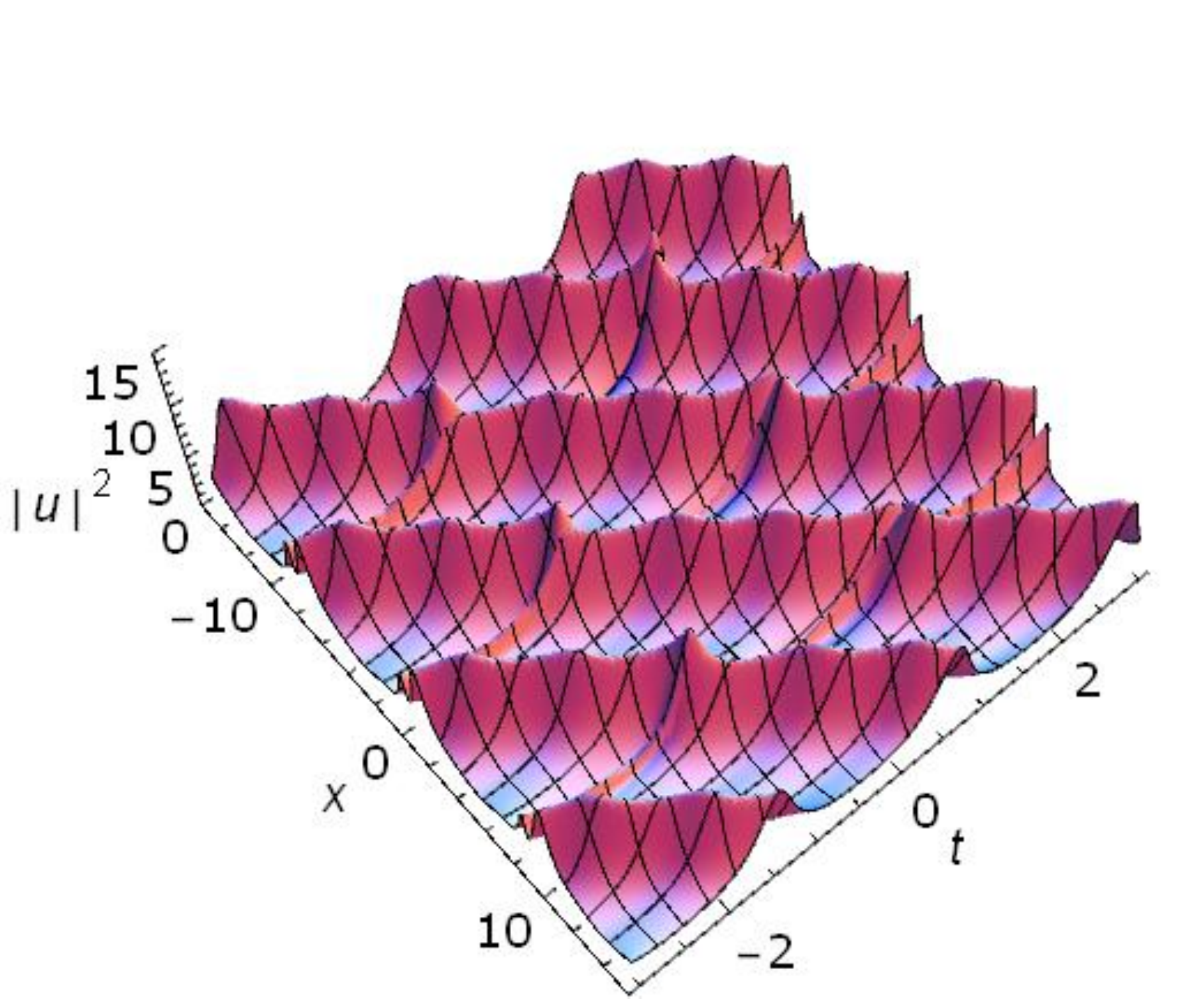}
\end{minipage}%
}
\subfigure[ ]{
\begin{minipage}[t]{0.45\linewidth}
\centering
\includegraphics[width=1.5in]{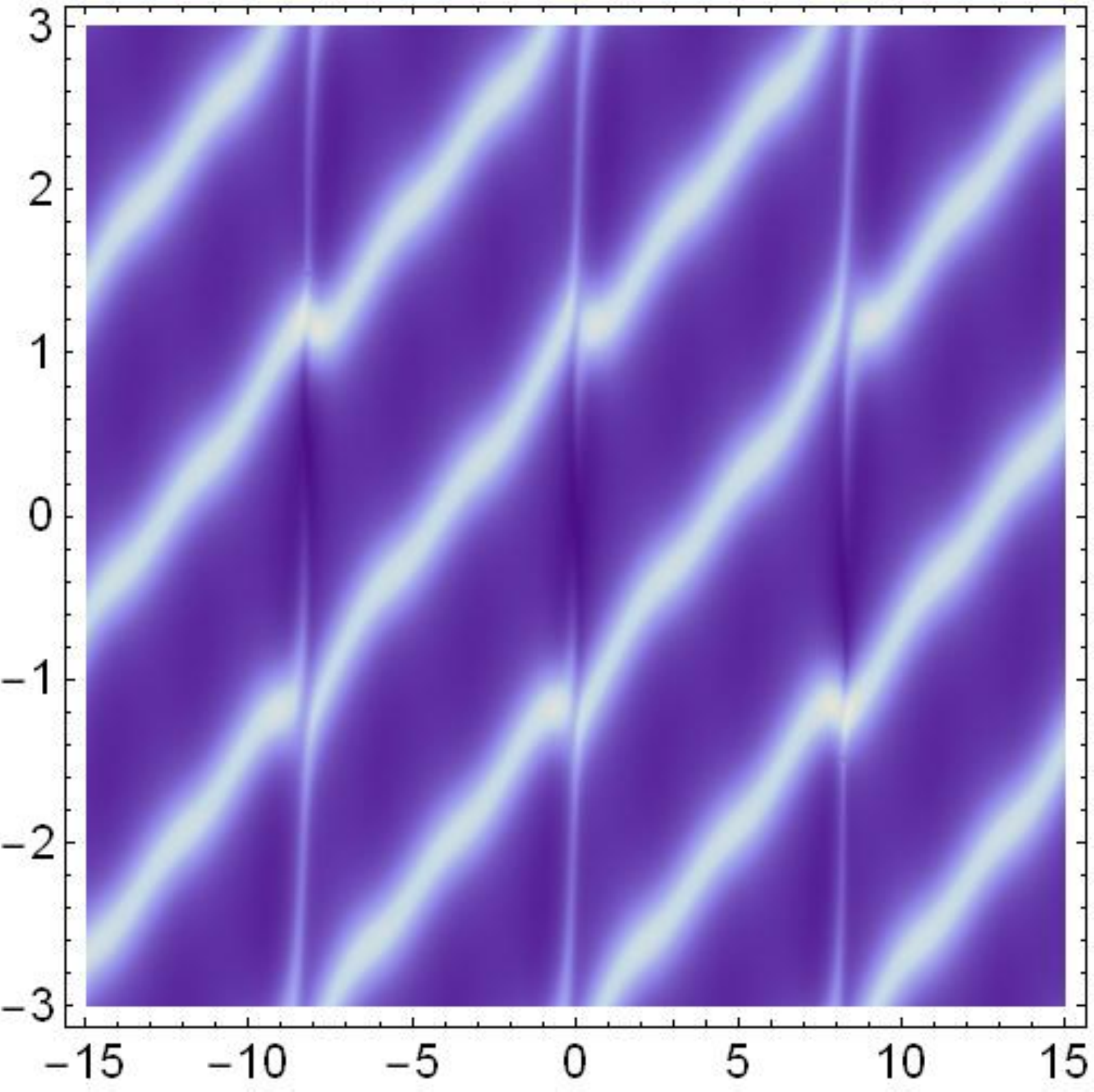}
\end{minipage}%
}
\caption{Double periodic solution of FL equation \eqref{FL2} in Case (2).
(a) Envelop $|u|^2$ of 2SS given by \eqref{2ss-pp} with \eqref{2ss-phs}
 where $k_{1}=2$, $h_{1}=-1.8$, $k_{2}=1$, $h_{2}=0.5$,
$c_1=d_1=c_2=d_2=1$.~(b) Density plot of (a).}
\label{F-5}
\end{figure}

\subsubsection{Solitary waves with algebraic decay}\label{sec-4-2-2}

Although when $\mathbf{K}_N$ and $\mathbf{H}_N$ are diagonal with distinct diagonal elements,
solutions exhibit (multi-)periodic interaction behavior,
in resonant case, for example, both $\mathbf{K}_N$ and $\mathbf{H}_N$ are Jordan blocks with $N=2$,
the resonance leads to algebraic decayed waves asymptotically, without periodic interaction.
See Fig.\ref{F-6} as an example.

Let us consider both $\mathbf{K}_N$ and $\mathbf{H}_N$ to be 2 by 2  Jordan blocks
\begin{equation}\label{KH-jordan}
\mathbf{K}_{2}=\left(
  \begin{array}{cc}
    k_1 & 0  \\
    1& k_1
  \end{array}
\right),
~~
\mathbf{H}_{2}=\left(
  \begin{array}{cc}
    h_1 & 0  \\
    1& h_1
  \end{array}
\right),~~ k_1,h_1\in \mathbb{R}.
\end{equation}
Solution is given by \eqref{2ss-pp} but, from \eqref{psi-phi-pmm}
\begin{align*}
& \phi=\left(c_1 e^{\eta(k_1)}, c_1 \partial_{k_1} e^{\eta(k_1)},
d_1 e^{\eta(h_1)}, d_1\partial_{h_1} e^{\eta(h_1)}\right)^T,\\
& \psi=\left(k_1c_{1}^{*}e^{-\eta(k_1)}, c_{1}^{*} e^{-\eta(k_1)}+c_{1}^{*} k_1\partial_{k_1} e^{-\eta(k_1)},
 -h_1d_{1}^{*}e^{-\eta(h_1)}, -d_{1}^{*}e^{-\eta(h_1)}-d_{1}^{*}h_1\partial_{h_1} e^{-\eta(h_1)}\right)^T,
\end{align*}
where $\eta$ is defined by \eqref{eta}, $k_1, h_1 \in \mathbb{R}$ and $c_1,d_1\in \mathbb{C}$.
Envelope is
\begin{equation}\label{case2-jordon}
|u_2(x,t)|^{2}=\frac{4(h_{1}^{2}-k_{1}^{2})^{2}G_{2}}{F_{2}},
\end{equation}
where
\begin{align*}
&G_{2}(x,t)=M_{1}^{2}+M_{2}^{2}+M_{3}^{2}+M_{4}^{2}
+2(M_{1}M_{3}+M_{2}M_{4})\cos \vartheta_1
+2(M_{1}M_{4}-M_{2}M_{3})\sin \vartheta_1,\\
&F_{2}(x,t)=N_{1}^{2}+N_{2}^{2}+N_{3}^{2}\!+\!N_{4}^{2}
+2N_{1}N_{2}\cos 2\vartheta_1 -\!2 N_{3}(N_{1}\!-\!N_{2}) \sin \vartheta_1
+2 N_{4}(N_{1}\!+\!N_{2})\cos \vartheta_1,
\end{align*}
with $\vartheta_1$ defined in \eqref{vartheta},
\begin{align*}
&M_{1}=-2k_{1}(k_{1}^{2}-h_{1}^{2})(h_{1}^{4}x-t),
\quad M_{2}=-k_{1}h_{1}^{2}(3k_{1}^{2}+h_{1}^{2}),\\
&M_{3}=2h_{1}(k_{1}^{2}-h_{1}^{2})(k_{1}^{4}x-t),
\quad M_{4}=-k_{1}^{2}h_{1}(3h_{1}^{2}+k_{1}^{2}),\\
&N_{1}=-4h_{1}^{5}k_{1}^{3},\quad N_{2}=-4h_{1}^{3}k_{1}^{5},~~
N_{3}=-2(k_{1}^{2}-h_{1}^{2})(k_{1}^{4}-h_{1}^{4}) (h_{1}^{2}k_{1}^{2} x-t),\\
&N_{4}=-h_{1}^{2}k_{1}^{2}(h_{1}^{4}+6h_{1}^{2}k_{1}^{2}+k_{1}^{4})
+ 4(k_{1}^{2}-h_{1}^{2})^2(t-k_{1}^{4}x)(t-h_{1}^{4}x),
\end{align*}
and we have taken $c_1=d_1=1$.

Obviously, when both $|x|$ and $|t|$ go to infinity, $|u|^2$ is dominated by
\[|u|^2 \sim~ \frac{M_1^2+M_3^2}{N_3^2+N_4^2}.
\]
Thus, we consider the above $|u|^2$ in the coordinate frame
$(X_1,t)$ and $(X_2,t)$, respectively, where
\begin{equation}\label{XX}
X_1=x-\frac{t}{k_{1}^{4}},~~X_2=x-\frac{t}{h_{1}^{4}}.
\end{equation}
After taking $t \to \pm\infty$, we arrive at the following.
\begin{proposition}\label{prop-4-2}
Asymptotically,  $|u|^2$ given in \eqref{case2-jordon} obeys
\begin{equation}
|u|^2 \sim \frac{4}{k_1^{2}(1+4k_1^{4}X_{1}^{2})}
\end{equation}
in $(X_1,t)$, and
\begin{equation}
|u|^{2}\sim \frac{4}{h_1^{2}(1+4h_1^{4}X_{2}^{2})}
\end{equation}
in $(X_2,t)$, where $X_j$ are given in \eqref{XX}.
\end{proposition}

This indicates that, when $t$ is large enough,  $|u|^2$  are two algebraic decayed waves, with
amplitudes $4/k_1^2$ and $4/h^2_1$, respectively, as depicted in Fig.\ref{F-6}(a).
When $t$ is not large, periodic effect  can still be observed, see Fig.\ref{F-6}(d).
It is worthy to mention that the above asymptotic analysis indicates that,
asymptotically, there is no phase shift after interaction, which is different from
normal soliton interactions.
This is illustrated in Fig.\ref{F-6}(b) where the density plot of (a) is overlapped by the lines $X_1=0$ and $X_2=0$.

\captionsetup[figure]{labelfont={bf},name={Fig.},labelsep=period}
\begin{figure}[ht]
\centering
\subfigure[ ]{
\begin{minipage}[t]{0.45\linewidth}
\centering
\includegraphics[width=2.3in]{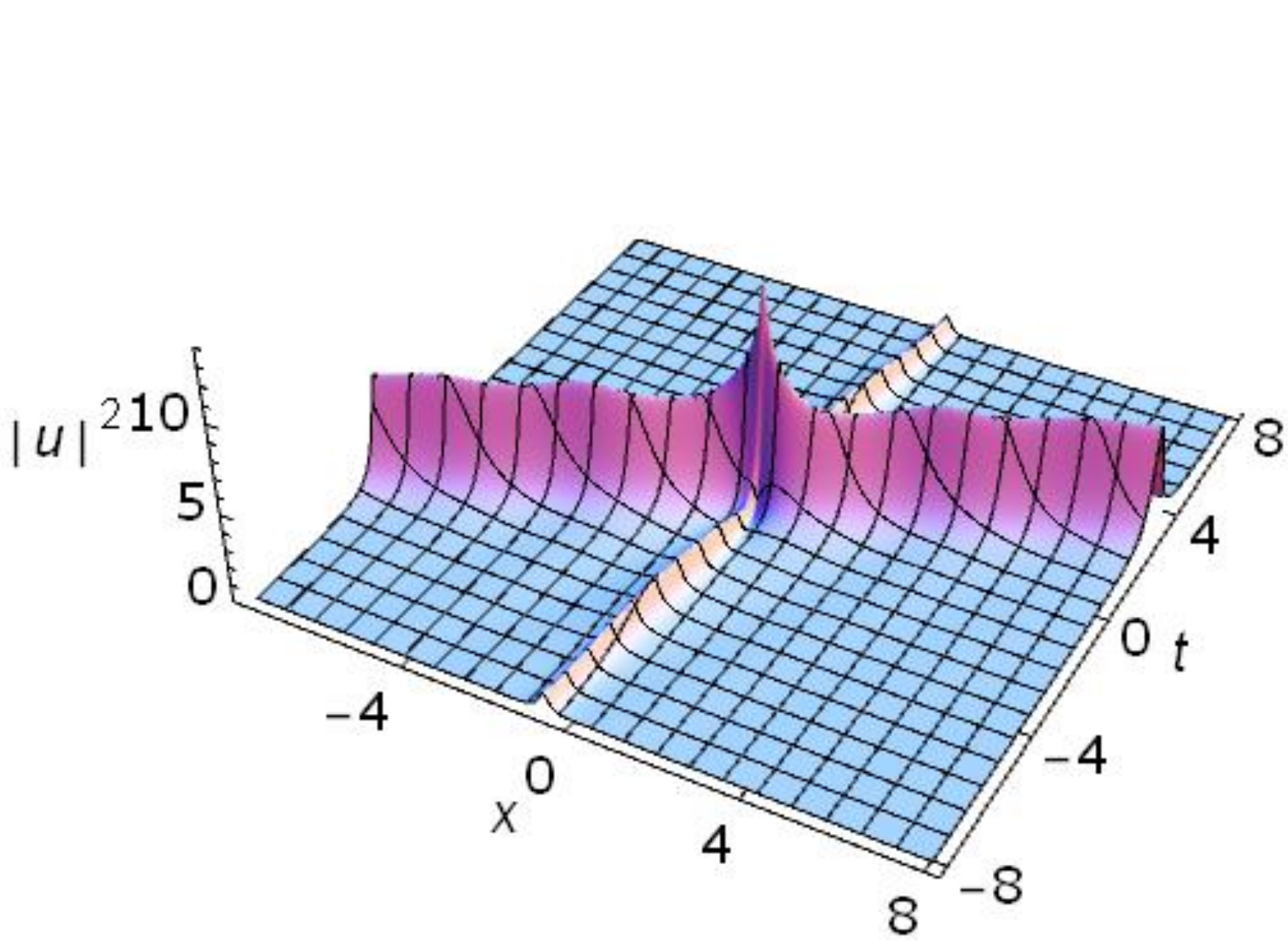}
\end{minipage}%
}%
\subfigure[ ]{
\begin{minipage}[t]{0.45\linewidth}
\centering
\includegraphics[width=1.5in]{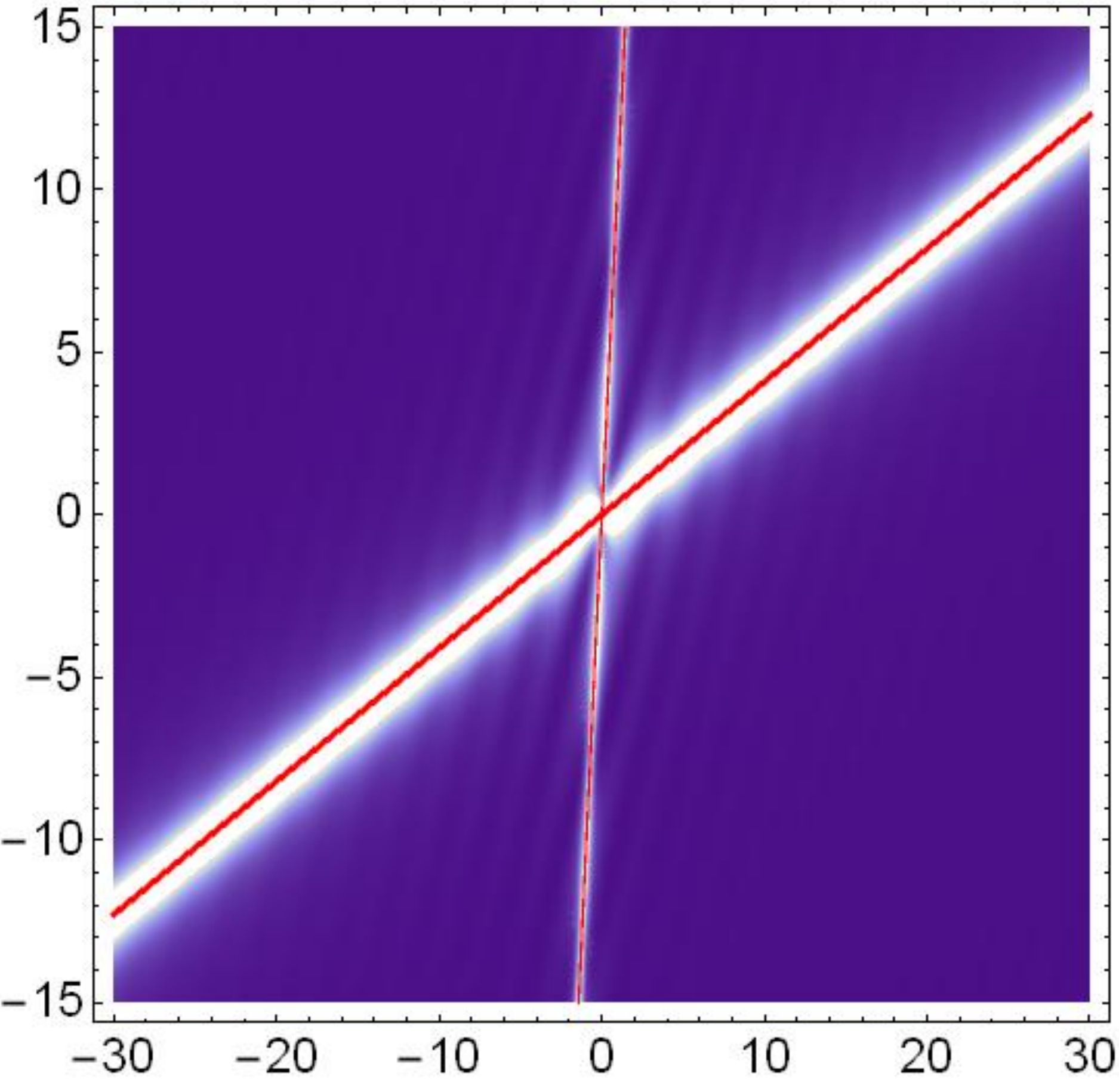}
\end{minipage}%
}
\centering
\subfigure[ ]{
\begin{minipage}[t]{0.45\linewidth}
\centering
\includegraphics[width=2.1in]{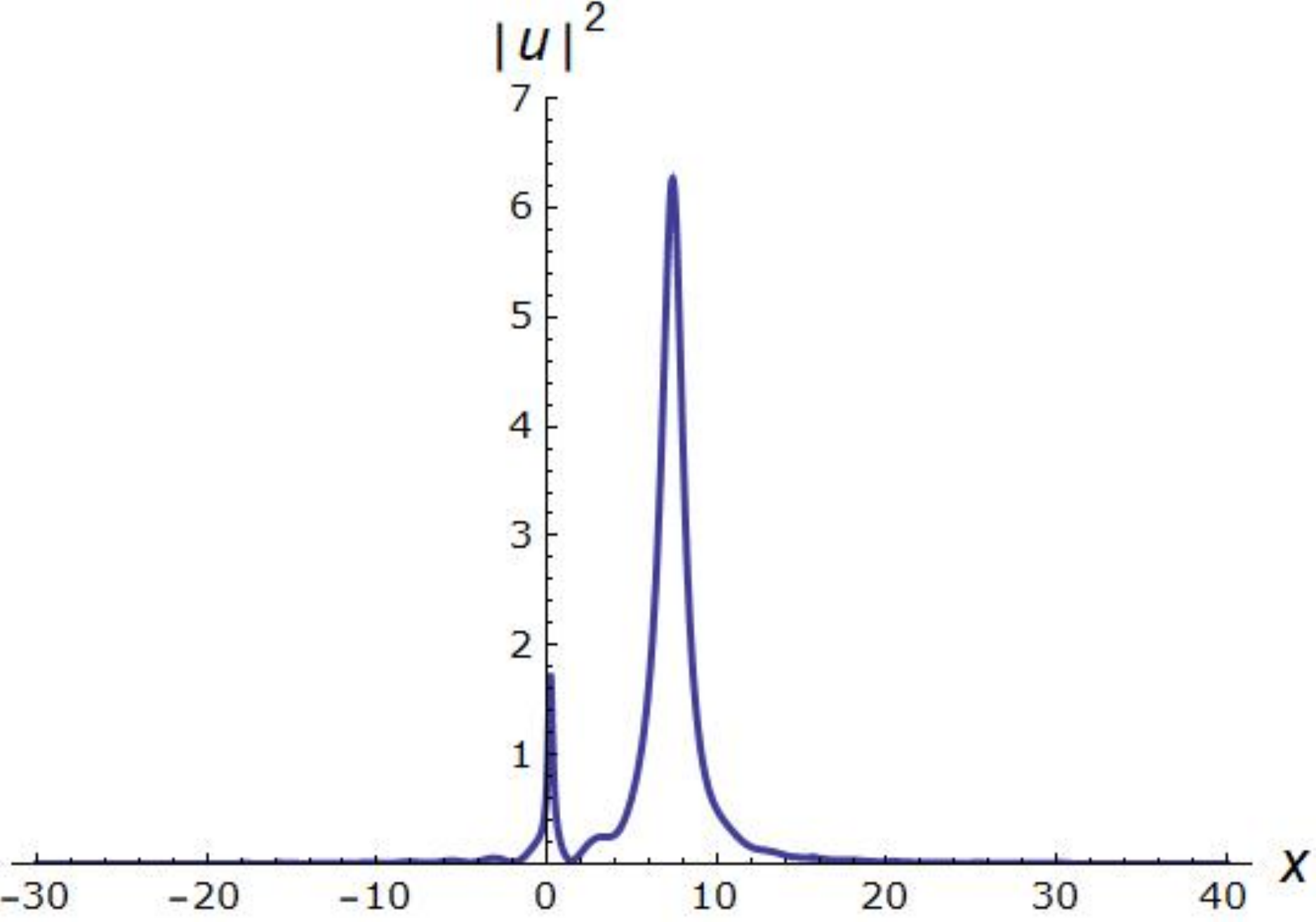}
\end{minipage}%
}%
\subfigure[ ]{
\begin{minipage}[t]{0.45\linewidth}
\centering
\includegraphics[width=2.1in]{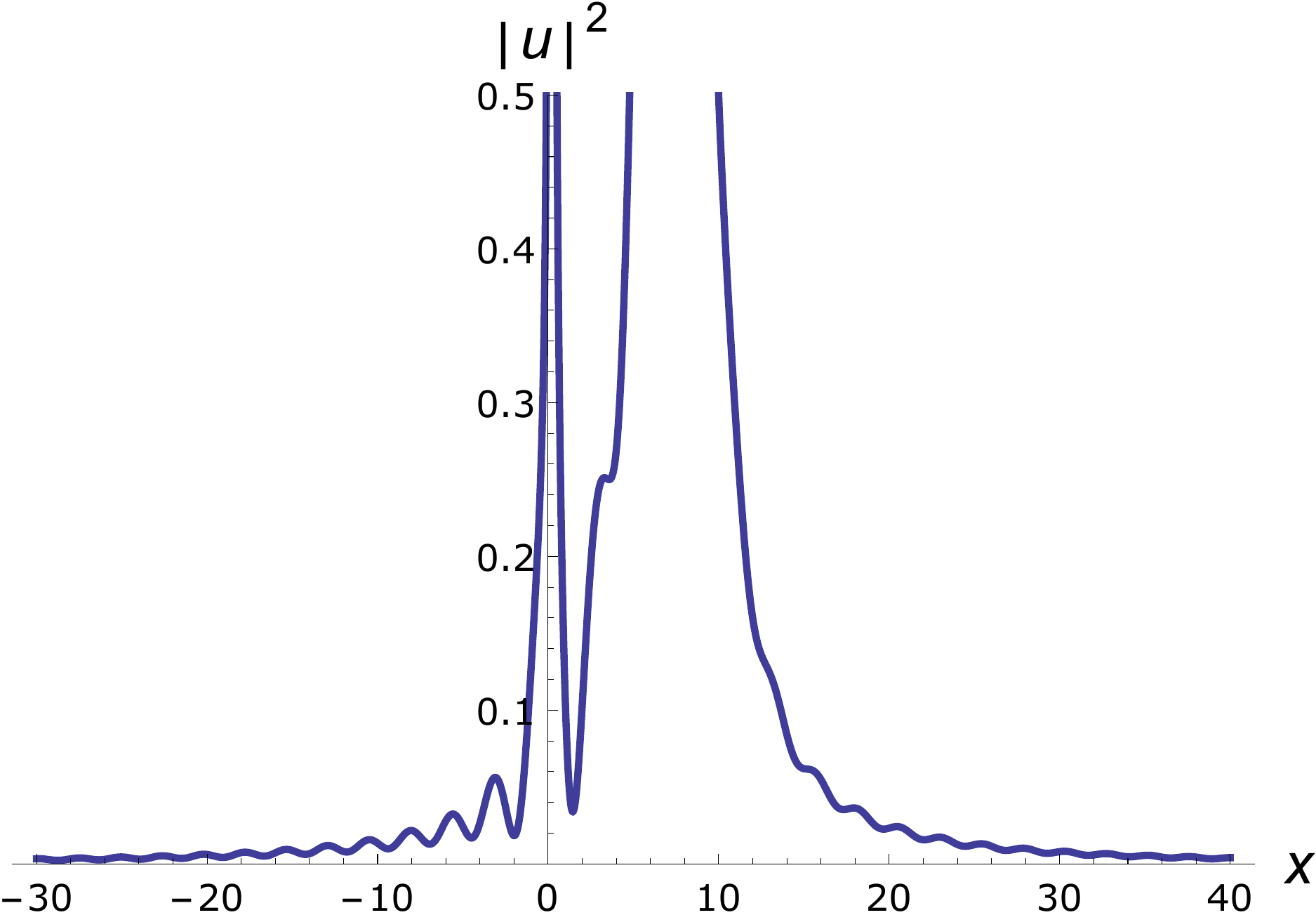}
\end{minipage}%
}
\caption{Jordan block solution of the FL equation \eqref{FL2}.
(a) $|u|^2$ given by \eqref{case2-jordon} with  $k_{1}=-1.8$, $h_{1}=0.8$ and $c_1=d_1=1$.~
(b) Density plot of (a) overlapped by the lines $X_1=0$ and $X_2=0$.~
(c) 2D plot of (a) at $t=3$.~ (d) An enlarged plot of (a) at $t=3$.}
\label{F-6}
\end{figure}

One may also consider mixed solutions resulted from diagonal $\mathbf{K}_N$ and Jordan block $\mathbf{H}_N$.
When $N=2$, $u$ is given by \eqref{2ss-pp} where
\begin{subequations}\label{aaa}
\begin{align}
& \phi=\left(c_1 e^{\eta(k_1)}, c_2 e^{\eta(k_2)},
d_1 e^{\eta(h_1)}, d_1\partial_{h_1} e^{\eta(h_1)}\right)^T,\\
& \psi=\left(c_{1}^{*}k_1e^{-\eta(k_1)}, c_{2}^{*}k_2 e^{-\eta(k_2)},
 -d_{1}^{*}h_1e^{-\eta(h_1)}, -d_{1}^{*}e^{-\eta(h_1)}-d_{1}^{*}h_1\partial_{h_1}e^{-\eta(h_1)}\right)^T,
\end{align}
\end{subequations}
and $k_j,h_j\in \mathbb{R}$, $c_j,d_j\in \mathbb{C}$.
$|u|^2$  is illustrated in Fig.\ref{F-7}, from which we can see one solitary wave is interacting with
a periodic wave.

\captionsetup[figure]{labelfont={bf},name={Fig.},labelsep=period}
\begin{figure}[ht]
\centering
\subfigure[ ]{
\begin{minipage}[t]{0.45\linewidth}
\centering
\includegraphics[width=2.3in]{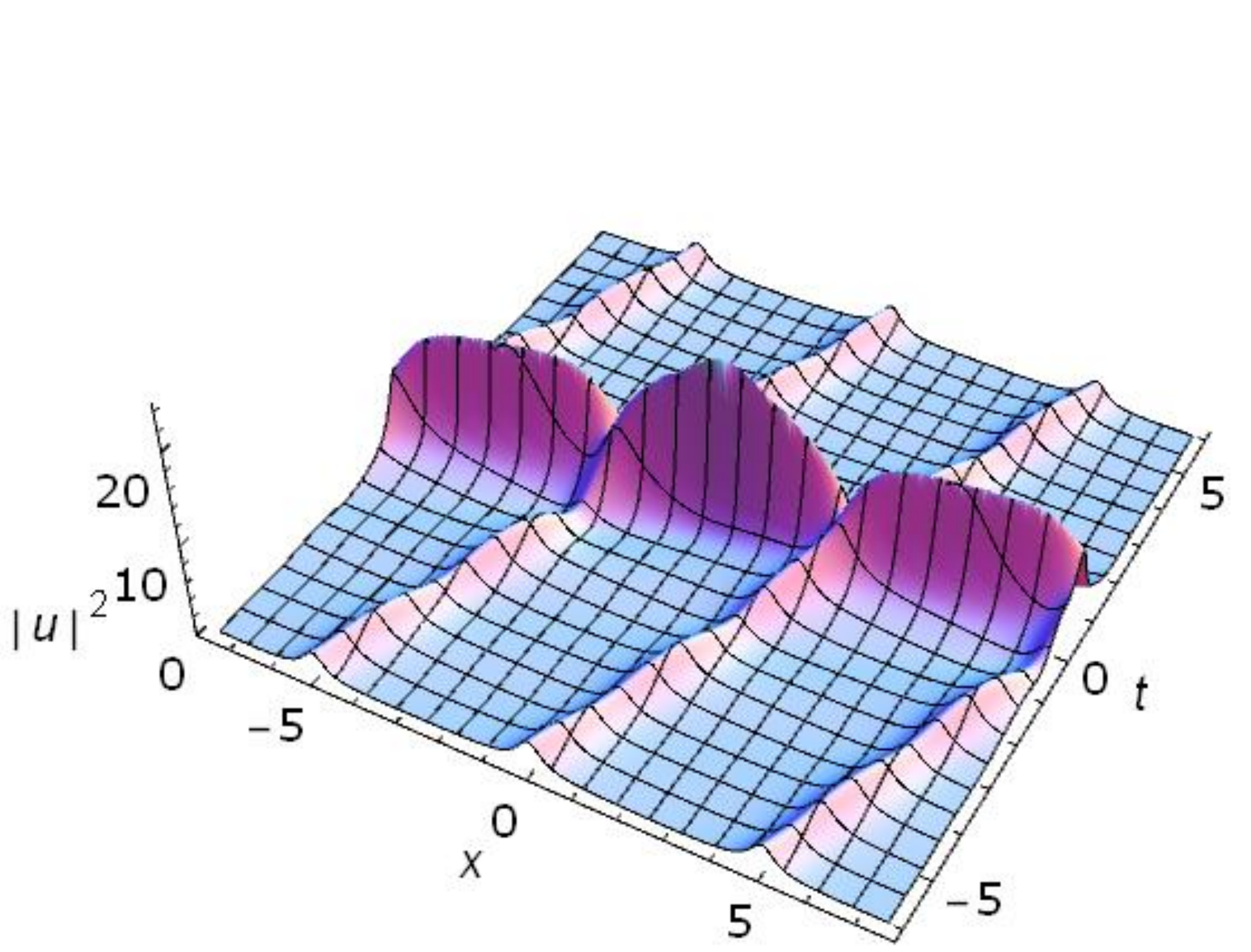}
\end{minipage}%
}%
\subfigure[ ]{
\begin{minipage}[t]{0.45\linewidth}
\centering
\includegraphics[width=1.5in]{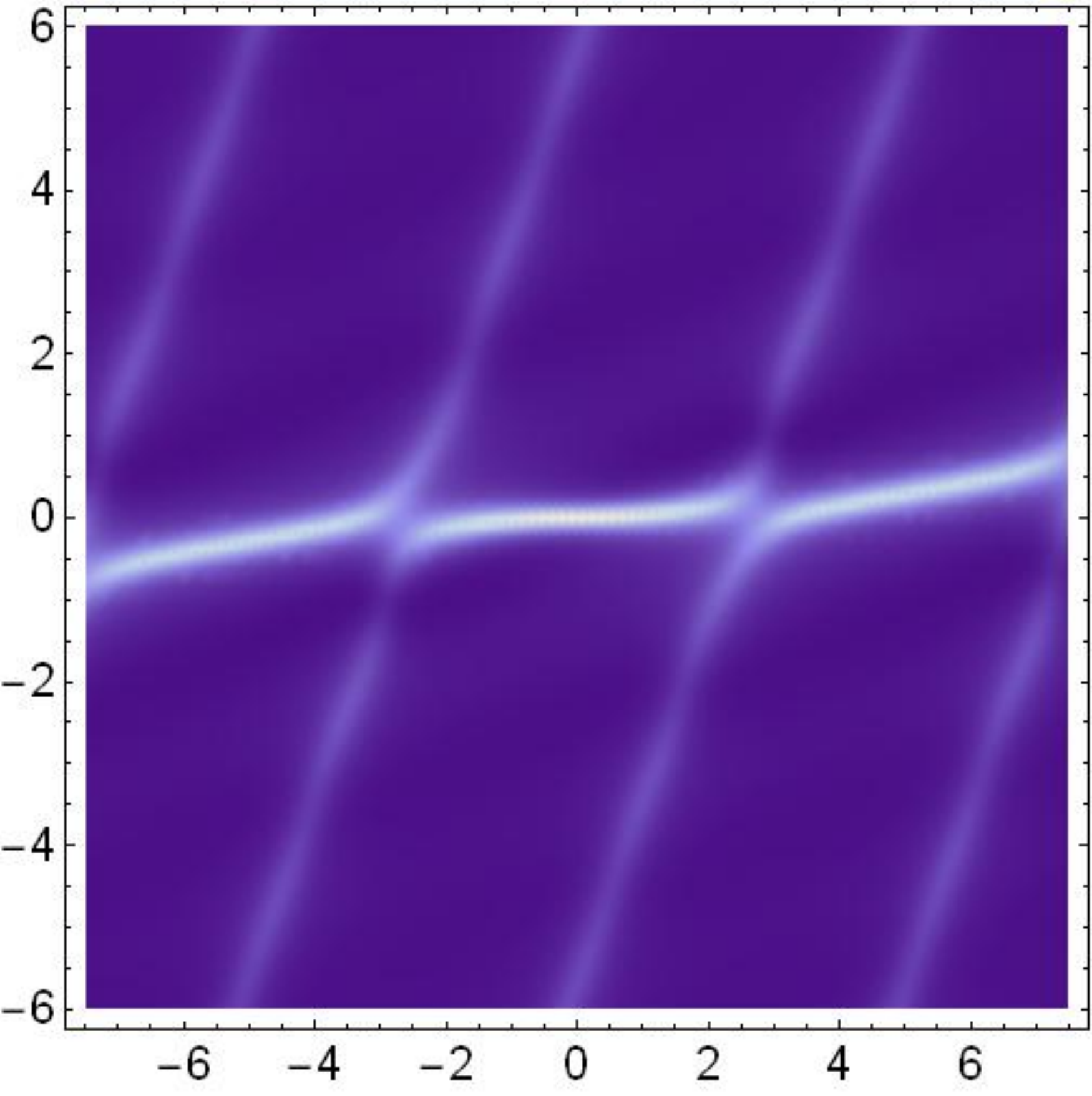}
\end{minipage}%
}
\caption{The mixed solution of the FL equation \eqref{FL2} in Case (2).
(a) $|u|^2$  given from \eqref{2ss-pp} with \eqref{aaa} and
$k_{1}=-1$, $k_{2}=1.5$, $h_{1}=0.5$, $c_{1}=d_{1}=c_{2}=d_{2}=1$.~
(b) Density plot of (a).}
\label{F-7}
\end{figure}

\section{Dynamics of the nonlocal FL equation \eqref{non-equv}}\label{sec-5}

In nonlocal case,  $\mathbf{K}_N$ and  $\mathbf{H}_N$ are complex matrices,
$\phi$ is given through \eqref{phi-pm} as described in Sec.\ref{sec-3-2-2}
and  $\psi$ takes the form \eqref{psi-phi-pmm-non}.
In the following we will mainly investigate 1SS with details,
while for 2SS we only list our formulae with figures as illustrations.
We only consider the case  $\delta=1$.
Besides, note that when $\mathbf{K}_N, \mathbf{H}_N \in  \mathbb{R}_N$,
it is possible for $\psi$ \eqref{psi-phi-pmm-non} to take the same form as the
$\psi$ \eqref{psi-phi-pmm}, and then the FL equation \eqref{FL2}
and nonlocal FL equation \eqref{non-equv} share the corresponding solutions.

\subsection{1SS}\label{sec-5-1}

1SS of the nonlocal FL equation \eqref{non-equv} is given by \eqref{1ss-fg} where
\begin{align*}
& \phi=(c_1 e^{\eta(k_1)}, d_1 e^{\eta(h_1)})^T, \\
& \psi=(c_1k_1 e^{-\eta(k_1)}, -d_1h_1 e^{-\eta(h_1)})^T,
\end{align*}
where $\eta(k)$ is given in \eqref{eta}, $k_1, h_1, c_1, d_1\in \mathbb{C}$.
The explicit formula is
\begin{equation}\label{1ss-nonlocal-FL}
u^{}_{1\mathrm{SS}}=\frac{k_{1}^{2}-h_{1}^{2}}
{k_{1}h_{1}\left[k_{1}\mathrm{e}^{-i(h_{1}^{2}x+\frac{t}{h_{1}^{2}})}
+h_{1}\mathrm{e}^{-i(k_{1}^{2}x+\frac{t}{k_{1}^{2}})}\right]}.
\end{equation}
The corresponding envelop is
\begin{equation}\label{1ss-non-FL-1}
|u_{\mathrm{1SS}}^{}|^2=
\frac{(a_{1}^{2}-b_{1}^{2}-m_{1}^{2}+s_{1}^{2})^{2}+4(a_{1}b_{1}-m_{1}s_{1})^{2}}
{2|k_1|^{3}|h_1|^{3} \, \mathrm{e}^{2W_{1}}\left[\cosh \left(2W_{2}+\ln \frac{|h_1| }{|k_1| }\right)
+ \sin(W_{3}+\omega_{1})\right]},
\end{equation}
where we have taken $k_j=a_j+ib_j$, $h_j=m_j+is_j$, $c_1=d_1=1$, and
\begin{align*}
&W_{1}=(a_{1}b_{1}+m_{1}s_{1})x
-\left(\frac{a_{1}b_{1}}{|k_1|^{4} }+\frac{m_{1}s_{1}}{|h_1|^{4}}\right) t,\\
&W_{2}=(a_{1}b_{1}-m_{1}s_{1})x
-\left(\frac{a_{1}b_{1}}{|k_1|^{4} }-\frac{m_{1}s_{1}}{|h_1|^{4}}\right) t, \\
&W_{3}=(a_{1}^{2}-b_{1}^{2}-m_{1}^{2}+s_{1}^{2})x
+\left(\frac{a_{1}^{2}-b_{1}^{2}}{|k_1|^{4}}
-\frac{m_{1}^{2}-s_{1}^{2}}{|h_1|^{4}}\right) t,\\
& \omega_{1}=\arctan \frac{a_{1}m_{1}+b_{1}s_{1}}{a_{1}s_{1}-b_{1}m_{1}}.
\end{align*}

In order to understand dynamics of \eqref{1ss-non-FL-1} in an analytic way,
let us first investigate when $W_j$ vanish for all $(x,t)$. It can be found that
$W_1\equiv 0$ if
\begin{subequations}
\begin{align}
 (a_1,b_1)=(\pm m_1, \mp s_1),~ \mathrm{or} ~ (a_1,b_1)=(\pm s_1, \mp m_1),\label{W1=0-a}
\end{align}
or
\begin{align}
 a_1b_1=m_1s_1=0,~\mathrm{but}~ |k_1||h_1|\neq 0; \label{W1=0-b}
\end{align}
\end{subequations}
$W_2\equiv 0$ if \eqref{W1=0-b} holds, or
\begin{equation}
 (a_1,b_1)=( m_1,  s_1),~ \mathrm{or} ~ (a_1,b_1)=( s_1,  m_1); \label{W2=0-a}
 \end{equation}
$W_3\equiv 0$ if
\begin{equation}
 (a_1^2,b_1^2)=( m_1^2,  s_1^2). \label{W3=0-a}
 \end{equation}

With these in hand, we may obtain desirable solutions by arranging real parameters $a,b,m,s$.
For example, if we take $a_1=s_1=0$ but $|k_1|\neq |h_1| \neq 0$ so that \eqref{W1=0-b} holds,
we get $W_{1}=W_{2}\equiv 0$ and
\begin{equation}\label{1ss-p-non}
|u|^2=\frac{ (b_1^2+m_1^2)^{2}}{b_{1}^{2}m_{1}^{2}
\left[b_1^2+m_1^2- 2|b_1m_1| \sin \left((b_1^2+m_1^2)(x+\frac{t}{b_1^2 m_1^2})\right)\right]},
\end{equation}
which is a nonsingular periodic wave (by virtue of $|k_1|\neq |h_1|$, i.e. $|k_1|\neq |h_1|$).
This wave is depicted in Fig.\ref{F-8}(a).

When $W_{1}\equiv W_{3}\equiv 0$ but $W_{2}\neq 0$, which can hold by taking, e.g. $(a_1,b_1)=(m_1,-s_1)$,
we get 1SS
\begin{equation}\label{1ss-s-non}
|u|^2=\frac{8a_1^2b_1^2}{(a_1^2+b_1^2)^2
(\cosh 4 a_1b_1 W'_2+ \sin \omega'_1)},
\end{equation}
where
$W'_2= x-\frac{t}{(a_1^2+b_1^2)^{2}}$ and $\omega'_1=\arctan\frac{b_1^2-a_1^2}{2a_1b_1}$,
which is depicted in Fig.\ref{F-8}(b).
\captionsetup[figure]{labelfont={bf},name={Fig.},labelsep=period}
\begin{figure}[h]
\centering
\subfigure[ ]{
\begin{minipage}[t]{0.45\linewidth}
\centering
\includegraphics[width=2.1in]{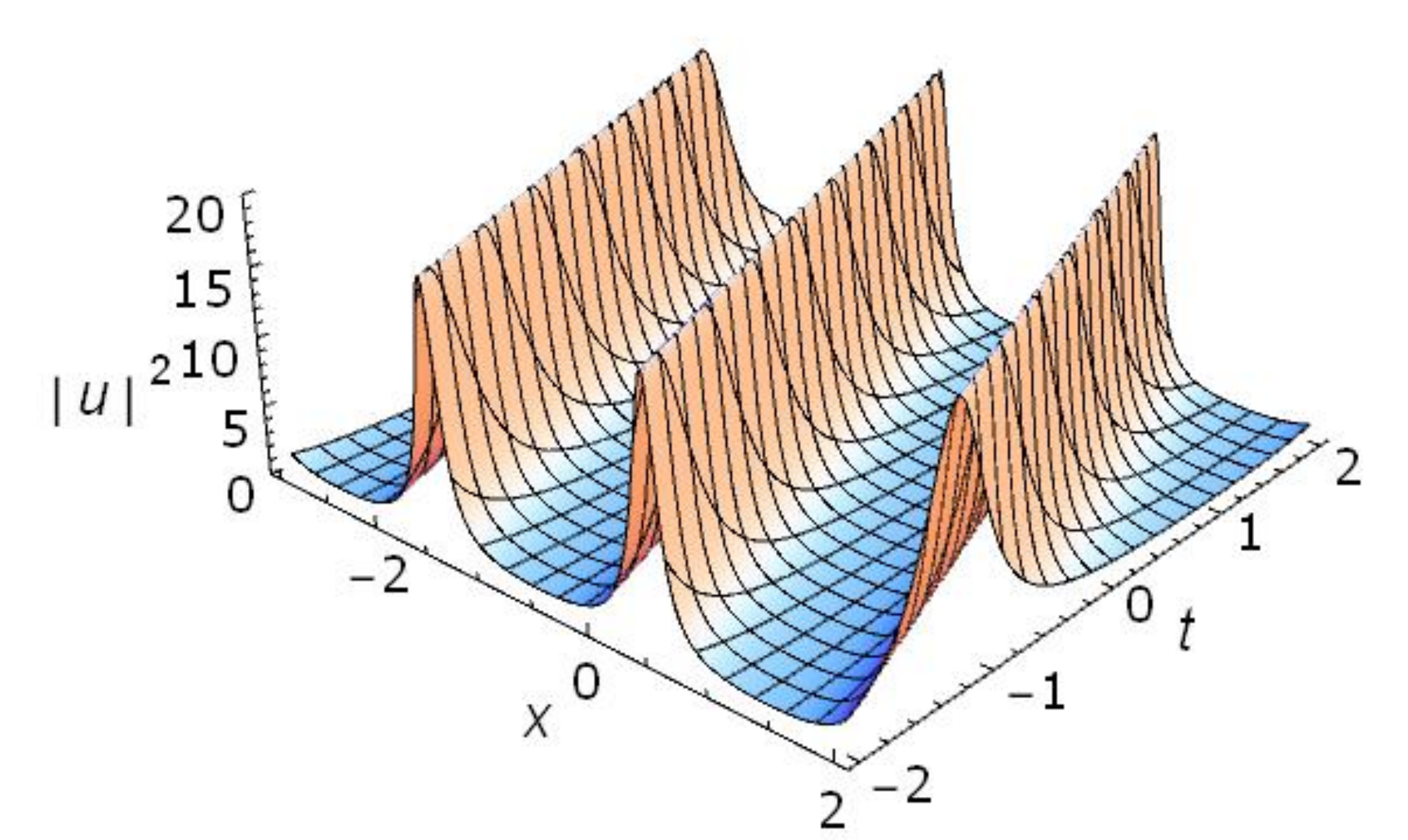}
\end{minipage}%
}%
\subfigure[ ]{
\begin{minipage}[t]{0.45\linewidth}
\centering
\includegraphics[width=2.2in]{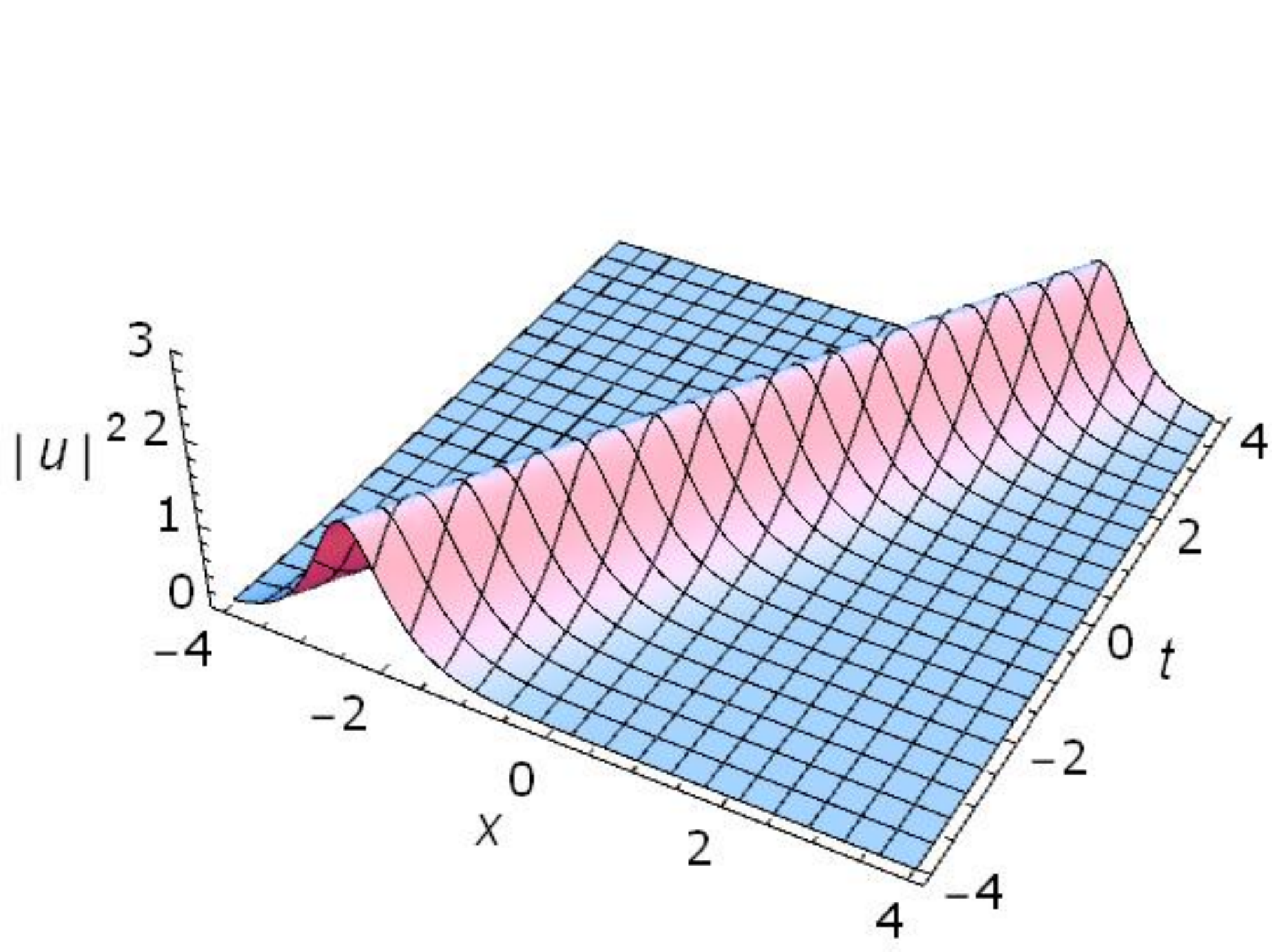}
\end{minipage}%
}%
\caption{Shape and motion of 1SS of the nonlocal FL equation \eqref{non-equv}. ~
(a) Periodic wave given by \eqref{1ss-p-non} with $k_1=i$ and $h_1=1.5$.~
(b) Soliton given by \eqref{1ss-s-non} with $k_1=0.8+0.8i$ and $h_1=0.8-0.8i$.   }
\label{F-8}
\end{figure}

There can have kink-type waves but always with singularities.
Considering the case that there is only one number being zero among $(a_1,b_1,m_1,s_1)$,
e.g. only $m_1=0$, i.e.
\begin{equation}\label{abms}
(a_1,b_1,m_1,s_1)=(a_1,b_1,0,s_1),~ \mathrm{and} ~ a_1b_1s_1\neq 0,
\end{equation}
we have
\begin{subequations}\label{W1-omega}
\begin{align}
& W_1=W_2=a_1b_1\left(x-\frac{t}{|k_1|^4}\right),\\
& W_3=(a_1^2-b_1^2+s_1^2)x+\left(\frac{a_1^2-b_1^2}{|k_1|^4}+\frac{1}{s_1^2}\right)t,~~
\omega_1=\arctan\frac{b_1}{a_1}.
\end{align}
\end{subequations}
It is easy to check that the slopes of lines $W_1=0$ and $W_3=0$ can never be same
in light of \eqref{abms}.
In this case, \eqref{1ss-non-FL-1} turns out to be
\begin{equation}\label{1ss-k-non}
|u_{\mathrm{1SS}}^{}|^2=
\frac{(a_{1}^{2}-b_{1}^{2}+s_{1}^{2})^{2}+4(a_{1}b_{1})^{2}}
{|k_1|^{2}|h_1|^{2} (|h_1|^2 y^2+ |k_1|^2+2 y |k_1||h_1|\sin z)},
\end{equation}
where
\[y=e^{2W1},~ z=W_3+\omega_1,\]
and $W_1, ~\omega_1$ take the forms in \eqref{W1-omega}.
This is a kink-type wave for any given $t$: when $x\to \pm\infty$,
$|u|^2$ goes to zero on one side and
$\frac{(a_{1}^{2}-b_{1}^{2}+s_{1}^{2})^{2}+4(a_{1}b_{1})^{2}}{|k_1|^{4}|h_1|^{2}}$
on the other side, or the other way around, depending on sgn$[a_1b_1]$.
However, there are infinitely many poles appearing at the intersections
\[\left\{\begin{array}{l}
W_1=\frac{1}{2}\ln \frac{|k_1|}{|h_1|},\\
W_3+\omega_1=2j\pi-\frac{\pi}{2},~~ j\in \mathbb{Z},
\end{array}\right.
\]
and all poles are located at the line $W_1=\frac{1}{2}\ln \frac{|k_1|}{|h_1|}$.
Such a solution is illustrated in Fig.\ref{F-9}.

Note that some 1SS of the nonlocal FL equation \eqref{non-equv} have been explored in
\cite{ZhangY-AML} using Darboux transformation.

\captionsetup[figure]{labelfont={bf},name={Fig.},labelsep=period}
\begin{figure}[h]
\centering
\begin{minipage}[t]{0.45\linewidth}
\centering
\includegraphics[width=2.3in]{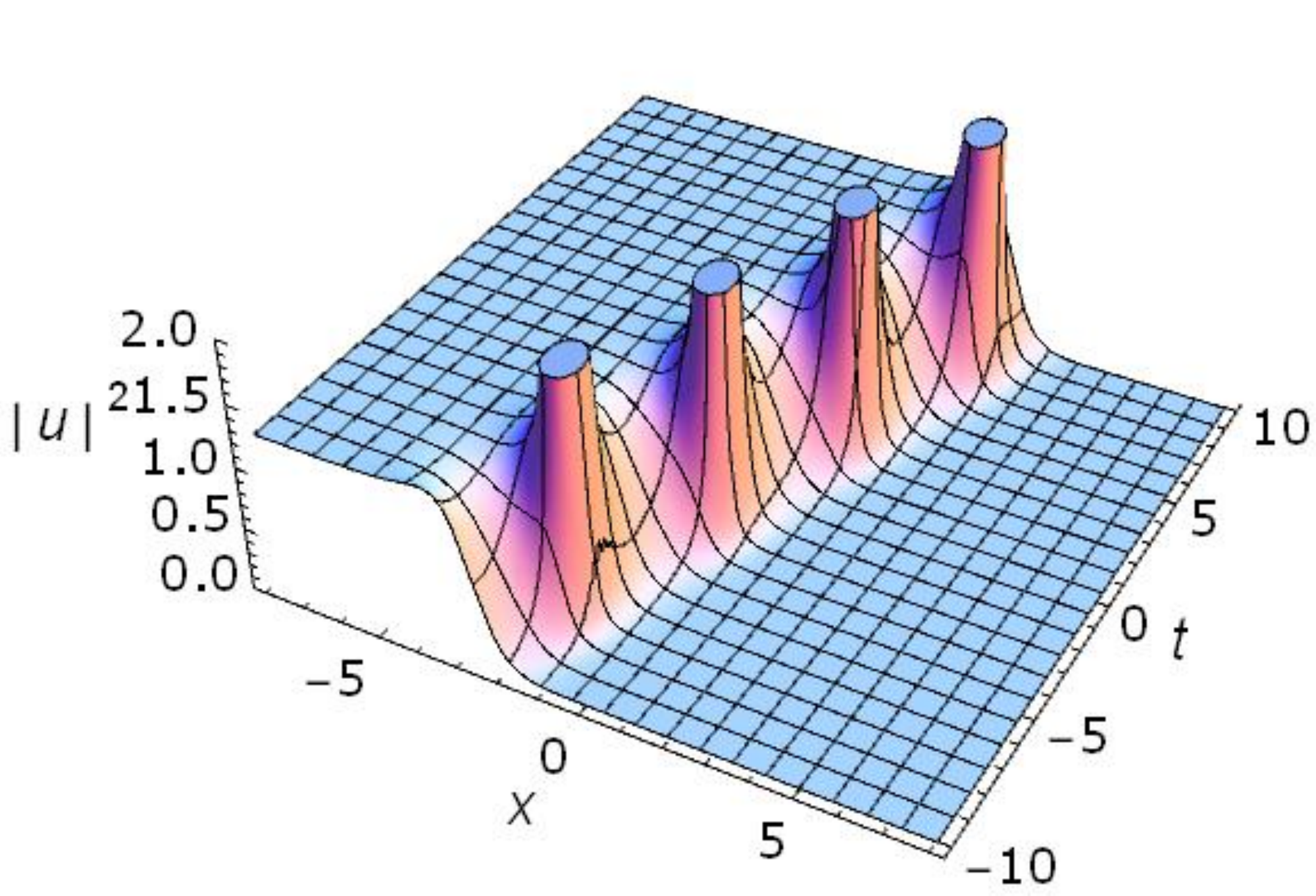}
\end{minipage}%
\caption{ Shape and motion of a kink-type wave of the nonlocal FL equation \eqref{non-equv},
given by \eqref{1ss-k-non}  with $k_1=1+i$ and $h_1=i$.}
\label{F-9}
\end{figure}

\subsection{2SS}\label{sec-5-2}

The above analysis we have made for 1SS is helpful to understand two-soliton interactions.
In the following we only list out illustrations.

2SS is given via \eqref{2ss-pp} where in nonlocal case, when both $\mathbf{K}_N$ and  $\mathbf{H}_N$
are diagonals, we have
\begin{subequations}\label{a}
\begin{align}
& \phi=(c_1e^{\eta(k_1)}, c_2e^{\eta(k_2)}, d_1e^{\eta(h_1)}, d_2e^{\eta(h_2)})^T,\\
& \psi=(k_1c_1e^{-\eta(k_1)}, k_2c_2e^{-\eta(k_2)}, -h_1d_1e^{-\eta(h_1)}, -h_2 d_2e^{-\eta(h_2)})^T;
\end{align}
\end{subequations}
when both $\mathbf{K}_N$ and $\mathbf{H}_N$ are Jordan blocks, we have
\begin{subequations}\label{b}
\begin{align}
& \phi=\left(c_1 e^{\eta(k_1)}, c_1 \partial_{k_1} e^{\eta(k_1)},
d_1 e^{\eta(h_1)}, d_1\partial_{h_1} e^{\eta(h_1)}\right)^T,\\
& \psi\!=\!\left(k_1c_1e^{-\eta(k_1)}, c_1 e^{-\eta(k_1)}\!+\!c_1 k_1 \partial_{k_1} e^{-\eta(k_1)},
 -h_1d_1e^{-\eta(h_1)},- d_1e^{-\eta(h_1)}\!-\!h_1 d_1\partial_{h_1} e^{-\eta(h_1)}\right)^T;
\end{align}
\end{subequations}
and when  $\mathbf{K}_N$ is diagonal and $\mathbf{H}_N$ is a Jordan block, we have
\begin{subequations}\label{c}
\begin{align}
& \phi=\left(c_1 e^{\eta(k_1)}, c_2 e^{\eta(k_2)},
d_1 e^{\eta(h_1)}, d_1\partial_{h_1} e^{\eta(h_1)}\right)^T,\\
& \psi=\left(k_1c_1e^{-\eta(k_1)}, k_2 c_2 e^{-\eta(k_2)},
 -h_1d_1e^{-\eta(h_1)},- d_1e^{-\eta(h_1)}-h_1 d_1\partial_{h_1} e^{-\eta(h_1)}\right)^T.
\end{align}
\end{subequations}
Here  $k_j, h_j, c_j, d_j\in \mathbb{C}$,
and we have taken the LTTMs  $\mathcal{A}_N=\mathcal{B}_N=\mathbf{I}_N$.

Two-soliton interactions are illustrated in Fig.\ref{F-10} and  Fig.\ref{F-11},
from which we can see that, compared with classical case, the interactions of 2SS
are more complicated in nonlocal case (see also the nonlocal Gross-Pitaevskii equation
\cite{Liu-ROMP-2020}).

We also remark that, in principle, the $N$-soliton solutions obtained in this paper for the nonlocal FL equation
coincide with those obtained from Darboux transformation \cite{ZhangY-AML} with zero as a seed solution.
However, as we have shown in this subsection, the independency of $\mathbf{K}_N$ and $\mathbf{H}_N$
allows more variety in multisoliton and multiple-pole solutions.

\captionsetup[figure]{labelfont={bf},name={Fig.},labelsep=period}
\begin{figure}[h]
\centering
\subfigure[ ]{
\begin{minipage}[t]{0.40\linewidth}
\centering
\includegraphics[width=2.1in]{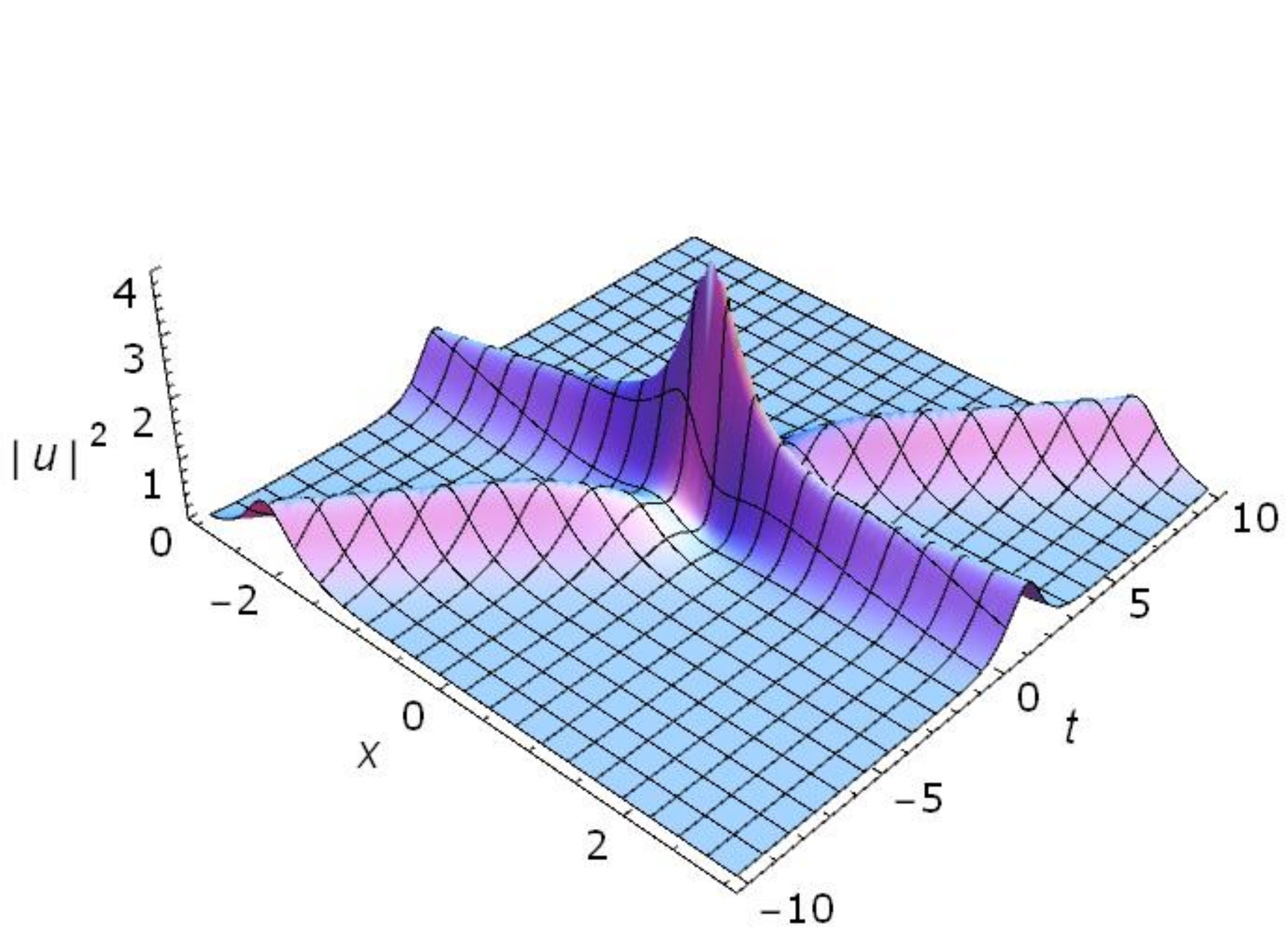}
\end{minipage}%
}%
\subfigure[ ]{
\begin{minipage}[t]{0.40\linewidth}
\centering
\includegraphics[width=2.1in]{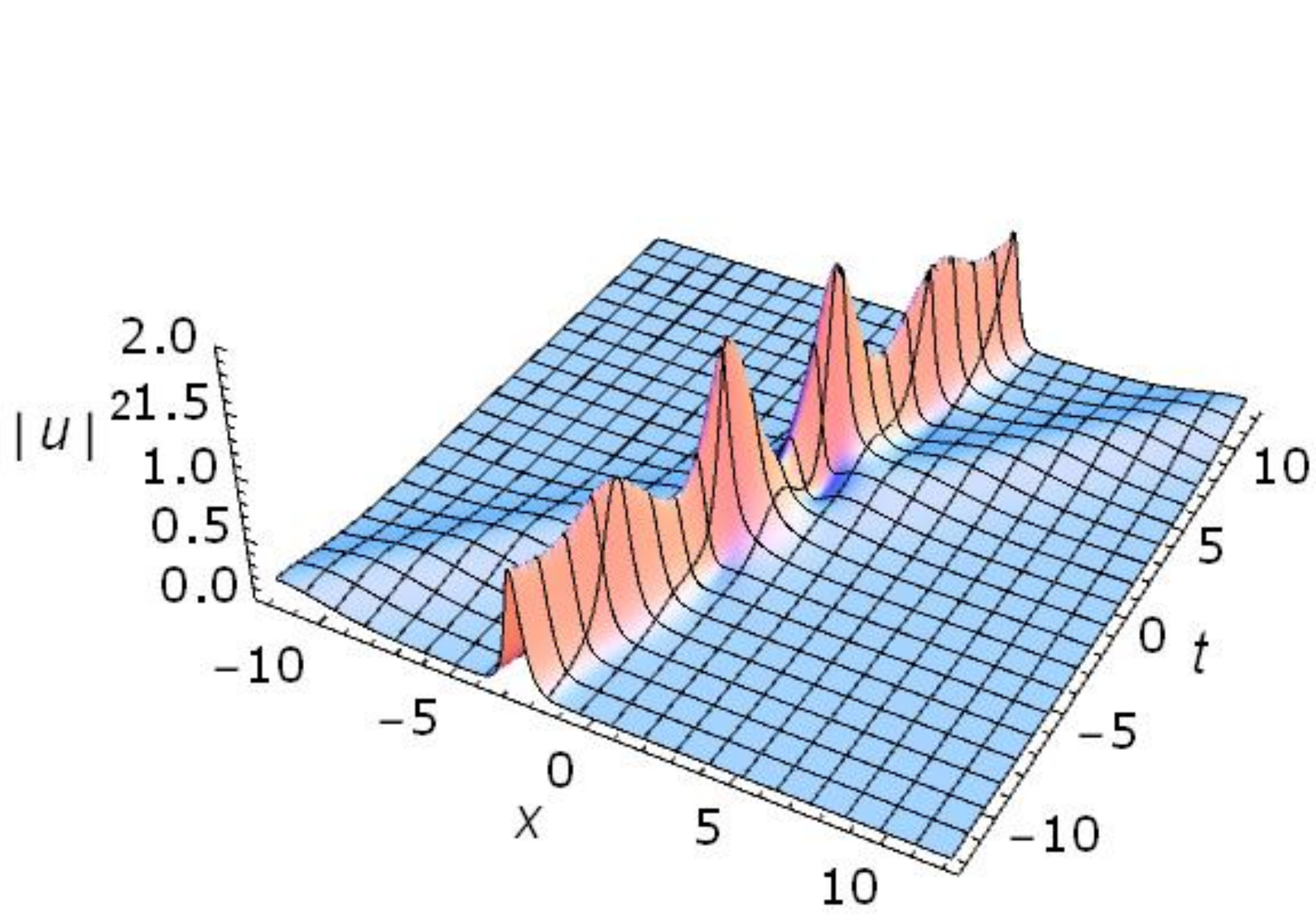}
\end{minipage}%
}\\
\subfigure[ ]{
\begin{minipage}[t]{0.40\linewidth}
\centering
\includegraphics[width=2.1in]{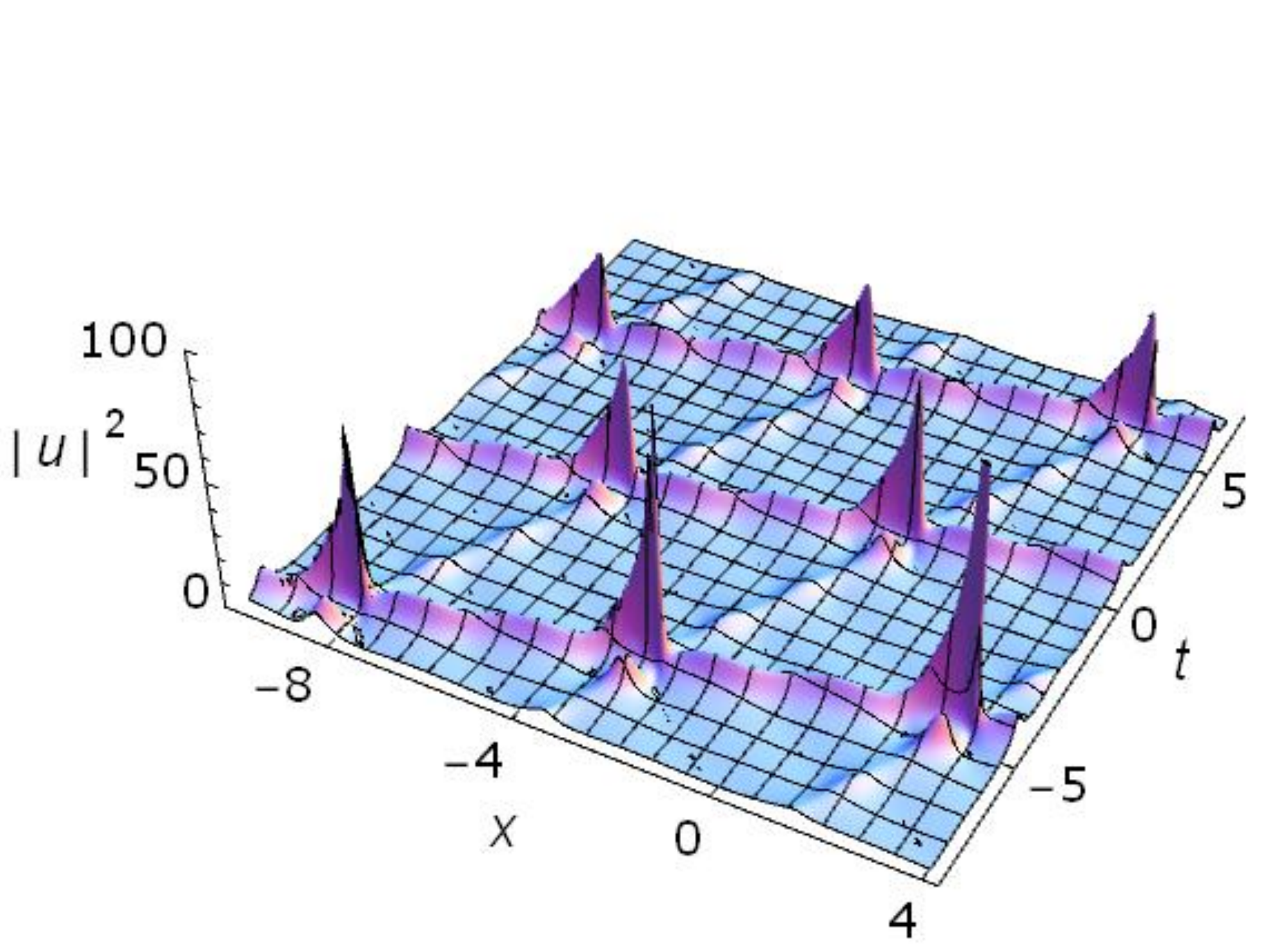}
\end{minipage}%
}%
\subfigure[ ]{
\begin{minipage}[t]{0.40\linewidth}
\centering
\includegraphics[width=2.1in]{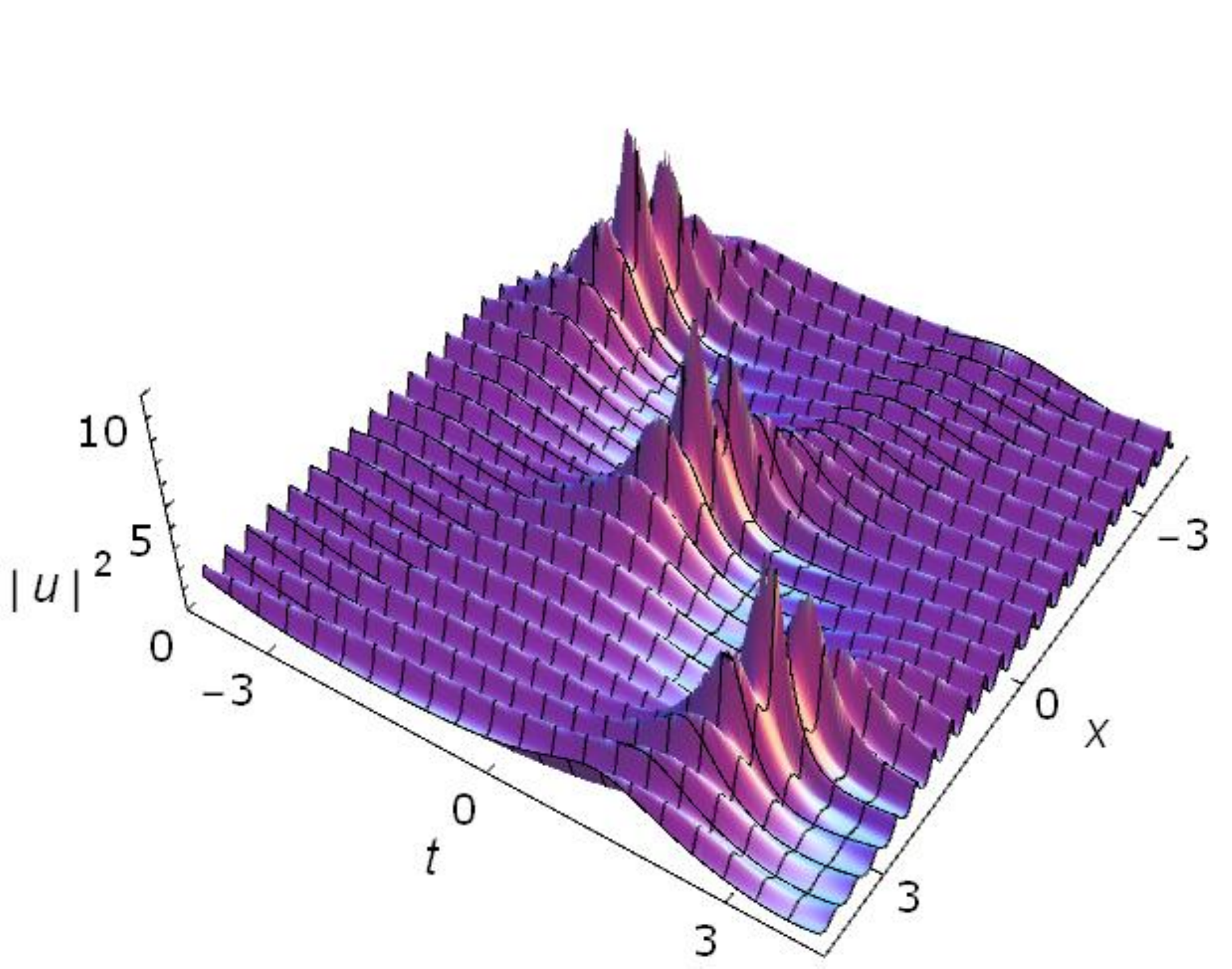}
\end{minipage}%
}
\caption{Shape and motion of 2SS of the nonlocal FL equation \eqref{non-equv}.~
(a) Envelope $|u|^2$ of \eqref{2ss-pp} with \eqref{a} in which $k_1=1+i$, $h_1=1-i$,  $k_2=0.8+0.4i$,
$h_2=0.8-0.4i$ and $c_1=c_2=d_1=d_2=1$. ~
(b) Envelope $|u|^2$ of \eqref{2ss-pp} with \eqref{a} in which $k_1=1+i$, $h_1=1-i$,  $k_2=1+0.2i$,
$h_2=1-0.2i$ and $c_1=c_2=d_1=d_2=1$. ~
(c) Envelope $|u|^2$ of \eqref{2ss-pp} with \eqref{a} in which $k_1=0.6i$, $h_1=-1.5$,  $k_2=0.8i$,
$h_2=-1$ and $c_1=c_2=d_1=d_2=1$.
~
(d) Envelope $|u|^2$ of \eqref{2ss-pp} with \eqref{a} in which $k_1=0.8+0.6i$, $h_1=0.8-0.6i$,  $k_2=i$,
$h_2=-4$ and $c_1=c_2=d_1=d_2=1$.
}
\label{F-10}
\end{figure}

\captionsetup[figure]{labelfont={bf},name={Fig.},labelsep=period}
\begin{figure}[h]
\centering
\subfigure[ ]{
\begin{minipage}[t]{0.40\linewidth}
\centering
\includegraphics[width=2.1in]{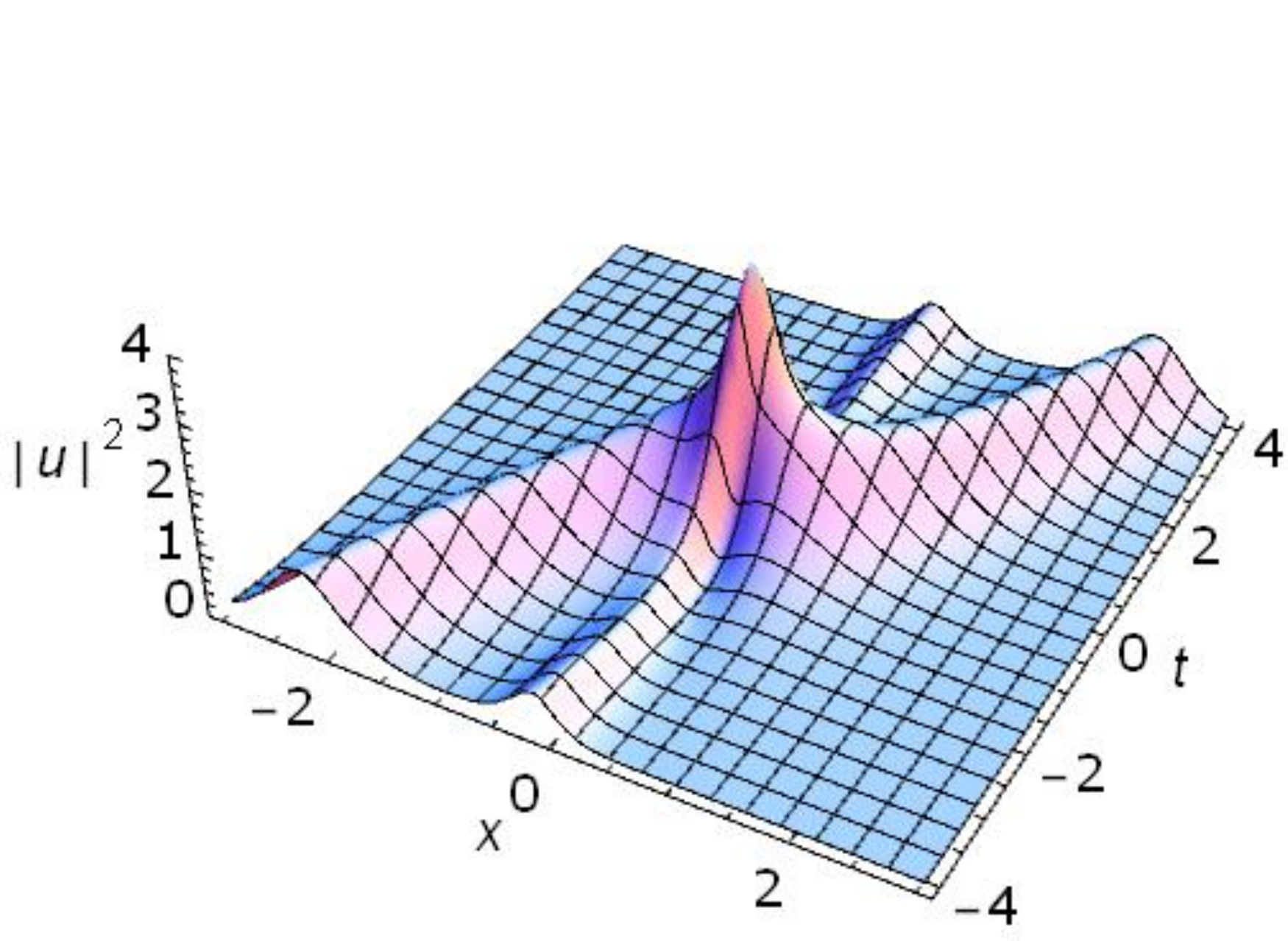}
\end{minipage}%
}%
\subfigure[ ]{
\begin{minipage}[t]{0.40\linewidth}
\centering
\includegraphics[width=2.1in]{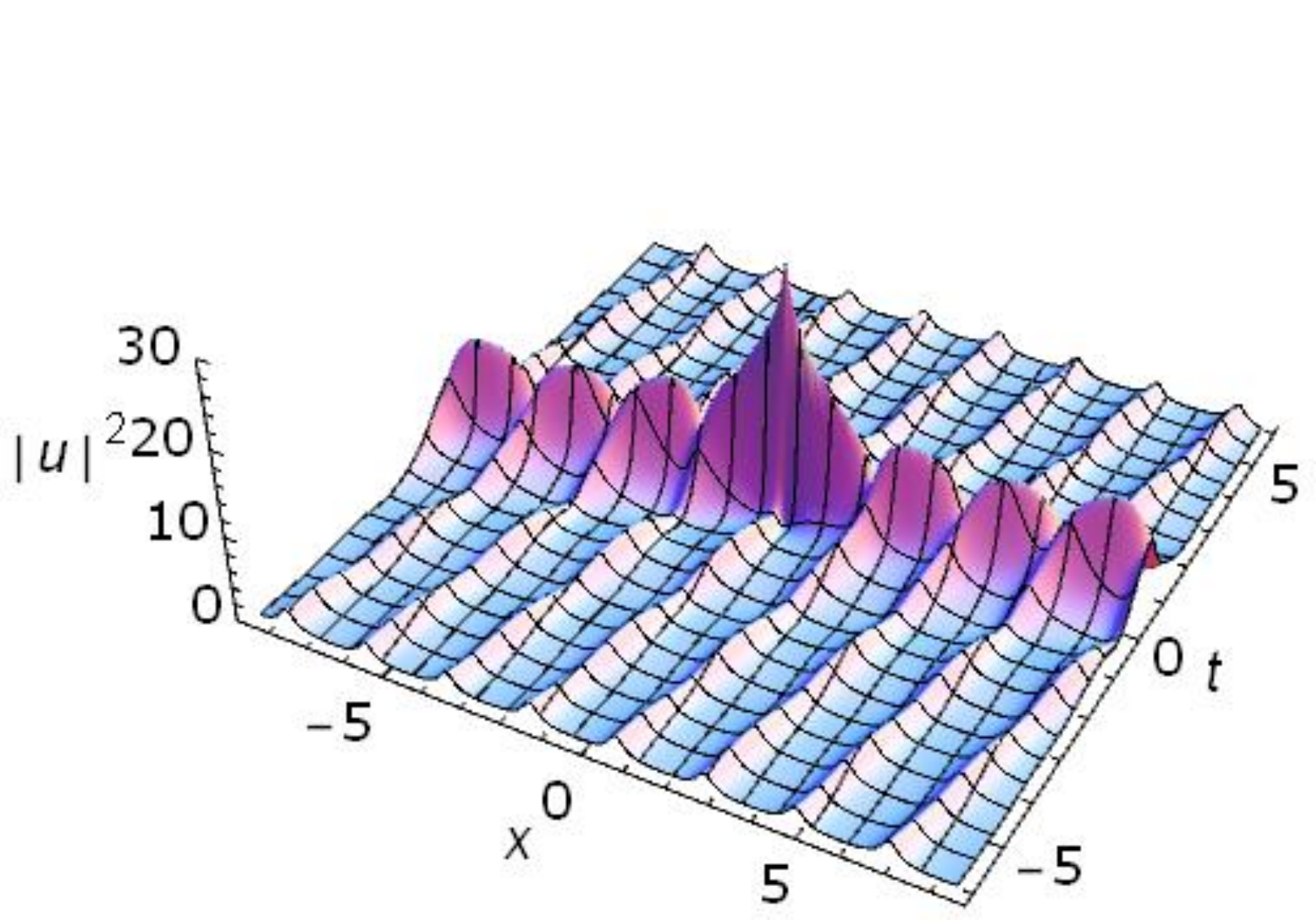}
\end{minipage}%
}%
\caption{Shape and motion of 2SS of the nonlocal FL equation \eqref{non-equv}.~
(a) Envelope $|u|^2$ of \eqref{2ss-pp} with \eqref{b} in which $k_1=1+i$, $h_1=1-i$ and $c_1=d_1=1$. ~
(b) Envelope $|u|^2$ of \eqref{2ss-pp} with \eqref{c} in which $k_1=1$, $k_2=2$, $h_1=-0.5$
and $c_1=c_2=d_1 =d_2=1$. }
\label{F-11}
\end{figure}

\section{Concluding remakks}\label{sec-6}

We have derived solutions for the classical FL equation \eqref{FL2} and nonlocal FL equation \eqref{non-equv}
from bilinear approach.
We introduced new double Wronskian expressions \eqref{wronskian-1} that are different from
those of the AKNS hierarchy, the KN equation and the Chen-Lee-Liu equation
(cf.\cite{ChenDLZ-SAPM-2018,KakeiSS-JPSJ-1995,ZhaiC-PLA-2008}).
The assumption \eqref{wron-cond-x} with a general $A$ for Wronskian entries
and the reduction technique enable us to have a full profile
for the solutions of the FL equations.
One- and two-soliton solutions were illustrated based on analysis in detail.
It is notable that the FL equation \eqref{FL2} also allows solutions related to real discrete eigenvalues.
They exhibit (multi-)periodic behavior for distinct eigenvalues and algebraic decayed solitary waves
for those eigenvalues with multiplicity two. The later case was not
found before in the analytic approaches (e.g. \cite{Lenells-F-Non-2009,Ai-Xu-AML-2019,Zhao-F-JNMP-2021})
that are based on analyzing analytic domains of wave functions.

Before Fokas and Lenells, the FL equation \eqref{FL2} was already  explored around 40 years ago
(cf.\cite{GIK-1980,NCQV-1983}),
as it could generate solutions to the massive Thirring model arose in relativistic quantum field theory.
In Appendix \ref{app-0} we will recall the links between the FL equation \eqref{FL2} and the
massive Thirring model. As a result, all the solutions we obtained for the FL equation \eqref{FL2} can
generate solutions to the massive Thirring model (see Theorem \ref{Theorem-app} in Appendix \ref{app-0}).

In this paper solutions are presented in terms of (double) Wronskians.
Solutions in this form and similar forms are usually obtained via Darboux transformations
or bilinear method.
Compared with other popular forms of $N$-soliton solutions (e.g. Hirota's form using polynomials of
exponential functions given in Appendix \ref{app-A}
and dressed Cauchy matrix form obtained in \cite{Lenells-JNS-2010,Wang-XL-AML-2020,Matsuno-JPA-2012a}),
by virtue of their special  structure,
Wronskian solutions have advantage in presenting limit solutions, i.e. multiple-pole solutions,
In bilinear approach, such solutions are alternatively obtained by taking, for example,
$A$ in \eqref{wron-cond-x} to be composed of Jordan blocks.
One may refer to \cite{MatS-Book-1991} (see page 22 of the book) for the limit procedure
and to \cite{ZDJ-arxiv} for the connections between the LTTMs and limit solutions.
For the multiple-pole solutions in terms of dressed Cauchy matrix form, one may refer to
\cite{SY-2003} for the Riemann-Hilbert method and to \cite{ZhangZ-SAPM-2013}
for the Cauchy matrix approach.
In addition, by employing double Wronskians,  coefficient matrix $A$  and bilinear approach,
we have also illustrated an effective reduction technique to obtain solutions for the reduced equations.
In this technique, looking for vectors $\phi$ and $\psi$ such that $u$ and $v$ satisfy desired constraints
when they are expressed in terms of double Wronskians  $f,g,h$ and $s$, is boiled down to solving the
algebraic equation \eqref{TA} (and \eqref{TA2} for nonlocal case).
This enables us to approach to new solutions that might be missed before,
for example, the case (2) in Table \ref{tab-1} that corresponds to real eigenvalues and the mixed case related to
\eqref{TA-mix}.

Finally, as remarks we list several possible interesting questions arising from the current paper.
The first is to reinvestigate the coupled KN equation \eqref{KN2} and the DNLS equation \eqref{DNLS}
using the double Wronskian structure given in \eqref{wronskian-1}. Since the couple system \eqref{CFL}
is the potential form of the KN$(-1)$ system \eqref{KN-1},
it is possible to get solutions for the  coupled KN equation \eqref{KN2} from $q=(g/f)_x, ~r=(h/s)_x$ after
redefining $\phi, \psi$ with the dispersion relation of \eqref{KN2}.
Similar treatment was done in \cite{Lenells-JNS-2010}.
The second is to investigate the FL equations with nonzero backgrounds from bilinear approach.
Note that the assumption \eqref{wron-cond-x} corresponds to $q=r=0$ in the Lax pair.
In addition, it would be interesting to reinvestigate possible analytic domains of wave functions
for the Cauchy problem where $|u(x,t=0)|^2$ is algebraic decayed as $|x|\to \infty$.
This was not touched in \cite{Lenells-F-Non-2009,Ai-Xu-AML-2019,Zhao-F-JNMP-2021}.
Finally, solving the nonlocal FL equation from an analytic approach is also an interesting
problem. Note that the analysis in Sec.\ref{sec-5-1} implies
soliton solutions arise from the eigenvalue distribution  $\mathbf{H}_N=\mathbf{K}_N^*$
or $\mathbf{H}_N=-\mathbf{K}_N^*$.
However, interactions of 2SS exhibit more varieties in nonlocal case.

\subsection*{Acknowledgments}

This project is  supported by the NSF of China (Nos.11875040 and 11631007).

\begin{appendix}

\section{The massive Thirring model, KN spectral problem and the FL equation}\label{app-0}

The pKN$(-1)$ \eqref{CFL} is called the Mikhailov model by Gerdjikov and his collaborators \cite{GIK-1980,GI-1982}.
Its reduction gives rise to the FL equation \eqref{FL2}.
The latter provides solutions to the
massive Thirring model,
which describes the theory of a massive fermion field coupled to a two-component vector field
interacting with itself via a Fermi interaction \cite{T-1958,W-1967}.

The two-dimensional massive Thirring model is \cite{T-1958,W-1967,M-1976}
\begin{equation}
(-i\partial_{\mu}\gamma^{\mu}+m)\mathcal{X}
+g\gamma^{\mu}\mathcal{X} (\overline{\mathcal{X}}\gamma_{\mu}\mathcal{X})=0,
\end{equation}
where $m$ stands for mass, $g$ is a parameter,
$\mathcal{X}=(\mathcal{X}_1,\mathcal{X}_2)^T$,
$\gamma_0=\Bigl(\begin{array}{cc}0 & 1 \\ 1 & 0\end{array}\Bigr)$,
$\gamma_1=\Bigl(\begin{array}{cc}0 & -1 \\ 1 & 0\end{array}\Bigr)$,
$\overline{\mathcal{X}}=\mathcal{X}^{\dag}\gamma^{0}=(\mathcal{X}^*_1,\mathcal{X}^*_2)\gamma^{0}$,
$\gamma_{\mu}=(\gamma^{\mu})^{-1}$ and the Einstein summation convention is used.
Denoting $\partial_\mu=\partial_{x_\mu}$,
the above equation is written as (with $m=2, g=1$)
\begin{subequations}\label{MTM}
\begin{align}
& -i(\partial_{x_0}+\partial_{x_1})\mathcal{X}_1+2\mathcal{X}_2+2|\mathcal{X}_2|^2\mathcal{X}_1=0,\\
& -i(\partial_{x_0}-\partial_{x_1})\mathcal{X}_2+2\mathcal{X}_1+2|\mathcal{X}_1|^2\mathcal{X}_2=0.
\end{align}
\end{subequations}
It is Mikhailov \cite{M-1976} who first gave a Lax pair of the massive Thirring model, for \eqref{MTM} which reads
\cite{M-1976,KM-1977}
\begin{subequations}\label{MTM-Lax}
\begin{align}
& \Phi_{x_0}=M_0\Phi,~~
M_0=\frac{i}{2}\left(\begin{array}{cc}
         -|\mathcal{X}_1|^2+|\mathcal{X}_2|^2-\lambda^2+\lambda^{-2}
         & 2\lambda\mathcal{X}_2^*-2\lambda^{-1}\mathcal{X}^*_1\\
         2\lambda\mathcal{X}_2-2\lambda^{-1}\mathcal{X}_1
         & |\mathcal{X}_1|^2-|\mathcal{X}_2|^2+\lambda^2-\lambda^{-2}
         \end{array}\right),\\
& \Phi_{x_1}=M_1\Phi,~~
M_1=\frac{i}{2}\left(\begin{array}{cc}
         |\mathcal{X}_1|^2+|\mathcal{X}_2|^2-\lambda^2-\lambda^{-2}
         & 2\lambda\mathcal{X}_2^*+2\lambda^{-1}\mathcal{X}^*_1\\
         2\lambda\mathcal{X}_2+2\lambda^{-1}\mathcal{X}_1
         & -|\mathcal{X}_1|^2-|\mathcal{X}_2|^2+\lambda^2+\lambda^{-2}
         \end{array}\right),
\end{align}
\end{subequations}
where $\lambda$ is a spectral parameter.
In light-cone coordinates $(x,t)=(x_0+x_1, x_0-x_1)$, the equation \eqref{MTM} and its Lax pair are written as
\begin{subequations}\label{MTM-1}
\begin{align}
& \mathcal{X}_{1,x}+i \mathcal{X}_2+i|\mathcal{X}_2|^2\mathcal{X}_1=0, \label{MTM-1a}\\
& \mathcal{X}_{2,t}+i\mathcal{X}_1+i|\mathcal{X}_1|^2\mathcal{X}_2=0, \label{MTM-1b}
\end{align}
\end{subequations}
and
\begin{subequations}\label{MTM-1-Lax}
\begin{align}
& \Phi_{x}=M'\Phi,~~
M'=\frac{i}{2}\left(\begin{array}{cc}
         -\lambda^2+|\mathcal{X}_2|^2 & 2\lambda\mathcal{X}_2^*\\
         2\lambda \mathcal{X}_2     & \lambda^2 -|\mathcal{X}_2|^2
         \end{array}\right),\label{GI-sp}\\
& \Phi_{t}=N'\Phi,~~
N'=\frac{i}{2}\left(\begin{array}{cc}
         |\mathcal{X}_1|^2-\lambda^{-2} &  2\lambda^{-1}\mathcal{X}^*_1\\
         2\lambda^{-1}\mathcal{X}_1      & -|\mathcal{X}_1|^2+\lambda^{-2}
         \end{array}\right),
\end{align}
\end{subequations}
where \eqref{GI-sp} is known as the spectral problem of the
derivative Schr\"odinger equations of Chen-Lee-Liu's version \cite{WS-1983}.

Introducing
\[\Phi=\left(\begin{array}{cc}
         e^{\frac{i}{2}\beta}  &  0\\
         0      & e^{-\frac{i}{2}\beta}
         \end{array}\right)\tilde{\Phi},~~
q= \mathcal{X}^*_2 e^{-i\beta},~~ \beta=\int^{x}_{-\infty}|\mathcal{X}_2(y)|^2\mathrm{d}y,
\]
and noting that $\partial_x |\mathcal{X}_1|^2=-\partial_t |\mathcal{X}_2|^2$,
one can prove that the Lax pair \eqref{MTM-1-Lax} is gauge equivalent to \cite{KN-1977}
\begin{subequations}\label{MTM-2-Lax}
\begin{align}
& \tilde{\Phi}_{x}=M\tilde{\Phi},~~
M=\left(\begin{array}{cc}
         -\frac{i}{2} \lambda^2 &  i\lambda q \\
         i\lambda q^*    &\frac{i}{2} \lambda^2
         \end{array}\right),\label{KN-spr}\\
& \tilde{\Phi}_{t}=N \tilde{\Phi},~~
N= \left(\begin{array}{cc}
         i(|\mathcal{X}_1|^2-\frac{\lambda^{-2}}{2}) &  i\lambda^{-1}\mathcal{X}^*_1e^{-i\beta}\\
         i\lambda^{-1}\mathcal{X}_1e^{i\beta}      & -i(|\mathcal{X}_1|^2-\frac{\lambda^{-2}}{2})
         \end{array}\right),
\end{align}
\end{subequations}
where Eq.\eqref{KN-spr} coincides with the KN spectral problem \eqref{KN-sp} with $r=q^*$
and $\lambda \rightarrow i\lambda$.

The FL equation \eqref{FL2} (with $\delta=1$) is alternatively written as (all see Eq.(4.17) in \cite{GIK-1980})
\begin{equation}\label{FL3}
q_t+u-2i|u|^2 q=0,~~ u_x=q.
\end{equation}
By the transformation (cf.\cite{KN-1977,GIK-1980})
\begin{equation}\label{tran-M-MTM}
 \mathcal{X}_2 = q^* e^{-i\beta},~~ \mathcal{X}_1=-i u^*e^{-i\beta},~~
 \beta=\int^{x}_{-\infty}|q|^2\mathrm{d}y,
 ~~q=u_x,
\end{equation}
and noticing that $\partial_t|q|^2=-\partial_x|u|^2$,
the complex conjugate of Eq.\eqref{FL3} gives rise to Eq.\eqref{MTM-1a},
and Eq.\eqref{MTM-1b} holds automatically in light of \eqref{tran-M-MTM}.

With regard to the solutions between the FL equation and the massive Thirring model, making use of
relations Eq.\eqref{u-enve} and $\partial_t|q|^2=-\partial_x|u|^2$,  we have the following.
\begin{theorem}\label{Theorem-app}
If $u(x,t)$ is a solution to the FL equation \eqref{FL2} with $\delta =1$, then Eq.\eqref{tran-M-MTM} provides
solutions to the massive Thirring model \eqref{MTM-1} in light-cone coordinates.
In terms of $f$ and $g$, they are
\begin{equation}
\mathcal{X}_1=-i\frac{g^*}{f},~~ \mathcal{X}_2=\Bigl(\frac{g^*}{f^*}\Bigr)_x\frac{f^*}{f}.
\end{equation}

\end{theorem}

\section{$N$-soliton solution in Hirota's form}\label{app-A}

Employing the standard procedure of Hirota's method, one can derive 1-,2-,3-soliton solutions
for \eqref{bilinear-form}, which obey the following general form
\begin{subequations}\label{fghs1}
 \begin{eqnarray}
 && g_{N}(x, t)= \sum_{\mu=0,1}A_{2}(\mu)\exp\left\{\sum_{j=1}^{2N}\mu_{j}\zeta_{j}^{'}
  +\sum_{1\leq j<s}^{2N}\mu_{j}\mu_{s}\vartheta_{j,s}\right\}, \label{1-solution-of-dnls-g1}\\
 && f_{N}(x, t)=\sum_{\mu=0,1}A_{1}(\mu)\exp\left\{\sum_{j=1}^{2N}\mu_{j}\zeta_{j}^{''}
 +\sum_{1\leq j<s}^{2N}\mu_{j}\mu_{s}\vartheta_{j,s}\right\}, \label{1-solution-of-dnls-f1}  \\
 && h_{N}(x, t)= \sum_{\mu=0,1}A_{3}(\mu)\exp\left\{\sum_{j=1}^{2N}\mu_{j}\eta_{j}^{'}
 +\sum_{1\leq j<s}^{2N}\mu_{j}\mu_{s}\vartheta_{j,s}\right\},\label{1-solution-of-dnls-h1}\\
 && s_{N}(x, t)=\sum_{\mu=0,1}A_{1}(\mu)\exp\left\{\sum_{j=1}^{2N}\mu_{j}\eta_{j}^{''}
 +\sum_{1\leq j<s}^{2N}\mu_{j}\mu_{s}\vartheta_{j,s}\right\},
 \end{eqnarray}
\end{subequations}
where for $j,s=1,2,\cdots, N$,
\begin{eqnarray*}
 && \zeta_{j}=k_{j}x-\frac{1}{k_{j}}t+\zeta_{j}^{(0)},\quad w_{j}=\frac{1}{k_{j}},~~
 \eta_{j}=-l_{j}x+\frac{1}{l_{j}}t+\eta_{j}^{(0)},\quad m_{j}=-\frac{1}{l_{j}},\\
 && \zeta_{j}^{'}=\zeta_{j},\quad \zeta_{N+j}^{'}=\eta_{j}+\ln l_{j}+\frac{\pi}{2}i,~~
 \zeta_{j}^{''}=\zeta_{j}+\ln k_{j}+\frac{\pi}{2}i,\quad \zeta_{N+j}^{''}=\eta_{j}, \\
 &&\eta_{j}^{'}=\zeta_{j}+\ln k_{j}+\frac{\pi}{2}i,\quad \eta_{N+j}^{'}=\eta_{j},~~
 \eta_{j}^{''}=\eta_{j}+\ln l_{j}+\frac{\pi}{2}i,\quad\eta_{N+j}^{''}=\zeta_{j},\\
 && e^{\vartheta_{j,N+s}}=\frac{1}{(k_{j}-l_{s})(w_{j}+m_{s})},\quad (j,s=1,2,\ldots,N),\\
 && e^{\vartheta_{j,s}}=(k_{j}-k_{s})(w_{j}-w_{s}),\quad (j<s=2,3,\ldots,N),\\
 && e^{\vartheta_{N+j,N+s}}=-(l_{j}-l_{s})(m_{j}-m_{s}),\quad (j<s=2,3,\ldots, N),
 \end{eqnarray*}
$k_{j},l_{j},\zeta_{j}^{(0)},\eta_{j}^{(0)}\in \mathbb{C}$,
and $A_{1}(\mu)$, $A_{2}(\mu)$ and $A_{3}(\mu)$ take over all possible combinations of
$\mu_{j}=0,1$ $(j=1,2,\ldots,2N)$
and meanwhile satisfy the constraints
$\sum_{j=1}^{N}\mu_{j}=\sum_{j=1}^{N}\mu_{N+j}$,
$\sum_{j=1}^{N}\mu_{j}=1+\sum_{j=1}^{N}\mu_{N+j}$ and
$1+\sum_{j=1}^{N}\mu_{j}=\sum_{j=1}^{N}\mu_{N+j}$
respectively.

Consider reduction
\begin{equation}\label{reduc}
 l_{j}=-k_{j}^{*},~ m_{j}=w_{j}^{*},~\eta_{j}^{(0)}=\zeta_{j}^{(0)*},
\end{equation}
which indicates
$\zeta_{j}=\eta_{j}^{*}$, $e^{\vartheta_{j,(N+s)}*}=e^{\vartheta_{s(N+j)}}$
and $e^{\vartheta_{j,s}*}=e^{\vartheta_{(N+j),(N+s)}}$ and further
$s=f^{*}$, $h=g^{*}$.
Thus, \eqref{reduc} reduces the bilinear pKN$(-1)$ \eqref{bilinear-form} to a bilinear
equation for the bilinear FL equation \eqref{bil-FL} with $\delta=1$,
and its solution is given by \eqref{1-solution-of-dnls-g1} and \eqref{1-solution-of-dnls-f1} with
$(j,s=1,2,\ldots,N)$,
\begin{eqnarray*}
 && \zeta_{j}=k_{j}x-\frac{1}{k_{j}}t+\zeta_{j}^{(0)},\quad w_{j}=\frac{1}{k_{j}},\\
 && \zeta_{j}^{'}=\zeta_{j},\quad \zeta_{N+j}^{'}=\zeta_{j}^{*}+\ln (-k_{j}^{*})+\frac{\pi}{2}i,~~
 \zeta_{j}^{''}=\zeta_{j}+\ln k_{j}+\frac{\pi}{2}i,\quad \zeta_{N+j}^{''}=\zeta_{j}^{*}, \\
 && e^{\vartheta_{j,N+s}}=\frac{1}{(k_{j}+k_{s}^{*})(w_{j}+w_{s}^{*})},\quad (j,s=1,2,\ldots,N),\\
 && e^{\vartheta_{j,s}}=(k_{j}-k_{s})(w_{j}-w_{s}),\quad (j<s=2,3,\ldots,N).
 \end{eqnarray*}

\section{Proof of Theorem \ref{Theorem 1}}\label{app-B}

From the condition \eqref{wron-cond-x} one can calculate derivations of $f$, $g$, $h$ and $s$:
\begin{align*}
f_x=& |\t{N-1}, N+1; \W{M-1}|+ |\t{N}; \W{M-2},M|,\nonumber \\
f_t=&-\frac{1}{4}(|0,\b{N}; \W{M-1}|+|\t{N}; -1,\t{M-1}|), \\
f_{xt}=&-\frac{1}{4}(2|\t{N}; \W{M-1}|+|0,\b{N-1},N+1; \W{M-1}|+|0,\b{N}; \W{M-2},M| \\
& +|\t{N-1},N+1; -1,\t{M-1}|+|\t{N}; -1,\t{M-2},M|),
\end{align*}
\begin{align*}
g_x =&|\W{N-1},N+1; \t{M-1}|+|\W{N}; \t{M-2},M|,\\
g_t =&-\frac{1}{4}(|-1,\t{N}; \t{M-1}|+|\W{N}; 0,\b{M-1}|),\\
g_{xt}=&-\frac{1}{4}(2|\W{N}; \t{M-1}|+|-1,\t{N-1},N+1; \t{M-1}|+|-1,\t{N}; \t{M-2},M|\\
&+|\W{N-1},N+1; 0,\b{M-1}|+|\W{N}; 0,\b{M-2},M|).
\end{align*}
Substituting them into equation \eqref{bilinear-a}, the left-hand side gives rise to
\begin{equation}\label{eqd}
\begin{array}{rl}
~&D_xD_tg\cdot f+gf\\
=&g_{xt}f-g_xf_t-g_tf_x+gf_{xt}+gf \vspace{1ex}\\
=&-\frac{1}{4}|\t{N}; \W{M-1}|(2|\W{N}; \t{M-1}|+|-1,\t{N-1},N+1; \t{M-1}|+|-1,\t{N}; \t{M-2},M|
\vspace{1ex}\\
~&+|\W{N-1},N+1; 0,\b{M-1}|+|\W{N}; 0,\b{M-2},M|)+\frac{1}{4}(|0,\b{N}; \W{M-1}|
+|\t{N}; -1,\t{M-1}|)\vspace{1ex}\\
~&(|\W{N-1},N+1; \t{M-1}|\!+\!|\W{N}; \t{M-2},M|)\!+\!\frac{1}{4}(|\t{N-1}, N+1; \W{M-1}|
\!+\!|\t{N}; \W{M-2},M|)\vspace{1ex}\\
~&(|\!-\!1,\t{N}; \t{M\!-\!1}|+|\W{N}; 0,\b{M\!-\!1}|)\!-\!\frac{1}{4}|\W{N}; \t{M\!-\!1}|
(2|\t{N}; \W{M\!-\!1}|+|0,\b{N\!-\!1},N+1; \W{M\!-\!1}|\vspace{1ex}\\
~&+|0,\b{N}; \W{M\!-\!2},M|\!+\!|\t{N\!-\!1},N+1;\!-\!1,\t{M\!-\!1}|\!+\!|\t{N}; \!-\!1,\t{M\!-\!2},M|)
\!+\!|\W{N}; \t{M\!-\!1}||\t{N}; \W{M\!-\!1}|.\vspace{1ex}
\end{array}
\end{equation}
To simplify the right hand side, we making use of Lemma \ref{lemma 2}.
Consider $\Xi=|\t{N}; \W{M-1}|$ and $\gamma_{ij}=\partial_x^{-1}$ for $j=1,2,\cdots,N$
and $\gamma_{ij}=-\partial_x^{-1}$ for $j=N+1,N+2,\cdots,N+M$.
Using  Lemma \ref{lemma 2} and relation \eqref{wron-cond-x} we have
\[-2i \,\mathrm{Tr} (A^{-2}) |\t{N}; \W{M-1}| =|0,\b{N}; \W{M-1}|-|\t{N}; -1,\t{M-1}|,\]
where Tr$(A)$ stands for the trace of matrix $A$.
In a similar way, we have
\begin{align*}
&-2i \,\mathrm{Tr} (A^{-2}) |\W{N}; \t{M-1}|
=|-1,\t{N}; \t{M-1}|-|\W{N}; 0,\b{M-1}|,\\
&-2i \,\mathrm{Tr} (A^{-2}) |\t{N-1},N+1; \W{M-1}|
=|0,\b{N-1},N+1; \W{M-1}|+|\t{N}; \W{M-1}|\\
&~~~~~~~~~~~~~~~~~~~~~~~~~~~~~~~~~~~~~~~~~~~~~~~~~~ -|\t{N-1},N+1; -1,\t{M-1}|,\\
&-2i \,\mathrm{Tr} (A^{-2}) |\t{N}; \W{M-2},M|=|0,\b{N}; \W{M-2},M|-|\t{N}; \W{M-1}|
-|\t{N}; -1,\t{M-2},M|,\\
&-2i \,\mathrm{Tr} (A^{-2}) |\W{N-1},N+1; \t{M-1}|=|-1,\t{N-1},N+1; \t{M-1}|
+|\W{N}; \t{M-1}|\\
&~~~~~~~~~~~~~~~~~~~~~~~~~~~~~~~~~~~~~~~~~~~~~~~~~~-|\W{N-1}, N+1;0,\b{M-1}|,\\
&-2i \,\mathrm{Tr} (A^{-2}) |\W{N}; \t{M-2},M|=|-1,\t{N}; \t{M-2},M|-|\W{N}; \t{M-1}|
-|\W{N}; 0,\b{M-2},M|.
\end{align*}
From these relations we have
\begin{equation*}
\begin{array}{rl}
~& |\t{N}; \W{M-1}|(|-1,\t{N-1},N+1; \t{M-1}|+|-1,\t{N}; \t{M-2},M|-|\W{N}; 0,\b{M-2},M|\vspace{1ex}\\
~&-|\W{N-1},N+1; 0,\b{M-1}|)\vspace{1ex}\\
=&(|\W{N-1},N+1; \t{M-1}|+|\W{N}; \t{M-2},M|)(|0,\b{N}; \W{M-1}|-|\t{N}; -1,\t{M-1}|),\vspace{1ex}\\
~&|\W{N}; \t{M-1}|(|0,\b{N-1},N+1; \W{M-1}|+|0,\b{N}; \W{M-2},M|-|\t{N-1},N+1; -1,\t{M-1}|\vspace{1ex}\\
~&-|\t{N}; -1,\t{M-2},M|)\vspace{1ex}\\
=&(|\t{N-1}, N+1; \W{M-1}|+|\t{N}; \W{M-2},M|)(|-1,\t{N}; \t{M-1}|+|\W{N}; 0,\b{M-1}|),
\end{array}
\end{equation*}
by which we can reduce \eqref{eqd} to
\begin{equation*}
\begin{array}{rl}
~&g_{xt}f-g_xf_t-g_tf_x+gf_{xt}+gf \vspace{1ex}\\
=&-\frac{1}{2}|\t{N}; \W{M-1}|(|\W{N-1},N+1; 0,\b{M-1}|+|\W{N}; 0,\b{M-2},M|)\vspace{1ex}\\
&+\frac{1}{2}|\t{N}; -1,\t{M-1}|(|\W{N-1},N+1; \t{M-1}|+|\W{N}; \t{M-2},M|)\vspace{1ex}\\
&+\frac{1}{2}|\W{N}; 0,\b{M-1}|(|\t{N-1}, N+1; \W{M-1}|+|\t{N}; \W{M-2},M|)\vspace{1ex}\\
&-\frac{1}{2}|\W{N}; \t{M-1}|(|\t{N-1},N+1; -1,\t{M-1}|+|\t{N}; -1,\t{M-2},M|),
\end{array}
\end{equation*}
in which some terms can vanish by using   Lemma \ref{lemma 1} and we then come to
\begin{equation}\label{eqd2}
\begin{array}{rl}
~&g_{xt}f-g_xf_t-g_tf_x+gf_{xt}+gf \vspace{1ex}\\
 = & \frac{1}{2}(-|\W{N-1}; \W{M-1}||\t{N+1}; 0,\b{M-1}|+|\W{N}; \W{M-2}||\t{N}; 0,\b{M}|\vspace{1ex}\\
~&+|\t{N+1}; \t{M-1}||\W{N-1}; -1,\t{M-1}|-|\t{N}; \t{M}||\W{N}; -1,\t{M-2}|).
\end{array}
\end{equation}
To show the right hand side being zero, let us employ relation \eqref{wron-cond-x} to rewrite
some double Wronskians as
\begin{align*}
&|\W{N-1}; \W{M-1}|= (-1)^{M}(-2i)^{N+M}|A|^{-2}|\t{N}; \t{M}|,\\
&|\W{N}; \W{M-2}|= (-1)^{M-1}(-2i)^{N+M}|A|^{-2}|\t{N+1}; \t{M-1}|,\\
&|\W{N-1}; -1,\t{M-1}|=(-1)^{M}(-2i)^{N+M}|A|^{-2}|\t{N}; 0,\b{M}|,\\
&|\W{N}; -1,\t{M-2}|=(-1)^{M-1}(-2i)^{N+M}|A|^{-2}|\t{N+1}; 0,\b{M-1}|.
\end{align*}
Substituting them into \eqref{eqd2} we immediately find the right hand side vanished.
Thus, we have completed the proof for Eq.\eqref{bilinear-a}.
Eqs.\eqref{bilinear-b},  \eqref{bilinear-c} and  \eqref{bilinear-d} can be proved similarly.

\end{appendix}

\end{document}